\documentclass[11pt,reqno]{amsart}
\usepackage{amsmath,amssymb, mathrsfs, stmaryrd}
\usepackage{graphicx}
\usepackage{slashed}
\newcommand{\RN}[1]{%
  \textup{\uppercase\expandafter{\romannumeral#1}}%
}

\newcommand{\R}{\mathbb{R}}
\newcommand{\Div}{\textup{div}}
\newcommand{\nablas}{\slashed{\nabla}}
\newtheorem{theorem}{Theorem}
\newtheorem{proposition}[theorem]{Proposition}

\newtheorem{lemma}[theorem]{Lemma}

\theoremstyle{definition}
\newtheorem{definition}[theorem]{Definition}

\begin{document} 

\title[Stability of Higher Dimensional Schwarzschild]{Linear Stability of Higher Dimensional Schwarzschild Spacetimes: Decay of Master Quantities}

\author{Pei-Ken Hung}
\address{Pei-Ken Hung\\
Department of Mathematics\\
Massachusetts Institute of Technology, USA}
\email{pkhung@mit.edu}

\author{Jordan Keller}
\address{Jordan Keller\\
Black Hole Initiative\\
Harvard University, USA}
\email{jordan\_keller@fas.harvard.edu}

\author{Mu-Tao Wang}
\address{Mu-Tao Wang\\
Department of Mathematics\\
Columbia University, USA}
\email{mtwang@math.columbia.edu}

\thanks{Portions of this work were carried out during visits to the Department of Mathematics and The Institute of Mathematical Sciences at the Chinese University of Hong Kong.  The authors wish to thank these institutions for their hospitality and support.  The second author wishes to thank the John Templeton Foundation for its support. This material is based upon work supported by the National Science Foundation under Grant No. DMS-1810856 (Mu-Tao Wang).}

\begin{abstract}
In this paper, we study solutions to the linearized vacuum Einstein equations centered at higher-dimensional Schwarzschild metrics.  We employ Hodge decomposition to split solutions into scalar, co-vector, and two-tensor pieces; the first two portions respectively correspond to the closed and co-closed, or polar and axial, solutions in the case of four spacetime dimensions, while the two-tensor portion is a new feature in the higher-dimensional setting.  Rephrasing earlier work of Kodama-Ishibashi-Seto in the language of our Hodge decomposition, we produce decoupled gauge-invariant master quantities satisfying Regge-Wheeler type wave equations in each of the three portions.  The scalar and co-vector quantities respectively generalize the Moncrief-Zerilli and Regge-Wheeler quantities found in the setting of four spacetime dimensions; beyond these quantities, we further discover a higher-dimensional analog of the Cunningham-Moncrief-Price quantity in the co-vector portion.  In the analysis of the master quantities, we strengthen the mode stability result of Kodama-Ishibashi to a uniform boundedness estimate in all dimensions; further, we prove decay estimates in the case of six or fewer spacetime dimensions.  Finally, we provide a rigorous argument that linearized solutions of low angular frequency are decomposable as a sum of pure gauge solution and linearized Myers-Perry solution, the latter solutions generalizing the linearized Kerr solutions in four spacetime dimensions.
\end{abstract}
\maketitle

%%%%%%%%%%%%%%%%%%%%%%%%%%%%%%%%%%%%%%%%%%%%%%%%
\section{Introduction}

The Schwarzschild-Tangherlini black holes are higher-dimensional generalizations of the Schwarzschild spacetimes, comprising a static, spherically symmetric family of black hole solutions to higher-dimensional vacuum gravity:
\begin{equation}\label{vacuumEinstein}
Ric (g) = 0.
\end{equation}

Non-linear stability of the Schwarzschild-Tangherlini black holes as solutions of \eqref{vacuumEinstein} is a matter of considerable mathematical interest, owing to the developments in geometric analysis necessary in the problem's resolution.  Such work would add to non-linear stability results in four spacetime dimensions, in particular that of Christodoulou-Klainerman \cite{ChristKlainerman} for the Minkowski spacetime, in addition to the more recent non-linear stability results of Hintz-Vasy \cite{HintzVasy} for the slowly rotating Kerr-de Sitter spacetimes and Klainerman-Szeftel \cite{KlainermanSzeftel} for the Schwarzschild spacetime subject to polarized axisymmetric perturbations.

In this paper, we consider the simpler matter of linear stability of the Schwarzschild-Tangherlini solutions, concerning solutions $\delta g$ of the linearization of the vacuum Einstein equations about a member of the Schwarzschild-Tangherlini family $(\mathcal{M},g_{M})$, with mass $M > 0$:
\begin{equation}\label{linEinstein}
\delta Ric_{g_{M}}(\delta g) = 0.
\end{equation}
Owing to diffeomorphism-invariance of the Einstein equation, infinitesimal deformations of the background spacetime via smooth co-vector fields $X$
\begin{equation}
\pi_{X} := \mathcal{L}_{X}g_{M},
\end{equation}
referred to as pure gauge solutions, are solutions to the linearized equation \eqref{linEinstein}.  Moreover, the Schwarzschild-Tangherlini family is contained within the larger family of Myers-Perry solutions, yielding solutions to \eqref{linEinstein} corresponding to infinitesimal changes in mass and angular velocity.  To demonstrate linear stability it suffices to show that, with a choice of well-posed pure gauge solution $\pi_{X}$, the normalized solution
\[ \widehat{\delta g} = \delta g - \pi_{X} \]
decays through a suitable foliation to a Myers-Perry perturbation under appropriate initial conditions. 

In both the physics and mathematics literature, the identification and analysis of gauge-invariant quantities satisfying decoupled wave equations forms the basis of linear stability.  Building upon our earlier work \cite{HKW} in four spacetime dimensions, we utilize the spherical symmetry of the background Schwarzschild-Tangherlini spacetimes to split linearized solutions into scalar, co-vector, and two-tensor portions in a spacetime Hodge decomposition.  Identification of gauge-invariant master quantities satisfying decoupled Regge-Wheeler type wave equations for each of the three portions appears in the work of Kodama-Ishibashi-Seto \cite{KI1} and Kodama-Ishibashi \cite{KI2}, with the scalar and co-vector quantities respectively generalizing the Moncrief-Zerilli \cite{Moncrief, Zerilli} and Regge-Wheeler \cite{RW} quantities in four spacetime dimensions.  Beyond recasting the quantities of Kodama-Ishibashi-Seto, we further identify a higher-dimensional analog of the Cunningham-Moncrief-Price quantity \cite{CMP} in the co-vector portion.  Our work is the first of numerous recent results on linear stability \cite{HKW, DHR, Ma, ABW, Johnson, Hung} to consider higher-dimensional gravity.  We remark that a generalization of the four dimensional approach of Dafermos-Holzegel-Rodnianski \cite{DHR}, involving decoupled Weyl curvature components satisfying the Bardeen-Press equation \cite{BardeenPress}, could provide another avenue towards higher-dimensional linear stability.

The analysis of the Regge-Wheeler type equations (\ref{RW1},\ \ref{RW3},\ \ref{spinraised},\ \ref{RW4}) satisfied by the master quantities is informed by the study of the scalar wave equation, regarding as a ``poor man's'' linearization of the vacuum Einstein equations.  We draw upon the pioneering efforts and later refinements of many authors in four spacetime dimensions for the Schwarzschild and Kerr spacetimes \cite{KayWald, BlueSoffer2, DR, DR2, Luk, Tataru1, FSKY, Tataru2, AnderssonBlue, DRR}, in addition to the higher-dimensional generalization of Schlue \cite{Schlue}, utilizing the red-shift, Morawetz, and $r^p$ estimates described in these works.  Owing to the non-positivity of their potentials, the equations we consider are more challenging to analyze than the standard wave equation.  We overcome these difficulties with respect to uniform boundedness, using Hardy estimates in the spirit of Kodama-Ishibashi \cite{KI3} to strengthen their mode stability result to a uniform boundedness estimate in all dimensions.  In addition, we prove uniform decay estimates in the case of six and fewer spacetime dimensions.

Spherical symmetry of the background Schwarzschild-Tangherlini spacetime allows for an additional decomposition of the metric perturbation into tensor spherical harmonics.  In particular, we decompose a linearized solution into portions of lower and higher angular frequency:
\begin{equation}
\delta g = \delta g^{\ell < 2} + \delta g^{\ell \geq 2},
\end{equation}
per Proposition \ref{mode_decomp}.  The master quantities discussed above are central to controlling the higher angular frequency portion $\delta g^{\ell \geq 2}$, but have no control over the lower angular frequency portion $\delta {g}^{\ell < 2}$.  Generalizing the situation in four spacetime dimensions, we prove that $\delta g^{\ell < 2}$ splits as the sum of a pure gauge solution and a linearized Myers-Perry solution.  Our proof makes rigorous the same claim in Kodama-Ishibashi \cite{KI2}, which the authors base upon an enumeration of degrees of freedom.

We summarize our results in the two main theorems below.  First, we have the analysis of the lower angular frequency portion $\delta g^{\ell < 2}$:

\begin{theorem}
Let $\delta g$ be a smooth, symmetric two-tensor on a Schwarzschild-Tangherlini spacetime, satisfying the linearized Einstein equation \eqref{linEinstein}.  For the lower frequency portion $\delta g^{\ell <2}$ of $\delta g$, there exists a smooth co-vector $X^{\ell < 2}$ on the Schwarzschild-Tangherlini background, unique modulo Killing fields, and constants $c, d_{m}$ such that
\begin{equation}
\delta g^{\ell <2} = \pi_{X^{\ell<2}} + cK + \sum_{m =1}^{\frac{1}{2}n(n+1)}d_{m}K_{m},
\end{equation}
where $K, K_{m}$ are the basis solutions of the linearized Myers-Perry family in Definition \ref{linMPBasis}.
\end{theorem}

Next, we have the analysis of the gauge-invariant master quantities for the higher angular frequency portion $\delta g^{\ell \geq 2}$:

\begin{theorem}
Making the same assumption on $\delta g$ as in Theorem 1, there exist gauge-invariant master quantities satisfying decoupled Regge-Wheeler type equations (\ref{RW1},\ \ref{RW3},\ \ref{spinraised},\ \ref{RW4}) for the scalar, co-vector, and two-tensor portions of $\delta g^{\ell \geq 2}$.  With further specification of an initial data slice $\Sigma_{0}$ and a decay foliation $\Sigma_{\tau} := \phi_{\tau}(\Sigma_{0})$ formed by flowing along the static Killing vector field, a symmetric traceless two-tensor $\Psi$ solving any one of the equations (\ref{RW1},\ \ref{RW3},\ \ref{spinraised},\ \ref{RW4}) satisfies the uniform boundedness estimate
\begin{equation}
\check{E}^{N}_{\Psi}(\Sigma_{\tau}) \leq C(n,M)\check{E}^{N}_{\Psi}(\Sigma_{0}),
\end{equation}
and, in spacetime dimension six and below, the uniform decay estimate
\begin{equation}
\check{E}^{N}_{\Psi}(\Sigma_{\tau}) \leq C(n,M)\frac{I_{\Psi}(\Sigma_0)}{\tau^{2}},
\end{equation}
with $\check{E}^{N}_{\Psi}$ and $I_{\Psi}(\Sigma_0)$ representing Sobolev data for $\Psi$ on members of the decay foliation $\Sigma_{\tau}$ and $C(n,M)$ being a universal constant depending upon the orbit sphere dimension $n$ and the background mass $M$.   
\end{theorem}

We remark that further pointwise uniform boundedness and uniform decay estimates can be derived from those above by means of commutation with the angular Killing fields and application of Sobolev estimates on the orbit spheres.

It is expected that, after estimating the master quantities and making a suitable choice of linearized gauge, the remaining linearized metric components can be controlled by the gauge-invariant master quantities to ensure their uniform boundedness and decay, yielding a complete proof of linear stability.  We do not treat this matter in the current paper, deferring it to later work.  We remark that such efforts have borne fruit in the case of four spacetime dimensions.  In particular, our earlier work \cite{HKW} controls the linearized metric via a rather irregular combination of the Regge-Wheeler and Chandrasekhar gauges, while Johnson \cite{Johnson} controls the linearized metric uniformly in the Regge-Wheeler gauge, after an intermediate passage through the wave-coordinate gauge.  Inasmuch as these results depend upon corresponding gauge choices, each approach has drawbacks in extending to the non-linear regime.  In this direction, there is more promising work of the first author \cite{Hung}, wherein control of a portion of the linearized metric in the wave-map gauge is accomplished.

The paper is organized as follows.  In Section 2, we present the Schwarzschild-Tangherlini black holes.  In Section 3, we discuss the more general class of spherically symmetric spacetimes; in particular, we present Hodge decomposition and tensor spherical harmonic decomposition.  In Section 4 we discuss gravitational perturbations of spherically symmetric spacetimes, identifying pure gauge solutions and linearized Myers-Perry solutions in Section 5.  We prove Theorem 1, decomposing $\delta g^{\ell < 2}$ as a sum of a pure gauge and linearized Myers-Perry solution, in Section 6.  In Section 7, we discuss general estimates for Regge-Wheeler type equations sufficient to prove uniform boundedness and uniform decay.  In Section 8, we identify and analyze the master quantity for the two-tensor portion, proving uniform boundedness and decay in all spacetime dimensions.  Similar analyses for the co-vector and scalar portions are carried out in Sections 9 and 10, respectively; in each case, we prove uniform boundedness estimates in all spacetime dimensions and uniform decay estimates in spacetime dimension six and fewer.  We summarize our results on the master quantities of $\delta g^{\ell \geq 2}$ in Section 11, wherein we prove Theorem 2.

%%%%%%%%%%%%%%%%%%%%%%%%%%%%%%%%%%%%%%%%%%%%%%%%
\section{Higher Dimensional Schwarzschild Spacetimes}

The higher dimensional Schwarzschild-Tangherlini black holes $(\mathcal{M}^{2+n},g_{M})$ generalize the well-known four-dimensional spacetimes, with the $(2 + n)$-dimensional family comprised of static, spherically-symmetric (i.e., $SO(n+1)$-invariant) members, parametrized by mass $M > 0$.  Each such member is a solution of vacuum gravity; i.e., each metric $g_{M}$ satisfies $Ric(g_{M}) = 0.$

In standard Schwarzschild coordinates $(t,r, x^{\alpha})$, $x^{\alpha}$ coordinates on $S^n$, the Schwarzschild metric takes the form
\begin{equation}
g_{M} = -(1-\mu)dt^2 + (1-\mu)^{-1}dr^2 + r^2\mathring\sigma_{\alpha\beta}dx^{\alpha}dx^{\beta},
\end{equation}
with
\begin{equation}
\mu := \frac{2M}{r^{n-1}}
\end{equation}
and
\begin{equation}
\mathring\sigma_{\alpha\beta}dx^{\alpha}dx^{\beta}
\end{equation}
understood to be the standard round metric of the unit $n$-sphere.

Defining the Regge-Wheeler coordinate via
\begin{equation}\label{RWcoordinate}
r_{*} = \int^{r}\left(1-\frac{2M}{s^{n-1}}\right)^{-1}ds,
\end{equation}
we find
\begin{equation}
g_{M} = -(1-\mu)dt^2 + (1-\mu)dr_{*}^2 + r^2\mathring\sigma_{\alpha\beta}dx^{\alpha}dx^{\beta}.
\end{equation}

Using the Regge-Wheeler coordinates, we specify the Eddington-Finkelstein double-null coordinates by
\begin{align}
\begin{split}
u = \frac{1}{2}(t - r_{*}),\\
v = \frac{1}{2}(t + r_{*}),
\end{split}
\end{align}
such that
\begin{equation}
g_{M} = -4(1-\mu)dudv + r^2\mathring\sigma_{\alpha\beta}dx^{\alpha}dx^{\beta}.
\end{equation}

We remark that $r_{*}$ is defined up to normalization.  All three of the coordinate systems cover the exterior region of the spacetime and degenerate at the event horizon.

We also make use of the ingoing Eddington-Finkelstein coordinate system
\begin{align}
\begin{split}
\bar{v} &= t + r_{*},\\
R &= r
\end{split}
\end{align}
with
\begin{equation}
g_{M} = -(1-\mu)d\bar{v}^2 + 2d\bar{v}dR + R^2\mathring\sigma_{\alpha\beta}dx^{\alpha}dx^{\beta}.
\end{equation}

Finally, a variant of the Regge-Wheeler coordinates takes
\begin{equation}
t_{*} = t - r + r_{*},
\end{equation}
such that
\begin{equation}
g_{M} = -(1-\mu)dt_{*}^2 +2\mu dt_{*}dr + (1+\mu)dr^2 + r^2\mathring\sigma_{\alpha\beta}dx^{\alpha}dx^{\beta}.
\end{equation}
These two coordinate systems remain regular up to and on the future event horizon.

The event horizon appears in this coordinate system as the null hypersurface
\begin{equation}
r = r_{h} := (2M)^{1/(n-1)}.
\end{equation}
Along the event horizon, the Schwarzschild-Tangherlini solution has positive surface gravity
\begin{equation}
\kappa_{n} := \frac{(n-1)}{2r_{h}},
\end{equation}
in addition to simple trapping at the timelike hypersurface
\begin{equation}
r = r_{P} := ((n+1)M)^{1/(n-1)}.
\end{equation}
This hypersurface is referred to as the photon sphere.  With it, we normalize the Regge-Wheeler coordinate by
\begin{equation}\label{RWnormalization}
r_{*}(r_{P}) = 0.
\end{equation}

For a detailed discussion of these and other issues related to the geometry of higher dimensional Schwarzschild spacetimes, we refer the reader to Schlue \cite{Schlue}.

%%%%%%%%%%%%%%%%%%%%%%%%%%%%%%%%%%%%%
\section{Spherically Symmetric Spacetimes}

%%%%%%%%%%%%%%%%%%%%%%%%%%%%%%%%%%%%%
\subsection{General Considerations}
Let $(\mathcal{Q}, \tilde{g})$ be a two-dimensional Lorentzian manifold with local coordinates $x^A, A=0,1$, and let $(S^n, \mathring{\sigma})$ be the unit $n$-sphere with the standard round metric in local coordinates $x^\alpha, \alpha=2, \hdots, n+1$.  Each point on $\mathcal{Q}$ represents an orbit sphere, with $r$ a positive function which represents the areal radius of each orbit sphere. We consider a general spherically symmetric spacetime in local coordinates $x^0, x^1, x^2, \hdots x^{n+1}$:
\begin{equation}\label{spacetime metric} 
g_{ab}dx^{a}dx^{b} = \tilde{g}_{AB} dx^A dx^B+r^2 \mathring{\sigma}_{\alpha\beta} dx^\alpha dx^\beta.
\end{equation}
The index notations above are adopted throughout the paper: $A, B, C, \cdots =0, 1$ for quotient indices, $\alpha, \beta, \gamma, \cdots =2, \hdots, n+1$ for spherical indices, and $a, b, c, \cdots =0, 1, 2, \hdots, n+1$ for spacetime indices.  

The Christoffel symbols $\Gamma_{ab}^c$ of a spherically symmetric spacetime are 
\[\begin{split}\Gamma_{AB}^C & =\tilde{\Gamma}_{AB}^C,\\
\Gamma_{\alpha\beta}^\gamma&=\mathring{\Gamma}_{\alpha\beta}^\gamma,\\
\Gamma_{\alpha A}^\beta&=r^{-1} \partial_A r (\delta_\alpha^\beta),\\
\Gamma_{\alpha\beta}^D&=-r \partial^Dr (\mathring{\sigma}_{\alpha\beta}),\end{split}\]
where $\Gamma_{AB}^C$ and $\mathring{\Gamma}_{\alpha\beta}^\gamma$ are the Christoffel symbols of $\tilde{g}_{AB}$ and $\mathring{\sigma}_{\alpha\beta}$, respectively. 

Using the Christoffel symbols, it is possible to calculate the curvature of the quotient $\mathcal{Q}$ and the $n$-sphere $S^n$ directly.  On the other hand, as $\mathcal{Q}$ is a two-manifold, we have immediately
\begin{align}
\begin{split}
\tilde{R}_{ABCD} &= \tilde{K}\left(\tilde{g}_{AC}\tilde{g}_{BD} - \tilde{g}_{AD}\tilde{g}_{BC}\right),\\
\tilde{R}_{AB} &= \tilde{K}\tilde{g}_{AB},\\
\tilde{R} &= 2\tilde{K},
\end{split}
\end{align}
relating the Riemannian curvature tensor, the Ricci tensor, and the scalar curvature of the quotient to its sectional curvature $\tilde{K}$.  Likewise, as the $n$-sphere is a space form with constant sectional curvature $\mathring{K} = 1$, we find
\begin{align}
\begin{split}
\mathring{R}_{\alpha\beta\gamma\eta} &= \left(\mathring{\sigma}_{\alpha\gamma}\mathring{\sigma}_{\beta\eta} - \mathring{\sigma}_{\alpha\eta}\mathring{\sigma}_{\beta\gamma}\right),\\
\mathring{R}_{\alpha\beta} &= (n-1)\mathring{\sigma}_{\alpha\beta},\\
\mathring{R} &= n(n-1).
\end{split}
\end{align}

With respect to the $2+n$ decomposition into quotient and spherical parts, we consider two types of differential operators,  $\tilde{\nabla}_A$ and $\mathring{\nabla}_\alpha$. When applied to functions, $\tilde{\nabla}_A$  and $\mathring{\nabla}_\alpha$ are just differentiation with respect to coordinate variables $x^A, A=0, 1$ and $x^\alpha, \alpha=2, \hdots, n+1$, respectively. For co-vectors, we define
\begin{equation}\label{mathring}\begin{split} \tilde{\nabla}_A dx^B&=-\tilde{\Gamma}_{AC}^B dx^C,\\
\tilde{\nabla}_A dx^\alpha&=0,\\
\mathring{\nabla}_\alpha dx^B&=0,\\
\mathring{\nabla}_\alpha dx^\beta&=-\mathring{\Gamma}_{\alpha\gamma}^\beta dx^\gamma,
\end{split}
\end{equation}
with an obvious extension of the operators to more elaborate tensor bundles.

We use the notation $\tilde\Box$ and $\mathring\Delta$ for the quotient d'Alembertian and the spherical Laplacian operators.  Furthermore, we denote the volume form for the quotient spaceby $\epsilon_{AB}$. 

Later in the work, we will make use of the commutation identities
\begin{align}
\begin{split}
\nabla_{a}\nabla_{b} v_{c} - \nabla_{b}\nabla_{a} v_{c} &= {R_{abc}}^{d}v_{d},\\
\nabla_{a}\nabla_{b} v_{cd} - \nabla_{b}\nabla_{a} v_{cd} &= {R_{abc}}^{e}v_{ed} + {R_{abd}}^{e}v_{ce},
\end{split}
\end{align}
which apply to either the quotient or the orbit spheres.

Specializing to the Schwarzschild spacetime, we note the formulae
\begin{align}\label{SchwarzFormulae}
\begin{split}
&\tilde\nabla_{A}\tilde\nabla_{B} r = \frac{M(n-1)}{r^{n}}\tilde{g}_{AB} = \frac{(n-1)}{2r}\mu \tilde{g}_{AB},\\
&\tilde\nabla_{A}\tilde\nabla_{B} t = -(1-\mu)^{-1}\left(\tilde\Box r\right)t_{(A}r_{B)},\\
&r^{A}r_{A} = |\tilde\nabla r|^2 = 1- \frac{2M}{r^{n-1}} = 1- \mu,\\
&\tilde{K} = \frac{n(n-1)M}{r^{n+1}} = \frac{n(n-1)}{2r^2}\mu.
\end{split}
\end{align}

%%%%%%%%%%%%%%%%%%%%%%%%%%%%%%%%%%%%%
\subsection{Tensors on $S^n$}

Specializing to the $n$-sphere, with the curvature calculations and commutation relations above taken into account, we find
\begin{align}\label{angularCommutation}
\begin{split}
\mathring\nabla_{\alpha}\mathring\nabla_{\beta}v_{\gamma} - \mathring\nabla_{\beta}\mathring\nabla_{\alpha}v_{\gamma} &= \mathring{\sigma}_{\alpha\gamma}v_{\beta} - \mathring{\sigma}_{\beta\gamma}v_{\alpha},\\
\mathring\nabla_{\alpha}\mathring\nabla_{\beta} v_{\gamma\delta} - \mathring\nabla_{\beta}\mathring\nabla_{\alpha} v_{\gamma\delta} &= \mathring{\sigma}_{\alpha\gamma}v_{\beta\delta} - \mathring{\sigma}_{\beta\gamma}v_{\alpha\delta} + \mathring\sigma_{\alpha\delta}v_{\gamma\beta} - \mathring\sigma_{\beta\delta}v_{\gamma\alpha}.
\end{split}
\end{align}

%%%%%%%%%%%%%%%%%%%%%%%%%%%%%%%
\subsubsection{Tensor Spherical Harmonics}

In this subsection we outline tensor spherical harmonics on $S^n$, following closely the discussion in Chodos-Myers \cite{ChodosMyers}.

The scalar spherical harmonics $Y^{\ell m_{s}(n,\ell)}$ are eigenfunctions of the spherical Laplacian, satisfying
\begin{equation}\label{scalarSpectraOne}
\mathring\Delta Y^{\ell m_{s}(n,\ell)} = -\ell(\ell + n - 1) Y^{\ell m_{s}(n,\ell)},
\end{equation}
with indices $\ell \geq 0$ and $m_{s}(n,\ell) \in \{1, \hdots, d_{s}(n,\ell)\},$ where
\begin{equation}\label{scalarEigenspace}
d_s(n,\ell) = \binom{n+\ell}{\ell} - \binom{n+\ell -2}{\ell-2}.
\end{equation}
Note that the formula gives
\begin{align}\label{lowerScalars}
\begin{split}
d_s(n,0) &= 1,\\
d_s(n,1) &= n+1\\
d_s(n,2) &= \frac{1}{2}(n+2)(n+1)-1.
\end{split}
\end{align}

Using the scalar spherical harmonics, we obtain eigensections for the sub-bundles of co-vectors and symmetric traceless two-tensors given by scalar potentials.  Namely, we have eigensections
\begin{equation}
Y^{\ell m_{s}(n,\ell)}_{\alpha} := \mathring\nabla_{\alpha} Y^{\ell m_s(n,\ell)},
\end{equation}
satisfying
\begin{equation}\label{scalarSpectraTwo}
\mathring\Delta Y^{\ell m_s(n,\ell)}_{\alpha} = \left((n-1)-\ell(\ell+n-1)\right)Y^{\ell m_s(n,\ell)}_{\alpha},
\end{equation}
and eigensections 
\begin{equation}\label{scalarTensorHarmonics}
Y^{\ell m_s(n,\ell)}_{\alpha\beta} := \mathring\nabla_{\alpha}\mathring\nabla_{\beta} Y^{\ell m_s(n,\ell)} - \frac{1}{n}\mathring\sigma_{\alpha\beta}\mathring\Delta Y^{\ell m_s(n,\ell)},
\end{equation}
such that
\begin{equation}\label{scalarSpectraThree}
\mathring\Delta Y^{\ell m_s(n,\ell)}_{\alpha\beta} = \left(2n-\ell(\ell + n - 1)\right)Y^{\ell m_s(n,\ell)}_{\alpha\beta},
\end{equation}
with $\ell \geq 2$.

In addition, the spherical Laplacian acts as an endomorphism on the sub-bundles of divergence-free co-vectors and divergence-free symmetric traceless two-tensors.  Regarding such co-vectors, we have eigensections $X^{\ell m_v(n,\ell)}_{\alpha}$ satisfying
\begin{equation}\label{coVectorSpectra}
\mathring\Delta X^{\ell m_v(n,\ell)}_{\alpha} =\left(1-\ell(\ell+n-1)\right) X^{\ell m_v(n,\ell)}_{\alpha},
\end{equation}
for $\ell \geq 1$ and $m_{v}(n,\ell) \in \{1, \hdots, d_{v}(n,\ell)\},$ where
\begin{align}
\begin{split}
d_v(n,\ell) &= (n+1)d_s(n,\ell) - d_s(n,\ell+1) - d_s(n,\ell-1)\\
&=(n+1)\left(\binom{n+\ell}{\ell} - \binom{n+\ell -2}{\ell-2}\right) \\
&- \binom{n+\ell+1}{\ell+1}+\binom{n+\ell-3}{\ell-3}.
\end{split}
\end{align}
Note that the formula above together with \eqref{lowerScalars} gives
\begin{equation}\label{l1Covector}
d_v(n,1) = \frac{1}{2}n(n+1).
\end{equation}

For those symmetric traceless two-tensors given by divergence-free co-vector potentials, we have eigensections
\begin{equation}
X^{\ell m_v(n,\ell)}_{\alpha\beta} :=  \mathring\nabla_{\alpha}X^{\ell m_v(n,\ell)}_{\beta} + \mathring\nabla_{\beta}X^{\ell m_v(n,\ell)}_{\alpha},
\end{equation}
with
\begin{equation}\label{coVectorSpectraTwo}
\mathring\Delta X^{\ell m_v(n,\ell)}_{\alpha\beta} = \left( n + 2 - \ell(\ell + n - 1)\right)X^{\ell m_v(n,\ell)}_{\alpha\beta}
\end{equation}
for $\ell \geq 2.$  
  
On the sub-bundle of the symmetric traceless two-tensors, we find
\begin{equation}\label{twoTensorSpectra}
\mathring\Delta U^{\ell m_t(n,\ell)}_{\alpha\beta} = \left(2-\ell(\ell + n - 1)\right) U^{\ell m_t(n,\ell)}_{\alpha\beta},
\end{equation}
for eigensections $U^{\ell m_t(n,\ell)}_{\alpha\beta}$, with $\ell \geq 2$ and and $m_{t}(n,\ell) \in \{1, \hdots, d_{t}(n,\ell)\},$ where
\begin{align}
\begin{split}
d_t(n,\ell) &= \frac{1}{2}(n+1)(n+2)d_s(n,\ell) - \left(d_v(n,\ell+1)+d_v(n,\ell-1)\right)\\
&-d_s(n,\ell+2) - 2d_s(n,\ell) - d_s(n,\ell-2)\\
&= \frac{1}{2}(n+1)(n+2)\left(\binom{n+\ell}{\ell} - \binom{n+\ell-2}{\ell-2}\right)\\
&-(n+1)\left(\binom{n+\ell+1}{\ell+1}-\binom{n+\ell-3}{\ell-3}\right).
\end{split}
\end{align}

We remark that the $\ell = 1$ scalar co-vectors correspond to conformal Killing vectors,
\begin{equation}\label{conformalKilling}
\mathring\nabla_{\alpha}Y_{\beta}^{1m_s(n,1)} + \mathring\nabla_{\beta}Y^{1m_s(n,1)}_{\alpha} = -2Y^{1m_s(n,1)}\mathring\sigma_{\alpha\beta},
\end{equation}
and the $\ell = 1$ divergence-free co-vectors correspond to Killing vectors
\begin{equation}\label{Killing}
\mathring\nabla_{\alpha}X_{\beta}^{1m_v(n,1)} + \mathring\nabla_{\beta}X^{1m_v(n,1)}_{\alpha} = 0.
\end{equation}

Concretely, the conformal Killing vector fields can be realized by considering a Cartesian coordinate system $(X^1,\hdots, X^{n+1})$ on $\R^{n+1}$.  Denoting the restriction of the coordinate functions to the unit sphere $S^n$ by $\tilde{X}^{A}, A=1\cdots n+1$, the set of Killing vector fields $\{ X_{\alpha}^{1 m_v(n,1)}\}_{m_v(n,1) = 1,\hdots,d_v(n,1)}$ is given by
\begin{equation}\label{killingConcrete}
\tilde{X}^{A}\mathring\nabla_{\alpha}\tilde{X}^{B} - \tilde{X}^{B}\mathring\nabla_{\alpha}\tilde{X}^{A},
\end{equation}
where $1\leq A < B \leq n+1$, and the set of conformal Killing vector fields $\{ Y_{\alpha}^{1 m_s(n,1)}\}_{m_s(n,1)=1,\hdots,d_s(n,1)}$ is given by
\begin{equation}\label{conformalKillingConcrete}
\mathring\nabla_{\alpha}\tilde{X}^{A},
\end{equation}
with $1 \leq A \leq n+1$.

%%%%%%%%%%%%%%%%%%%%%%%%%%%%%%%%%%%%%
\subsubsection{Tensor Decomposition}

The following decomposition lemmae, generalizing the situation for the two-sphere, are fundamental to the remainder of the work.

\begin{lemma} 
Given any co-vector $v_\alpha$ on $S^n$, there exist a scalar function $V$ and a divergence-free co-vector $\hat{v}_\alpha$ (i.e. $\mathring{\nabla}^\alpha \hat{v}_{\alpha}=0$) such that 
\[v_\alpha=\mathring{\nabla}_\alpha V+\hat{v}_\alpha.\]
The decomposition is unique modulo the $\ell = 0$ mode of $V$.
\end{lemma}

\begin{proof}
Let $V$ be a solution of Poisson's equation
\begin{equation}
\mathring{\Delta} V = \mathring\nabla^{\alpha}v_{\alpha}.
\end{equation}
The difference $\hat{v}_{\alpha} := v_{\alpha} - \mathring{\nabla}_{\alpha}V$ is manifestly divergence-free, and we have
\[ v_{\alpha} = \mathring\nabla_{\alpha}V + \hat{v}_{\alpha}.\]
Assuming that $v_{\alpha}$ has such a decomposition, it must be that the scalar function satisfies Poisson's equation above, so that $V$ is determined up to its constant $\ell = 0$ mode and $\hat{v}_{\alpha}$ is uniquely determined.
\end{proof}

\begin{lemma} Given any symmetric traceless $(0, 2)$-tensor $t_{\alpha\beta}$ on $S^n$, there exists a co-vector $v_\alpha$ and a divergence-free symmetric traceless $(0, 2)$-tensor $\hat{t}_{\alpha\beta}$ such that 
\[t_{\alpha\beta}=\mathring{\nabla}_\alpha v_\beta+\mathring{\nabla}_\beta v_\alpha-\frac{2}{n} \mathring{\sigma}_{\alpha\beta} \mathring{\nabla}^\gamma v_\gamma+\hat{t}_{\alpha\beta}.\]
Combined with the previous lemma, we have 
\begin{equation}
t_{\alpha\beta} = \left(\mathring\nabla_{\alpha}\mathring\nabla_{\beta}V - \frac{1}{n}\mathring{\Delta}V\mathring\sigma_{\alpha\beta}\right) + \left(\mathring\nabla_{\alpha}\hat{v}_{\beta} + \mathring\nabla_{\beta} \hat{v}_{\alpha}\right) + \hat{t}_{\alpha\beta}.
\end{equation}

The decomposition is unique modulo the $\ell < 2$ modes of $V$ and the $\ell = 1$ mode of $\hat{v}_{\alpha}.$
\end{lemma}

\begin{proof}
The divergence of $t_{\alpha\beta}$ has co-vector decomposition
\[ \mathring\nabla^{\alpha}t_{\alpha\beta} =: d_{\beta} = \mathring\nabla_{\beta} D + \hat{d}_{\beta},\]
per the previous lemma.  Note that $d_{\beta}$ is supported in $\ell \geq 2$, as 
\begin{align*}
&\int_{S^n} \mathring\nabla^{\alpha}t_{\alpha}^{\beta}Y^{1m}_{\beta} = -\int_{S^n}t_{\alpha}^{\beta}\mathring\nabla^{\alpha}Y^{1m}_{\beta} = 2\int_{S^n}t_{\alpha\beta}\mathring\sigma^{\alpha\beta} = 0\\
& \int_{S^n}\mathring\nabla^{\alpha}t_{\alpha\beta}X_{\beta}^{1m} = -\int_{S^n}t_{\alpha}^{\beta}\mathring\nabla^{\alpha}X_{\beta}^{1m} = -\int_{S^n} t^{\alpha\beta}\mathring\nabla_{(\alpha}X_{\beta)}^{1m} = 0,
\end{align*}
using the conformal Killing and Killing equations (\ref{conformalKilling},\ \ref{Killing}) and the symmetry and tracelessness of $t_{\alpha\beta}$.

We solve for $V$ and $\hat{v}_{\alpha}$ such that
\begin{align*}
& \frac{n-1}{n}\mathring\Delta\left(\mathring\Delta + n\right)V = \mathring\Delta D = \mathring\nabla^{\alpha}\mathring\nabla^{\beta}t_{\alpha\beta},\\
&\left(\mathring\Delta + (n-1)\right)\hat{v}_{\beta} = \hat{d}_{\beta}.
\end{align*}
Each of the elliptic operators in the left-hand side is self-adjoint, with kernels being the $\ell < 2$ scalar modes and the $\ell = 1$ co-vector modes, respectively.  Owing to support of $d_{\beta}$ and its constituents in $\ell \geq 2$, we have perpendicularity of the right-hand side and the kernel, from which existence of solutions to the equations follows.

The quantity
\[\hat{t}_{\alpha\beta} := t_{\alpha\beta} - \left(\mathring\nabla_{\alpha}\mathring\nabla_{\beta}V - \frac{1}{n}\mathring{\Delta}V\mathring\sigma_{\alpha\beta}\right) - \left(\mathring\nabla_{\alpha}\hat{v}_{\beta} + \mathring\nabla_{\beta} \hat{v}_{\alpha}\right),\]
satisfies the divergence-free property owing to the choices of $V$ and $\hat{v}_{\alpha}$, and we have 
\[t_{\alpha\beta} = \left(\mathring\nabla_{\alpha}\mathring\nabla_{\beta}V - \frac{1}{n}\mathring{\Delta}V\mathring\sigma_{\alpha\beta}\right) + \left(\mathring\nabla_{\alpha}\hat{v}_{\beta} + \mathring\nabla_{\beta} \hat{v}_{\alpha}\right) + \hat{t}_{\alpha\beta}\]
by definition.  

Given such a decomposition, the constituents $V$ and $\hat{v}_{\alpha}$ necessarily satisfy the equations above; hence they are uniquely determined up to the $\ell < 2$ modes of $V$ and the $\ell = 1$ mode of $\hat{v}_{\alpha}.$
\end{proof}

In the first lemma, the divergence-free co-vector $\hat{v}_{\alpha}$ generalizes the co-closed potentials in our earlier work \cite{HKW}.

As outlined in the second lemma, the situation for symmetric traceless two-tensors is more interesting.  Subtracting off the piece involving the scalar potential $V$, the remainder of the tensor satisfies the double-divergence free property, analogous to the co-closed potential in our earlier work.  However, this remainder admits a further decomposition, amounting to the term $\hat{t}_{\alpha\beta}$ with the stronger divergence-free property.  The term $\hat{t}_{\alpha\beta}$ is a novel feature in this higher dimensional setting.

%%%%%%%%%%%%%%%%%%%%%%%%%%%%%%%%%%%%%%%
\subsection{The Projected Covariant Derivative}

In what follows, we consider quantities which are scalars, co-vectors, or symmetric traceless two-tensors on the spheres of symmetry.  The associated sphere bundles, respectively referred to as $\mathcal{L}(0), \mathcal{L}(-1),$ and $\mathcal{L}(-2)$, come equipped with projected covariant derivative operators $\slashed{\nabla}$, defined for scalars by ordinary differentiation and for co-vectors by
\[\slashed{\nabla}_a dx^\alpha =-\Gamma^{\alpha}_{a\gamma} dx^\gamma, \text{  for  } a=0, 1, 2, 3,\]
extending to symmetric traceless two-tensors via the product rule.  We denote the associated d'Alembertian operators by
\begin{equation}\label{SpindAlembertian}
\slashed{\Box}_{\mathcal{L}(-s)} := \slashed{\nabla}^a\slashed{\nabla}_a,
\end{equation}
with $s = 0, 1, 2$ and the appropriate covariant derivative operator.  Note that $\slashed{\Box}_{\mathcal{L}(0)} = \Box$ is the standard d'Alembertian operator on $\mathcal{M}$. 

The projected connection, as well as the associated d'Alembertian and Laplacian operators, are related to the quotient and spherical operators of the first subsection in a straightforward fashion.  We illustrate the procedure on the bundle $\mathcal{L}(-2)$:

\begin{align*}
\slashed{\nabla}_{A}t_{\alpha\beta} &= \partial_{A}t_{\alpha\beta} - \Gamma^{\gamma}_{A\alpha}t_{\gamma\beta} - \Gamma^{\gamma}_{A\beta}t_{\alpha\gamma}\\
&=\tilde\nabla_{A}t_{\alpha\beta} - 2r^{-1}r_{A}t_{\alpha\beta},
\end{align*}

\begin{align*}
\slashed\nabla_{B}\slashed\nabla_{A}t_{\alpha\beta} &= \partial_{B}\left(\slashed\nabla_{A}t_{\alpha\beta}\right) - \Gamma^{C}_{BA}\slashed\nabla_{C}t_{\alpha\beta}\\
&- \Gamma^{\gamma}_{BA}\slashed\nabla_{\gamma}t_{\alpha\beta} - \Gamma^{\gamma}_{B\alpha}\slashed\nabla_{A}t_{\gamma\beta} - \Gamma^{\gamma}_{B\beta}\slashed\nabla_{A}t_{\alpha\gamma}\\
&=\tilde\nabla_{B}\left(\slashed\nabla_{A}t_{\alpha\beta}\right) - 2r^{-1}r_{B}\left(\slashed\nabla_{A}t_{\alpha\beta}\right)\\
&=\tilde\nabla_{B}\tilde\nabla_{A}t_{\alpha\beta} -2r^{-1}r_{A}\tilde\nabla_{B}t_{\alpha\beta} - 2r^{-1}r_{B}\tilde\nabla_{A}t_{\alpha\beta}\\
&+ 6r^{-2}r_{A}r_{B}t_{\alpha\beta} - 2r^{-1}\left(\tilde\nabla_{A}\tilde\nabla_{B}r\right)t_{\alpha\beta},
\end{align*}

\[\slashed\nabla_{\gamma}t_{\alpha\beta} = \mathring\nabla_{\gamma}t_{\alpha\beta},\]

\begin{align*}
\slashed\nabla_{\lambda}\slashed\nabla_{\gamma}t_{\alpha\beta}&=\partial_{\lambda}\left(\slashed\nabla_{\gamma}t_{\alpha\beta}\right) - \Gamma^{\delta}_{\lambda\gamma}\slashed\nabla_{\delta}t_{\alpha\beta} \\
&- \Gamma^{\delta}_{\lambda\alpha}\slashed\nabla_{\gamma}t_{\delta\beta} - \Gamma^{\delta}_{\lambda\beta}\slashed\nabla_{\gamma}t_{\alpha\delta} - \Gamma^{A}_{\lambda\gamma}\slashed\nabla_{A}t_{\alpha\beta}\\
&= \mathring\nabla_{\lambda}\mathring\nabla_{\gamma}t_{\alpha\beta} + rr^{A}\mathring\sigma_{\lambda\gamma}\left(\slashed\nabla_{A}t_{\alpha\beta}\right)\\
&=\mathring\nabla_{\lambda}\mathring\nabla_{\gamma}t_{\alpha\beta}+rr^{A}\mathring\sigma_{\lambda\gamma}\left(\tilde\nabla_{A}t_{\alpha\beta} - 2r^{-1}r_{A}t_{\alpha\beta}\right).
\end{align*}

Contracting the above, we deduce the relation

\begin{align}\label{TwoTensorWave}
\begin{split}
\slashed{\Box}_{\mathcal{L}(-2)} t_{\alpha\beta} &= \tilde\Box t_{\alpha\beta} + (n-4)r^{-1}r^{A}\tilde\nabla_{A}t_{\alpha\beta} + r^{-2}\mathring\Delta t_{\alpha\beta} \\
&+ (6-2n)r^{-2}r^{A}r_{A}t_{\alpha\beta} - 2r^{-1}\left(\tilde\Box r\right) t_{\alpha\beta}.
\end{split}
\end{align}

Likewise, we calculate
\begin{align}\label{CovectorWave}
\begin{split}
\slashed{\Box}_{\mathcal{L}(-1)} v_{\alpha} &= \tilde\Box v_{\alpha} + (n-2)r^{-1}r^{A}\tilde\nabla_{A} v_{\alpha} + r^{-2}\mathring\Delta v_{\alpha} \\
&+(2-n)r^{-2}r^{A}r_{A} v_{\alpha} - r^{-1}\left(\tilde\Box r\right) v_{\alpha},
\end{split}
\end{align}

\begin{equation}\label{ScalarWave}
\slashed{\Box}_{\mathcal{L}(0)} V = \Box V = \tilde\Box V + nr^{-1}r^{A}\tilde\nabla_{A}V + r^{-2}\mathring\Delta V.
\end{equation}

%%%%%%%%%%%%%%%%%%%%%%%%%%%%%%%%%%%%%
\section{Gravitational Perturbations}

%%%%%%%%%%%%%%%%%%%%%%%%%%%%%%%%%%%%%
\subsection{Decomposition of the Linearized Metric}

As a $(0,2)$-tensor on the background, a linear perturbation $h_{ab} = \delta g_{ab}$ admits the $(2 + n)$-decomposition
 \begin{equation}\label{h_decomposition0} 
 \delta g=h_{AB} dx^A dx^B+2h_{A\alpha} dx^A dx^\alpha+h_{\alpha\beta} dx^\alpha dx^\beta,
 \end{equation}
with each component of $h_{AB}, h_{A\alpha}, h_{\alpha\beta}$ depending upon all spacetime variables. 

Applying the above lemmae, we have: 
\begin{proposition}\label{decomposition}
Any symmetric two-tensor $\delta g$ of the form \eqref{h_decomposition0} can be decomposed as $\delta g=h_1+h_2 + h_3,$ with

\begin{align}\label{h_1}
\begin{split}
h_1&=h_{AB} dx^A dx^B+2 (\mathring{\nabla}_\alpha H_A)dx^\alpha dx^A+ H\mathring{\sigma}_{\alpha\beta}dx^{\alpha}dx^{\beta}\\
&+\left(\mathring{\nabla}_{\alpha}\mathring{\nabla}_\beta {H}_2-\frac{1}{n}\mathring{\sigma}_{\alpha\beta} \mathring{\Delta} {H}_2\right) dx^\alpha dx^\beta,
\end{split}
\end{align}

\begin{equation}\label{h_2} 
h_2=2  \hat{h}_{A\alpha} dx^\alpha dx^A+(\mathring{\nabla}_\alpha\hat{h}_\beta+\mathring{\nabla}_\beta \hat{h}_\alpha)dx^{\alpha}dx^{\beta}, 
\end{equation} 

\begin{equation}\label{h_3}
h_3 = \hat{h}_{\alpha\beta}dx^{\alpha}dx^{\beta},
\end{equation}
where the following equations are satisfied:

\[\mathring{\nabla}^\alpha \hat{h}_{A\alpha}=0, \mathring{\nabla}^\alpha \hat{h}_\alpha=0, \mathring{\nabla}^\alpha \hat{h}_{\alpha\beta}=0.\]

\end{proposition}

When $n=2$, both $\hat{h}_{A\alpha}$ and $\hat{h}_\alpha$ have potentials ($\underline{H}_A$ and $\underline{H}_2$) and $\hat{h}_{\alpha\beta}$ must vanish, recovering the Hodge type decomposition in our earlier work \cite{HKW}. 

To summarize, in general dimension, we have a splitting of $\delta g$ into three pieces, scalar, co-vector, and two-tensor, as objects on the orbit spheres.

It is also possible to subdivide the linearized metric with respect to the spherical harmonic decomposition outlined in the previous section.  Namely, we split the linearized metric as
\begin{equation}
\delta g = \delta g^{\ell < 2} + \delta g^{\ell \geq 2},
\end{equation}
according to the following proposition.

\begin{proposition}\label{mode_decomp} 
Any symmetric two-tensor $\delta g$ on a spherically symmetric spacetime can be split into $\delta g = \delta g^{\ell <2} + \delta g^{\ell \geq 2}$, in which the components of $\delta g^{\ell \geq 2}$, further decomposed according to Proposition \ref{decomposition}, satisfy
\begin{align*}
\int_{S^2} h_{AB} Y^{\ell m_s(n,\ell)} &= 0, \\
\int_{S^2} HY^{\ell m_s(n,\ell)} &= 0,
\end{align*}
with respect to the scalar harmonics $Y^{\ell m_s(n,\ell)}$ having $\ell < 2$, and
\begin{align*}
\int_{S^2}  H_{A}Y^{1 m_s(n,1)} &= 0,\\
\int_{S^2} \hat{h}_{A\alpha}X_{\beta}^{1 m_v(n,1)}\mathring\sigma^{\alpha\beta} &=0,
\end{align*}
with respect to the $\ell = 1$ scalar harmonics $Y^{1 m_s(n,1)}$ and the $\ell = 1$ co-vector harmonics $X_{\alpha}^{1 m_v(n,1)}$.

We remark the components $H_2$, $\hat{h}_{\alpha},$ and $\hat{h}_{\alpha\beta}$ are necessarily supported in $\ell \geq 2$.
\end{proposition}

\subsection{Decomposition of the Linearized Ricci Tensor}

We recall that a perturbation of the Ricci curvature $\delta R_{bd}$ satisfies
\begin{equation}\label{linear_Ricci} 
2\delta R_{bd}= g^{ae}(\nabla_a\nabla_d h_{eb}+\nabla_a\nabla_b h_{ed}-\nabla_d \nabla_b h_{ea}-\nabla_a\nabla_e h_{bd}).
\end{equation}

Perturbing about a spherically symmetric spacetime, with radial function $r$, we record the calculations in Appendix B of Kodama-Ishibashi-Seto \cite{KI1}:
\begin{align}
\begin{split}\label{AB}
2\delta R_{AB} &= -\tilde{\Box}h_{AB} -\tilde{\nabla}_{A}\tilde{\nabla}_{B}\left(g^{CD}h_{CD}\right) + \tilde{\nabla}_{A}\tilde{\nabla}^{C}h_{CB} + \tilde{\nabla}_{B}\tilde{\nabla}^{C}h_{CA}\\
&+{\tilde{R}_{A}}^{C}h_{CB} + {\tilde{R}_{B}}^{C}h_{CA} - 2\tilde{R}_{ACBD}h^{CD} - r^{-2}\mathring\Delta h_{AB}\\
&+ nr^{-1}r^{C}\left(\tilde{\nabla}_{B}h_{CA} + \tilde{\nabla}_{A}h_{CB} -\tilde{\nabla}_{C}h_{AB}\right)\\
&+r^{-2}\left(\tilde{\nabla}_{A}\mathring\nabla^{\alpha}h_{B\alpha} + \tilde{\nabla}_{B}\mathring\nabla^{\alpha}h_{A\alpha}\right)\\
&-nr^{-3}r_{B}\tilde{\nabla}_{A}H - nr^{-3}r_{A}\tilde{\nabla}_{B}H\\
&+ 4nr^{-4}r_{A}r_{B}H - n\tilde{\nabla}_{A}\tilde{\nabla}_{B}\left(r^{-2}H\right),
\end{split}
\end{align}

\begin{align}
\begin{split}\label{Aalpha}
2\delta R_{A\alpha} &= \mathring\nabla_{\alpha}\tilde{\nabla}^{B}h_{AB} + (n-2)r^{-1}r^{B}\mathring\nabla_{\alpha}h_{AB} - r\mathring\nabla_{\alpha}\tilde{\nabla}_{A}\left(r^{-1}g^{BC}h_{BC}\right)\\
&-r\tilde{\Box}\left(r^{-1}h_{A\alpha}\right) - nr^{-1}r^{B}\tilde{\nabla}_{B}h_{A\alpha} - r^{-2}\mathring{\Delta} h_{A\alpha}\\
&+\left[(n+1)r^{-2}r^{B}r_{B} + (n-1)r^{-2}\left(1-r^{B}r_{B}\right) - r^{-1}\left(\tilde{\Box}r\right)\right]h_{A\alpha}\\
&-r_{A}\tilde{\nabla}^{B}\left(r^{-1}h_{B\alpha}\right) +(n+1)r^{-1}r^{B}\tilde{\nabla}_{A}h_{B\alpha} + r^{-2}r_{A}r^{B}h_{B\alpha}\\
&+ r\tilde{\nabla}_{A}\tilde{\nabla}^{B}\left(r^{-1}h_{B\alpha}\right) + {\tilde{R}_{A}}^{B}h_{B\alpha}\\
&+ (n+1)r\tilde{\nabla}_{A}\left(r^{-2}r^{B}\right)h_{B\alpha} - (n+2)r^{-1}\left(\tilde{\nabla}_{A}\tilde\nabla^{B} r\right)h_{B\alpha}\\
&+r^{-2}\mathring\nabla_{\alpha}\mathring\nabla^{\beta}h_{A\beta} + r\tilde{\nabla}_{A}\left(r^{-3}\mathring\nabla^{\beta}h_{\alpha\beta}\right) + r^{-3}r_{A}\mathring\nabla^{\beta}h_{\alpha\beta}\\
&-nr^{-3}r_{A}\mathring\nabla_{\alpha}H - nr\mathring\nabla_{\alpha}\tilde{\nabla}_{A}\left(r^{-3}H\right),
\end{split}
\end{align}

\begin{align}
\begin{split}\label{alphabeta}
2\delta R_{\alpha \beta} &= \left[2rr^{A}\tilde\nabla^{B}h_{AB} + 2(n-1)r^{A}r^{B}h_{AB} + 2r\left(\tilde\nabla^A\tilde\nabla^B r\right)h_{AB}\right]\mathring\sigma_{\alpha\beta}\\
&-\mathring\nabla_{\alpha}\mathring\nabla_{\beta}\left(h_{AB}g^{AB}\right) - rr^{A}\tilde\nabla_{A}\left(h_{BC}g^{BD}\right)\mathring\sigma_{\alpha\beta}\\
&+r\tilde\nabla^{A}\left(r^{-1}\left(\mathring\nabla_{\alpha} h_{A\beta} + \mathring\nabla_{\beta} h_{A\alpha}\right)\right)\\
&+(n-1)r^{-1}r^{A}\left(\mathring\nabla_{\alpha}h_{A\beta} + \mathring\nabla_{\beta} h_{A\alpha}\right) + 2r^{-1}r^{A}\mathring\nabla^{\gamma}h_{A\gamma}\mathring\sigma_{\alpha\beta}\\
&-r^2\tilde\Box(r^{-2}h_{\alpha\beta}) - nr^{-1}r^{A}\tilde\nabla_{A}h_{\alpha\beta}-r^{-2}\mathring\Delta h_{\alpha\beta}\\
&+ r^{-2}\left(\mathring\nabla_{\alpha}\mathring\nabla^{\gamma}h_{\gamma\beta}+ \mathring\nabla_{\beta}\mathring\nabla^{\gamma}h_{\gamma\alpha}\right)\\
&+2\left[(n-1)r^{-2}+2r^{-2}r^{A}r_{A}-r^{-1}\left(\tilde\Box r\right)\right]h_{\alpha\beta}\\
&-2r^{-2}(1-r^{A}r_{A})\left(nH\mathring\sigma_{\alpha\beta} - h_{\alpha\beta}\right) - 2nr^{-2}r^{A}r_{A}H\mathring\sigma_{\alpha\beta}\\
&-nr^{-2}\mathring\nabla_{\alpha}\mathring\nabla_{\beta}H -nrr^{A}\tilde\nabla_{A}\left(r^{-2}H\right)\mathring\sigma_{\alpha\beta}.
\end{split}
\end{align}

In the remainder of the subsection, we rewrite the above expressions, with the dual aims of expanding in the metric perturbation components of Proposition \ref{decomposition} and of writing the linearized Ricci tensor in such a decomposed form.

\subsubsection{Pure Quotient Term}
The pure quotient portion \eqref{AB} needs little modification.  Expanding, we find
\begin{align}
\begin{split}\label{AB}
2\delta R_{AB} &= -\tilde{\Box}h_{AB} -\tilde{\nabla}_{A}\tilde{\nabla}_{B}\left(g^{CD}h_{CD}\right) + \tilde{\nabla}_{A}\tilde{\nabla}^{C}h_{CB} + \tilde{\nabla}_{B}\tilde{\nabla}^{C}h_{CA}\\
&+{\tilde{R}_{A}}^{C}h_{CB} + {\tilde{R}_{B}}^{C}h_{CA} - 2\tilde{R}_{ACBD}h^{CD} - r^{-2}\mathring\Delta h_{AB}\\
&+ nr^{-1}r^{C}\left(\tilde{\nabla}_{B}h_{CA} + \tilde{\nabla}_{A}h_{CB} -\tilde{\nabla}_{C}h_{AB}\right)\\
&+r^{-2}\left(\tilde{\nabla}_{A}\mathring\Delta H_{B} + \tilde{\nabla}_{B}\mathring\Delta H_{A}\right)\\
&-nr^{-3}r_{B}\tilde{\nabla}_{A}H - nr^{-3}r_{A}\tilde{\nabla}_{B}H\\
&+ 4nr^{-4}r_{A}r_{B}H - n\tilde{\nabla}_{A}\tilde{\nabla}_{B}\left(r^{-2}H\right),
\end{split}
\end{align}
involving only the scalar piece of $\delta g$.

\subsubsection{Cross Term}

Writing the cross-term of the linearized Ricci tensor \eqref{Aalpha} in terms of a scalar potential and a divergence-free co-vector, we find
\begin{align}\label{Aalpha2}
\begin{split}
2\delta R_{A\alpha} &= \mathring\nabla_{\alpha}\Big[\tilde\nabla^{B}h_{AB} + (n-2)r^{-1}r^{B}h_{AB}+r^{-1}r_{A}\left(g^{BC}h_{BC}\right)-\tilde\nabla_{A}\left(g^{BC}h_{BC}\right)\\
&-\tilde\Box H_{A} + (2-n)r^{-1}r^{B}\tilde\nabla_{B}H_{A} + \tilde\nabla^{B}\tilde\nabla_{A}H_{B} -2r^{-1}\left(\tilde\nabla_{A}\tilde\nabla^{B} r\right)H_{B}\\
&+ 2(1-n)r^{-2}r_{A}r^{B}H_{B} - 2r^{-1}r_{A}\tilde\nabla^{B}H_{B} + nr^{-1}r^{B}\tilde\nabla_{A}H_{B}\\
&+(1-n)\tilde\nabla_{A}\left(r^{-2}H\right) + (n-1)\tilde\nabla_{A}\left(r^{-2}\left(H_2+\frac{1}{n}\mathring\Delta H_2\right)\right)\Big]\\
&+\Big[-\tilde\Box \hat{h}_{A\alpha} - r^{-2}\mathring\Delta\hat{h}_{A\alpha} + (2-n)r^{-1}r^{B}\tilde\nabla_{B}\hat{h}_{A\alpha}\\
&+(n-1)r^{-2}\hat{h}_{A\alpha} +\tilde\nabla^{B}\tilde\nabla_{A}\hat{h}_{B\alpha} -2r^{-1}\left(\tilde\nabla_{A}\tilde\nabla^{B}r\right)\hat{h}_{B\alpha}\\
&+2(1-n)r^{-2}r_{A}r^{B}\hat{h}_{B\alpha}-2r^{-1}r_{A}\tilde\nabla^{B}\hat{h}_{B\alpha}+nr^{-1}r^{B}\tilde\nabla_{A}\hat{h}_{B\alpha}\\
&+(n-1)\tilde\nabla_{A}\left(r^{-2}\hat{h}_{\alpha}\right) + \tilde\nabla_{A}\left(r^{-2}\mathring\Delta\hat{h}_{\alpha}\right)\Big],
\end{split}
\end{align}
where we have likewise expanded the metric perturbation in terms of the decomposition of Proposition \ref{decomposition}.  Note that the scalar potential of $\delta R_{A\alpha}$ involves only the scalar piece of $\delta g$, with the co-vector pieces being similarly related.

We remark that, although the reduction seems quite complex, much of the work in obtaining \eqref{Aalpha2} from \eqref{Aalpha} is straightforward.  Namely, the only difficult terms in \eqref{Aalpha} constitute 
\[- r^{-2}\mathring{\Delta} h_{A\alpha} +r^{-2}\mathring\nabla_{\alpha}\mathring\nabla^{\beta}h_{A\beta} + r\tilde{\nabla}_{A}\left(r^{-3}\mathring\nabla^{\beta}h_{\alpha\beta}\right) + r^{-3}r_{A}\mathring\nabla^{\beta}h_{\alpha\beta}.\]  To deal with such terms, we use the angular commutation relations \eqref{angularCommutation}; in particular, we rely upon the identity
\begin{align*}
&\mathring\Delta h_{A\alpha} = \mathring\nabla_{\alpha}\mathring\Delta H_{A} + (n-1)\mathring\nabla_{\alpha} H_{A} + \mathring\Delta \hat{h}_{A\alpha},\\
&\mathring\nabla^{\gamma}h_{\alpha\gamma} = \mathring\nabla_{\alpha}H + \frac{n-1}{n}\mathring\nabla_{\alpha}\left(\mathring\Delta + n\right)H_2 + \left(\mathring\Delta + (n-1)\right)\hat{h}_{\alpha}
\end{align*}

\subsubsection{Pure Angular Term}
The decomposition of the angular term \eqref{alphabeta} is lengthiest of all.  We begin by calculating the trace

\begin{align}\label{alphabetaTrace}
\begin{split}
2\delta R_{\alpha\beta}\mathring\sigma^{\alpha\beta} &= n\Big(2rr^{A}\tilde\nabla^{B}h_{AB} + 2(n-1)r^{A}r^{B}h_{AB} + 2r\left(\tilde\nabla^{A}\tilde\nabla^{B}r\right)h_{AB}\\
&-rr^{A}\tilde\nabla_{A}\left(h_{BC}g^{BC}\right)\Big) - \mathring\Delta(h_{AB}g^{AB})\\
&+4(n-1)r^{-1}r^{A}\mathring\Delta H_{A} + 2\tilde\nabla^{A}\mathring\Delta H_{A}\\
&-nr^{2}\tilde\Box\left(r^{-2}H\right) + 2(1-n)r^{-2}\mathring\Delta H - 2n^2r^{-1}r^{A}\tilde\nabla_{A}H\\
&- 2nr^{-1}\left(\tilde\Box r\right) H + 2n(n+1)r^{-2}r^{A}r_{A}H\\
&+2(n-1)r^{-2}\left(\mathring\Delta H_2 + \frac{1}{n}\mathring\Delta\mathring\Delta H_2\right).
\end{split}
\end{align}

In calculating the traceless portion, we introduce the notation 
\begin{align}
\begin{split}
\check{h}_{\alpha\beta} &:= h_{\alpha\beta} - H\mathring\sigma_{\alpha\beta},\\
\check{2\delta R}_{\alpha\beta} &:= 2\delta R_{\alpha\beta} - \frac{1}{n}\left(2\delta R_{\gamma\delta} \mathring\sigma^{\gamma\delta}\right)\mathring\sigma_{\alpha\beta}.
\end{split}
\end{align}

Concretely,
\begin{equation}
\check{h}_{\alpha\beta} = \left(\mathring\nabla_{\alpha}\mathring\nabla_{\beta}H_2 - \frac{1}{n}\mathring\Delta H_2\mathring\sigma_{\alpha\beta}\right) + \left(\mathring\nabla_{\alpha}\hat{h}_{\beta} + \mathring\nabla_{\beta}\hat{h}_{\alpha}\right) + \hat{h}_{\alpha\beta}.
\end{equation}

A preliminary calculation yields
\begin{align}\label{tracelessOne}
\begin{split}
\check{2\delta R}_{\alpha\beta} &= -\mathring\nabla_{\alpha}\mathring\nabla_{\beta}\left(h_{AB}g^{AB}\right) + \frac{1}{n}\mathring\Delta(h_{AB}g^{AB})\mathring\sigma_{\alpha\beta}\\
&+2\tilde\nabla^{A}\left(\mathring\nabla_{\alpha}\mathring\nabla_{\beta}H_{A}-\frac{1}{n}\mathring\Delta H_{A}\mathring\sigma_{\alpha\beta}\right)\\
&+2(n-2)r^{-1}r^{A}\left(\mathring\nabla_{\alpha}\mathring\nabla_{\beta}H_{A} - \frac{1}{n}\mathring\Delta H_{A}\mathring\sigma_{\alpha\beta}\right)\\
&+\tilde\nabla^{A}\left(\mathring\nabla_{\alpha}\hat{h}_{A\beta} + \mathring\nabla_{\beta}\hat{h}_{A\alpha}\right)\\
&+(n-2)r^{-1}r^{A}\left(\mathring\nabla_{\alpha}\hat{h}_{A\beta} + \mathring\nabla_{\beta}\hat{h}_{A\alpha}\right)\\
&-r^{2}\tilde\Box\left(r^{-2}\check{h}_{\alpha\beta}\right)-nr^{-1}r^{A}\tilde\nabla_{A}\check{h}_{\alpha\beta}-r^{-2}\mathring\Delta\check{h}_{\alpha\beta}\\
&+2\left[(n-1)r^{-2}+2r^{-2}r^{A}r_{A}-r^{-1}\left(\tilde\Box r\right)\right]\check{h}_{\alpha\beta}\\
&+2r^{-2}(1-r^{A}r_{A})\check{h}_{\alpha\beta} - nr^{-2}\left(\mathring\nabla_{\alpha}\mathring\nabla_{\beta}H - \frac{1}{n}\mathring\Delta H \mathring\sigma_{\alpha\beta}\right)\\
&+r^{-2}\left[\left(\mathring\nabla_{\alpha}\mathring\nabla^{\gamma}h_{\gamma\beta} + \mathring\nabla_{\beta}\mathring\nabla^{\gamma}h_{\gamma\alpha}\right) - \frac{2}{n}\left(\mathring\nabla^{\gamma}\mathring\nabla^{\delta}h_{\gamma\delta}\right)\mathring\sigma_{\alpha\beta}\right].
\end{split}
\end{align}

It remains to expand the traceless part of the metric perturbation, and to collect the terms of the linearized Ricci quantity with respect to such a decomposition.  We emphasize that this amounts to expressing the symmetric traceless two-tensor $\check{\delta R}_{\alpha\beta}$ as the sum of the traceless Hessian of a scalar function, the symmetrized gradient of a divergence-free co-vector, and a divergence-free two-tensor. As with the cross-term, there are few troublesome terms, not admitting a straightforward expansion and collection.  Such terms constitute
\[r^{-2}\left[\left(\mathring\nabla_{\alpha}\mathring\nabla^{\gamma}h_{\gamma\beta} + \mathring\nabla_{\beta}\mathring\nabla^{\gamma}h_{\gamma\alpha}\right) - \frac{2}{n}\left(\mathring\nabla^{\gamma}\mathring\nabla^{\delta}h_{\gamma\delta}\right)\mathring\sigma_{\alpha\beta}\right] - r^{-2}\mathring\Delta\check{h}_{\alpha\beta}.\]

Again, these troublesome terms can be expanded and collected by careful application of the angular commutation relations \eqref{angularCommutation}.   In particular, it is useful to note
\begin{align*}
&\mathring\nabla^{\gamma}\check{h}_{\alpha\gamma} = \frac{n-1}{n}\mathring\nabla_{\alpha}\left(\mathring\Delta + n\right)H_2 + \left(\mathring\Delta + (n-1)\right)\hat{h}_{\alpha},\\
&\mathring\nabla^{\alpha}\mathring\nabla^{\gamma}\check{h}_{\alpha\gamma} = \frac{n-1}{n}\mathring\Delta\left(\mathring\Delta + n\right)H_2,\\
&\mathring\Delta\left(\mathring\nabla_{\alpha}\mathring\nabla_{\beta} H_2-\frac{1}{n}\mathring\Delta H_2\mathring\sigma_{\alpha\beta}\right)\\ 
&= \left(\mathring\nabla_{\alpha}\mathring\nabla_{\beta}\mathring\Delta H_2 - \frac{1}{n}\mathring\Delta\mathring\Delta H_2\mathring\sigma_{\alpha\beta}\right) + 2n\left(\mathring\nabla_{\alpha}\mathring\nabla_{\beta}H_2 - \frac{1}{n}\mathring\Delta H_2\mathring\sigma_{\alpha\beta}\right),\\
&\mathring\Delta\left(\mathring\nabla_{\alpha}\hat{h}_{\beta} + \mathring\nabla_{\beta}\hat{h}_{\alpha}\right) = \left(\mathring\nabla_{\alpha}\mathring\Delta\hat{h}_{\beta} + \mathring\nabla_{\beta}\mathring\Delta\hat{h}_{\alpha}\right) + (n+1)\left(\mathring\nabla_{\alpha}\hat{h}_{\beta}+\mathring\nabla_{\beta}\hat{h}_{\alpha}\right).\\
\end{align*}

In the end, we find
\begin{align}\label{tracelessTwo}
\begin{split}
\check{2\delta R}_{\alpha\beta} &= \left(\mathring\nabla_{\alpha}\mathring\nabla_{\beta}-\frac{1}{n}\mathring\Delta \mathring\sigma_{\alpha\beta}\right)\Big[\left(-h_{AB}g^{AB}\right)+2\tilde\nabla^{A}H_{A}\\
&+2(n-2)r^{-1}r^{A}H_{A}-r^2\tilde\Box\left(r^{-2}H_2\right)-nr^{-1}r^{A}\tilde\nabla_{A}H_2\\
&+ 2r^{-2}r^{A}r_{A}H_2 -2r^{-1}\left(\tilde\Box r\right)H_2 + 2(n-1)r^{-2}H_2\\
&+ \left(\frac{n-2}{n}\right)r^{-2}\mathring\Delta H_2 + (2-n)r^{-2}H\Big]\\
&+\Big[\mathring\nabla_{\alpha}\Big(\tilde\nabla^{A}\hat{h}_{A\beta} +(n-2)r^{-1}r^{A}\hat{h}_{A\beta} -r^{2}\tilde\Box\left(r^{-2}\hat{h}_{\beta}\right)\\
&-nr^{-1}r^{A}\tilde\nabla_{A}\hat{h}_{\beta} +2(n-1)r^{-2}\hat{h}_{\beta}+2r^{-2}r^{A}r_{A}\hat{h}_{\beta}\\
&-2r^{-1}\left(\tilde\Box r\right)\hat{h}_{\beta}\Big)\\
&+\mathring\nabla_{\beta}\Big(\tilde\nabla^{A}\hat{h}_{A\alpha} +(n-2)r^{-1}r^{A}\hat{h}_{A\alpha} -r^{2}\tilde\Box\left(r^{-2}\hat{h}_{\alpha}\right)\\
&-nr^{-1}r^{A}\tilde\nabla_{A}\hat{h}_{\alpha} +2(n-1)r^{-2}\hat{h}_{\alpha}+2r^{-2}r^{A}r_{A}\hat{h}_{\alpha}\\
&-2r^{-1}\left(\tilde\Box r\right)\hat{h}_{\alpha}\Big)\Big]\\
&+\Big[-r^{2}\tilde\Box\left(r^{-2}\hat{h}_{\alpha\beta}\right) - nr^{-1}r^{A}\tilde\nabla_{A}\hat{h}_{\alpha\beta} -r^{-2}\mathring\Delta \hat{h}_{\alpha\beta}\\
&+2r^{-2}\left(n+r^{A}r_{A}-r\tilde\Box r\right)\hat{h}_{\alpha\beta}\Big].
\end{split}
\end{align}
Note that the scalar, co-vector, and two-tensor parts of the linearized Ricci tensor are determined by the corresponding pieces of the metric perturbation. 

Subsequently, we assume that the linearized vacuum Einstein equations are satisfied, and study the scalar, co-vector, and two-tensor portions decomposed in this subsection.

%%%%%%%%%%%%%%%%%%%%%%%%%%%%%%%%%%%%%%
\section{Special Solutions of Linearized Gravity}

%%%%%%%%%%%%%%%%%%%%%%%%%%%%%%%%%%%%%%
\subsection{Pure Gauge Solutions}

Diffeomorphism invariance of the Einstein equations reduces in the linear theory to invariance under infinitesimal deformations of the underlying spacetime metric.  That is, with $X$ a co-vector and $\delta g$ a metric perturbation, satisfying the linearized vacuum Einstein equations, the concatenation $\delta g + \mathcal{L}_{X}g$ yields a new solution of the same.  In this subsection, we record how such gauge transformations affect the various components of $\delta g$ with respect to the decomposition outlined in the first subsection.

We decompose $X$ as
\begin{equation}
X = X_a dx^{a} = X_A dx^A + X_{\alpha} dx^{\alpha} = X_A dx^A + \left(\mathring\nabla_{\alpha}X_2\right) dx^{\alpha} + \hat{X}_{\alpha} dx^{\alpha},
\end{equation}
with $\hat{X}$ a divergence-free co-vector on the orbit spheres.  

Calculating the deformation tensor $\pi_{X}$, and splitting according to the decomposition of Proposition \ref{decomposition}, we have $\pi_X = \pi_1 + \pi_2$ with
\begin{equation}\label{pi_1}
\begin{split}
\pi_1&= \left[\tilde{\nabla}_A X_B+\tilde{\nabla}_B X_A\right] dx^A dx^B\\
&+ 2\mathring{\nabla}_\alpha \left[X_{A} + \tilde{\nabla}_A X_2-2r^{-1} r_{A}X_2\right] dx^A dx^\alpha\\
&+2\left[rr^{A}X_{A} + \frac{1}{n}\mathring\Delta X_2\right]\mathring\sigma_{\alpha\beta}dx^\alpha dx^\beta\\
&+2\left[\mathring\nabla_{\alpha}\mathring\nabla_{\beta}X_2 - \frac{1}{n}\mathring\Delta X_2 \mathring\sigma_{\alpha\beta}\right]dx^{\alpha}dx^{\beta},
\end{split}
\end{equation}
\begin{equation}\label{pi_2}
\begin{split}
\pi_2&= 2\left[\tilde{\nabla}_{A}\hat{X}_{\alpha} - 2r^{-1}r_{A}\hat{X}_{\alpha}\right] dx^A dx^\alpha\\
& + \left[\mathring\nabla_{\alpha}\hat{X}_{\beta} + \mathring\nabla_{\beta}\hat{X}_{\alpha}\right]dx^{\alpha}dx^{\beta}.
\end{split}
\end{equation}

In particular, we observe the divergence-free portion $\hat{h}$ remains unchanged; that is, $\hat{h}$ is a gauge-invariant quantity.

%%%%%%%%%%%%%%%%%%%%%%%%%%%%%%%%%%%%%
\subsection{Linearized Myers-Perry Solutions}

We briefly describe the standard presentation of the Myers-Perry solutions, generalizing the Kerr solution to higher dimensions, following \cite{Myers}.  Assuming an odd number of spacetime dimensions $d = 2q + 1$, the Minkowski metric can be written as
\begin{align*}
\bar{g} &= -dt^2 + \sum_{i=1}^{q} \left(dx_{i}^2 + dy_{i}^2\right)\\
&= -dt^2 + dr^2 + r^2\sum_{i=1}^{q}\left(d\mu_{i}^2 + \mu_{i}^2d\phi_{i}^2\right).
\end{align*}
Here, we have expressed the even number of spatial coordinates as paired Cartesian coordinates $(x_{i}, y_{i})$ for $q$ mutually orthogonal planes.  Rewritten in generalized polar coordinates, we have the relations
\begin{align}\label{coordinateRelation}
\begin{split}
x_{i} &= r\mu_{i}\cos\phi_{i},\\
y_{i} &= r\mu_{i}\sin\phi_{i},
\end{split}
\end{align}
with the constraint
\[ \sum_{i}^{q}\mu_{i}^2 = 1.\]
With even spacetime dimension $d = 2q+2$, there is an extra unpaired spatial coordinate, also regarded as an azimuthal polar coordinate,
\[ z = r\alpha,\]
with $\alpha \in [-1, 1],$
such that the Minkowski metric has the polar form
\[\bar{g} = -dt^2 + dr^2 + r^2\sum_{i=1}^{q}\left(d\mu_{i}^2 + \mu_{i}^2d\phi_{i}^2\right) + r^2d\alpha^2,\]
with the constraint
\[ \sum_{i}^{q}\mu_{i}^2 + \alpha^2 = 1.\]

Likewise, the Myers-Perry solutions have Boyer-Lindquist type coordinates featuring the generalized polar coordinates above.  Defining
\begin{align*}
F &:= 1 - \sum_{i=1}^{q}\frac{a_{i}^2\mu_{i}^2}{r^2+a_{i}^2},\\
\Pi &:= \prod_{i=1}^{q} (r^2+a_{i}^2),
\end{align*}
in even spacetime dimension $d = 2q + 2$ the metric takes the form
\begin{align*}
g_{M,a_{i}} = &-dt^2 + \frac{2M r}{\Pi F}\left(dt + \sum_{i=1}^{q}a_{i}\mu_{i}^2d\phi_{i}\right)^2 + \frac{\Pi F}{\Pi - 2Mr}dr^2\\
&+ \sum_{i=1}^{q}(r^2+a_{i}^2)(d\mu_{i}^2 + \mu_{i}^2d\phi_{i}^2) + r^2d\alpha^2,
\end{align*}
whereas in odd spacetime dimension $d = 2q + 1$,
\begin{align*}
g_{M,a_{i}} = &-dt^2 + \frac{2M r}{\Pi F}\left(dt + \sum_{i=1}^{q}a_{i}\mu_{i}^2d\phi_{i}\right)^2 + \frac{\Pi F}{\Pi - 2Mr}dr^2\\
&+ \sum_{i=1}^{q}(r^2+a_{i}^2)(d\mu_{i}^2 + \mu_{i}^2d\phi_{i}^2).
\end{align*}
Note that both metrics are parametrized by the mass $M > 0$ and the $q$ angular velocity parameters $a_{i}$, with parameter increases occurring only in odd dimension.  

In each case, the metric can be rewritten in a manner suggestive of our perturbative framework:
\begin{align}\label{MyersPerryExpansion}
\begin{split}
g_{M,a_{i}} = &-\left(1-\frac{2M}{r^{n-1}}\right)dt^2 + \left(1-\frac{2M}{r^{n-1}}\right)^{-1}dr^2 + r^2\mathring\sigma_{\alpha\beta}dx^{\alpha}dx^{\beta}\\
&+\sum_{i = 1}^{q} \frac{4M}{r^{n-1}}\mu_{i}^2a_{i}dtd\phi_{i} + O(a_{i})^2,
\end{split}
\end{align}
where $n = 2q$ in even dimension and $n = 2q -1$ in odd dimension as given above.  Note that we have rewritten the top-order polar terms using our earlier concise notation for the round metric on the unit sphere.

We treat separately the linearized change in mass and change in angular velocity arising from \eqref{MyersPerryExpansion} below.

\subsection{Linearized Change in Mass}
Linearized mass solutions are generated by constant multiples of
\begin{equation}\label{linMass}
h_{ab}dx^{a}dx^{b} = \frac{1}{r^{n-1}}dt^2 + \frac{r^{n-1}}{(r^{n-1}-2M)^2}dr^2.
\end{equation}
  By direct calculation, one can verify that \eqref{linMass} satisfies the linearized vacuum Einstein equations.  With respect to the aforementioned Hodge and spherical harmonic decompositions, such solutions are scalar with support at $\ell = 0$.

\subsection{Linearized Change in Angular Velocity}

Each pair $(x_i, y_i)$ in the construction of the Myers-Perry solution gives a rotation Killing vector field
\[ \frac{x_i}{r}d\left(\frac{y_i}{r}\right) - \frac{y_i}{r}d\left(\frac{x_i}{r}\right) = \mu_i^2d\phi_{i}\]
in the background Schwarzschild spacetime, where we have used the coordinate relation \eqref{coordinateRelation}.  In view of this, the Myers-Perry expansion \eqref{MyersPerryExpansion} gives rise to linearized solutions
\begin{align*}
&h_{ab}dx^{a}dx^{b} = \frac{1}{r^{n-1}}\mu_{i}^2dtd\phi_{i}, i=1\cdots q \\
\end{align*}

Since $\frac{x_i}{r}d\left(\frac{y_i}{r}\right) - \frac{y_i}{r}d\left(\frac{x_i}{r}\right)$ is dual to a rotation Killing field and the duals of $X_{\alpha}^{1 m_v(n,1)} dx^\alpha$,  $1\leq m_v(n,1) \leq \frac{1}{2}n(n+1)$ form a basis of the space of rotation Killing fields per \eqref{l1Covector}. For each $i=1\cdots q$, there exist $c_{m_v(n,1)}, m_v(n,1)=1 \cdots  \frac{1}{2}n(n+1)$ such that \[\mu_{i}^2d\phi_{i}=\sum_{m_v(n,1)=1}^{\frac{1}{2}n(n+1)}c_{m_v(n,1)}X_{\alpha}^{1 m_v(n,1)}dx^{\alpha}.\]

On the other hand, we claim that each $X_{\alpha}^{1 m_v(n,1)} dx^\alpha $ on $S^n$ can be expressed as a linear combination of 
$ \mu_i^2d\phi_{i}, i=1\cdots q$ by choosing a suitable coordinate system. We assume that $n=2q-1$ and  take an arbitrary Cartesian coordinate system $X^A, A=1\cdots 2q$ on $\R^{2q}$. Recall from \eqref{l1Covector} that $\tilde{X}^A d\tilde{X}^B-\tilde{X}^B d\tilde{X}^A, 1\leq A<B\leq 2q$ form a basis of the space of rotation Killing fields. In particular, there exists an alternating constant $2q\times 2q$ matrix $P_{AB}$, $P_{AB}=-P_{BA}$ such that 
$X_{\alpha}^{1 m_v(n,1)} dx^\alpha=\sum_{A, B=1}^{2q} P_{AB} \tilde{X}^A d\tilde{X}^B$.  Applying the spectral theorem to $P_{AB}$, we deduce that there exist a new coordinate system $x_i, y_i, i=1\cdots q$ of $\R^{2q}$ and constants $\lambda_i, i=1\cdots q$ such that 
\[ X_{\alpha}^{1 m_v(n,1)} dx^\alpha= \sum_{i=1}^q \lambda_i [\frac{x_i}{r}d\left(\frac{y_i}{r}\right) - \frac{y_i}{r}d\left(\frac{x_i}{r}\right)] ,\]
The coordinate change \eqref{coordinateRelation} turns the last expression into  $\sum_{i=1}^q \lambda_i \mu_{i}^2 d\phi_{i}$. The case of $n=2q$ can be derived similarly. 

Infinitesimal change in angular velocity is encoded in the basis solutions
\begin{equation}\label{linAngMom}
h_{ab}dx^{a}dx^{b} = \frac{1}{r^{n-1}}X_{\alpha}^{1 m_v(n,1)}dt dx^{\alpha},
\end{equation}
with the $X_{\alpha}^{1 m_v(n,1)}$ spanning the $\ell = 1$ eigenspace of divergence-free co-vectors, $1\leq m_v(n,1) \leq \frac{1}{2}n(n+1)$ \eqref{l1Covector}.  These basis two-tensors are co-vector solutions supported at $\ell =1$, satisfying the linearized vacuum Einstein equations.  

Together, linear combinations of the above linear perturbations of the Schwarzschild metric form the family of linearized Kerr solutions.

\begin{definition}\label{linMPBasis}
The linearized Kerr solutions of the linearized vacuum Einstein equation on the Schwarzschild spacetime in $n +2$ dimensions are linear combinations of the basis solutions
\begin{align}
\begin{split}
K &= \frac{1}{r^{n-1}}dt^2 + \frac{r^{n-1}}{(r^{n-1}-2M)^2}dr^2,\\
K_{m} &= \frac{1}{r^{n-1}}X_{\alpha}^{1 m}dt dx^{\alpha},
\end{split}
\end{align}
where $X_{\alpha}^{1 m} = X_{\alpha}^{1 m_v(n,1)}$, with eigenvalue given by \eqref{coVectorSpectra}, and
\[1 \leq m = m_v(n,1) \leq \frac{1}{2}n(n+1).\]
\end{definition}
%%%%%%%%%%%%%%%%%%%%%%%%%%%%%%%%%%%%%%
\section{Analysis of the Lower Angular Modes}

\subsection{The $\ell = 0$ Scalar Mode}

In this case the scalar portion $h_1$ has the form
\[ h_1 = h_{AB}dx^{A}dx^{B} + H\mathring\sigma_{\alpha\beta}dx^{\alpha}dx^{\beta},\]
and the linearized Ricci tensor reduces to
\begin{align*}
2\delta R_{AB} &= g_{AB}\left(\tilde\nabla^{C}\tilde\nabla^{D}h_{CD}\right) + 2\tilde{K}h_{AB}-\tilde{K}Sg_{AB}-\left(\tilde\Box S\right)g_{AB}\\
&+ nr^{-1}r^{C}\left(\tilde{\nabla}_{B}h_{CA} + \tilde{\nabla}_{A}h_{CB} -\tilde{\nabla}_{C}h_{AB}\right)\\
&-nr^{-3}r_{B}\tilde{\nabla}_{A}H - nr^{-3}r_{A}\tilde{\nabla}_{B}H\\
&+ 4nr^{-4}r_{A}r_{B}H - n\tilde{\nabla}_{A}\tilde{\nabla}_{B}\left(r^{-2}H\right),
\end{align*}
\begin{align*}
2\delta R_{\alpha\beta}\mathring\sigma^{\alpha\beta} &= n\Big(2rr^{A}\tilde\nabla^{B}h_{AB} + 2(n-1)r^{A}r^{B}h_{AB} + r\left(\tilde\Box r\right)S\\
&-rr^{A}\tilde\nabla_{A}S\Big) -nr^{2}\tilde\Box\left(r^{-2}H\right)  - 2n^2r^{-1}r^{A}\tilde\nabla_{A}H\\
&- 2nr^{-1}\left(\tilde\Box r\right) H + 2n(n+1)r^{-2}r^{A}r_{A}H,
\end{align*}
where we have used $S := g^{AB}h_{AB}$ and the identity
\begin{align}\label{quotientSimplify}
\begin{split}
&-\tilde{\Box}h_{AB} + \tilde{\nabla}^{C}\tilde{\nabla}_{A}h_{CB} + \tilde{\nabla}^{C}\tilde{\nabla}_{B}h_{CA}\\
&=g_{AB}\left(\tilde\nabla^{C}\tilde\nabla^{D}h_{CD}\right) + 2\tilde{K}h_{AB} -\tilde{K}Sg_{AB} +\tilde\nabla_{A}\tilde\nabla_{B}S-\left(\tilde\Box S\right)g_{AB}.
\end{split}
\end{align}

Given a co-vector $X = X_{A}dx^{A}$, the associated pure gauge solution $\pi_{X}$ modifies $h_1$ as
\begin{align}
\begin{split}
h_1 &\rightarrow h_1 - \pi_{X},\\
h_{AB} &\rightarrow h_{AB} - \tilde\nabla_{A}X_{B} - \tilde\nabla_{B} X_{A},\\
H &\rightarrow H - 2rr^{A}X_{A}.
\end{split}
\end{align}
We choose $X$ to eliminate $S$ and $H$; that is, $X$ satisfies
\begin{align}
\begin{split}
2\tilde\nabla^{A}X_{A} &= S,\\
2rr^{A}X_{A} &= H.
\end{split}
\end{align}
Rewriting the first equation as
\[-2T^{A}\tilde\nabla_{A}\left(T^{B}X_{B}\right) + 2r^{A}\tilde\nabla_{A}\left(r^{B}X_{B}\right) =r^{B}r_{B}S,\]
we observe that $r^{A}X_{A}$ is determined by the second equation, while $T^{A}X_{A}$ is determined by the first equation only up to specification on an initial data slice.  We utilize this additional gauge freedom to set
\begin{align}
\begin{split}
&T^{A}r^{B}\left(h_{AB} - \tilde\nabla_{A}X_{B} - \tilde\nabla_{B}X_{A}\right)\\
&=T^{A}r^{B}h_{AB} - T^{A}\tilde\nabla_{A}\left(r^{B}X_{B}\right)\\
&- (1-\mu)r^{A}\tilde\nabla_{A}\left((1-\mu)^{-1}T^{B}X_{B}\right) =0
\end{split}
\end{align}
on an initial data slice.  Note that there is still residual gauge freedom of the form $T^{B}X_{B} = c(1-\mu)$; these gauge transformations correspond to scalar multiples of the static Killing vector field $T$.  To summarize, we have performed a change of gauge eliminating $S$ and $H$ globally and $T^{A}r^{B}h_{AB}$ on an initial data slice.

The Einstein equations for the gauge-normalized solution $h_1^{*} = h_1 - \pi_{X}$ amount to
\begin{align*}
&g_{AB}\left(\tilde\nabla^{C}\tilde\nabla^{D}h^{*}_{CD}\right) + 2\tilde{K}h^{*}_{AB} + nr^{-1}r^{C}\left(\tilde\nabla_{B}h^{*}_{CA} + \tilde\nabla_{A}h^{*}_{CB} -\tilde\nabla_{C}h^{*}_{AB}\right) = 0,\\
&2rr^{A}\tilde\nabla^{B}h^{*}_{AB} + 2(n-1)r^{A}r^{B}h^{*}_{AB} = 0.
\end{align*}

Taking the trace of the first equation and comparing with the second equation, we can rewrite the first equation as
\begin{align}\label{lZeroMaster}
\begin{split}
&nr^{-1}r^{C}\left(\tilde\nabla_{B}h^{*}_{CA} + \tilde\nabla_{A}h^{*}_{CB} -\tilde\nabla_{C}h^{*}_{AB}\right) +2\tilde{K}h^{*}_{AB}\\
&+n(n-1)r^{-2}\left(r^{C}r^{D}h^{*}_{CD}\right)g_{AB} = 0.
\end{split}
\end{align}
Contracting \eqref{lZeroMaster} with $r^{A}r^{B}$ and noting the identity
\[ r^{A}r^{B}r^{C}\left(\tilde\nabla_{B}h^{*}_{CA} + \tilde\nabla_{A}h^{*}_{CB} -\tilde\nabla_{C}h^{*}_{AB}\right) = r^{A}\tilde\nabla_{A}\left(r^{B}r^{C}h^{*}_{BC}\right) -\frac{2r}{n}\tilde{K}\left(r^{B}r^{C}h^{*}_{BC}\right),\]
we deduce
\[r^{A}\tilde\nabla_{A}\left(r^{n-1}\left(r^{B}r^{C}h^{*}_{BC}\right)\right) = 0,\]
such that
\[ r^{A}r^{B}h^{*}_{AB} = c(t)r^{-(n-1)},\]
with $c(t)$ an arbitrary function of time.
Contracting \eqref{lZeroMaster} with $T^{A}r^{B}$ instead, we find
\[T^{A}\tilde\nabla_{A}\left(r^{B}r^{C}h^{*}_{BC}\right) = 0,\]
so that $c(t)$ is constant.  We add a linearized Schwarzschild solution of the form \eqref{linMass} to eliminate $r^{A}r^{B}h^{*}_{AB}$; note that this addition preserves the global vanishing of $S^{*}$ and $H^{*}$, as well as that of $T^{A}r^{B}h^{*}_{AB}$ on an initial data slice.  In particular, owing to the vanishing of $S^{*}$, this addition will also eliminate $T^{A}T^{B}h^{*}_{AB}$.  The solution $h_1^{**} = h^{*}_1 - cK = h_{1} - \pi_{X} - cK$ has only the cross-term $T^{A}r^{B}h_{AB}^{**}$ to be accounted for.  The solution still satisfies \eqref{lZeroMaster}, the contraction of which with $T^{A}T^{B}$ leads to
\[ T^{A}\tilde\nabla_{A}\left(T^{B}r^{C}h^{**}_{BC}\right) = 0.\]
Together with the vanishing of $T^{A}r^{B}h^{**}_{AB}$ on an initial data slice, the above gives vanishing of $T^{A}r^{B}h^{**}_{AB}$ globally.  In summary, we have $h_{1} = \pi_{X} + cK$.

%%%%%%%%%%%%%%%%%%%%%%%%%%%%%%%%%%%%%
\subsection{The $\ell = 1$ Scalar Mode}

The scalar portion $h_1$ has the form
\[ h_1 = h_{AB}dx^{A}dx^{B} + 2\mathring\nabla_{\alpha}H_{A}dx^{A}dx^{\alpha} + H\mathring\sigma_{\alpha\beta}dx^{\alpha}dx^{\beta},\]
and the linearized Ricci tensor appears as

\begin{align*}
2\delta R_{AB} &= g_{AB}\left(\tilde\nabla^{C}\tilde\nabla^{D}h_{CD}\right) + 2\tilde{K}h_{AB}-\tilde{K}Sg_{AB}-\left(\tilde\Box S\right)g_{AB}\\
&+ nr^{-2} h_{AB} + nr^{-1}r^{C}\left(\tilde{\nabla}_{B}h_{CA} + \tilde{\nabla}_{A}h_{CB} -\tilde{\nabla}_{C}h_{AB}\right)\\
&-nr^{-2}\left(\tilde{\nabla}_{A} H_{B} + \tilde{\nabla}_{B} H_{A}\right) -nr^{-3}r_{B}\tilde{\nabla}_{A}H - nr^{-3}r_{A}\tilde{\nabla}_{B}H\\
&+ 4nr^{-4}r_{A}r_{B}H - n\tilde{\nabla}_{A}\tilde{\nabla}_{B}\left(r^{-2}H\right),
\end{align*}
\begin{align*}
2\delta R_{A\alpha} &= \mathring\nabla_{\alpha}\Big[\tilde\nabla^{B}h_{AB} + (n-2)r^{-1}r^{B}h_{AB}+r^{-1}r_{A}S-\tilde\nabla_{A}S\\
&-\tilde\Box H_{A} + (2-n)r^{-1}r^{B}\tilde\nabla_{B}H_{A} + \tilde\nabla^{B}\tilde\nabla_{A}H_{B} -2r^{-1}\left(\tilde\nabla_{A}\tilde\nabla^{B} r\right)H_{B}\\
&+ 2(1-n)r^{-2}r_{A}r^{B}H_{B} - 2r^{-1}r_{A}\tilde\nabla^{B}H_{B} + nr^{-1}r^{B}\tilde\nabla_{A}H_{B}\\
&+(1-n)\tilde\nabla_{A}\left(r^{-2}H\right) \Big],
\end{align*}
\begin{align*}
2\delta R_{\alpha\beta}\mathring\sigma^{\alpha\beta} &= n\Big(2rr^{A}\tilde\nabla^{B}h_{AB} + 2(n-1)r^{A}r^{B}h_{AB} + r\left(\tilde\Box r\right)S -rr^{A}\tilde\nabla_{A}S\Big) + nS\\
&-4n(n-1)r^{-1}r^{A} H_{A} - 2n\tilde\nabla^{A}H_{A}\\
&-nr^{2}\tilde\Box\left(r^{-2}H\right) + 2n(n-1)r^{-2} H - 2n^2r^{-1}r^{A}\tilde\nabla_{A}H\\
&- 2nr^{-1}\left(\tilde\Box r\right) H + 2n(n+1)r^{-2}r^{A}r_{A}H,
\end{align*}
where we have used the notation $S := g^{AB}h_{AB}$ and \eqref{quotientSimplify}.

Given a co-vector $X = X_{A}dx^{A} + \mathring\nabla_{\alpha} X_2 dx^{\alpha}$, the associated pure gauge solution $\pi_{X}$ modifies $h_1$ as
\begin{align}
\begin{split}
h_1 &\rightarrow h_1 - \pi_{X},\\
h_{AB} &\rightarrow h_{AB} - \tilde\nabla_{A}X_{B} - \tilde\nabla_{B} X_{A},\\
H_{A} & \rightarrow H_{A} - X_{A} - r^2\tilde\nabla_{A}\left(r^{-2}X_2\right),\\
H &\rightarrow H - 2rr^{A}X_{A} + 2X_2.
\end{split}
\end{align}

We choose $X'$ to eliminate the quantities $H - 2rr^{A}H_{A}$ and $H_{A}$.  This reduction amounts to solving 
\[ -2r^3r^{A}\tilde\nabla_{A}\left(r^{-2}X_2\right) -2X_2 = H - 2rr^{A}H_{A}\]
for $X_2$, then solving
\[ X_{A} + r^{2}\tilde\nabla_{A}\left(r^{-2}X_2\right) = H_{A}\]
for $X_{A}$.  Note that there is residual freedom in the form
\begin{align*}
X_2 &= c(t)r(1-\mu)^{-1/(n-1)},\\
X_{A} &= -r^{2}\tilde\nabla_{A}\left(c(t)r^{-1}(1-\mu)^{-1/(n-1)}\right).
\end{align*}
In particular, we note the transformation
\[ r^{A}r^{B}h_{AB} \rightarrow r^{A}r^{B}h_{AB} + 2c(t)\mu r^{-1}(1-\mu)^{-1/(n-1)}\]
of the component $r^{A}r^{B}h_{AB}$ under this residual gauge freedom.

The linearized Einstein equations for the gauge-normalized solution $h_1^* = h_1 - \pi_{X'}$ amount to
\begin{align*}
&g_{AB}\left(\tilde\nabla^{C}\tilde\nabla^{D}h^*_{CD}\right) + 2\tilde{K}h^*_{AB}-\tilde{K}S^*g_{AB}-\left(\tilde\Box S^*\right)g_{AB}\\
&+ nr^{-2} h^*_{AB} + nr^{-1}r^{C}\left(\tilde{\nabla}_{B}h^*_{CA} + \tilde{\nabla}_{A}h^*_{CB} -\tilde{\nabla}_{C}h^*_{AB}\right) = 0,
\end{align*}
\[\tilde\nabla^{B}h^*_{AB} + (n-2)r^{-1}r^{B}h^*_{AB}+r^{-1}r_{A}S^*-\tilde\nabla_{A}S^* = 0,\]
\[2rr^{A}\tilde\nabla^{B}h^*_{AB} + 2(n-1)r^{A}r^{B}h^*_{AB} + r\left(\tilde\Box r\right)S^* -rr^{A}\tilde\nabla_{A}S^* + S^* = 0.\]
Replacing the first term in the third equation by means of the second equation, we find
\[ rr^{A}\tilde\nabla_{A}S^* - 2r^{A}r_{A}S^* + r(\tilde\Box r)S^* + S^* + 2\left(r^{A}r^{B}h^*_{AB}\right) = 0.\]
Taking the divergence of the second equation, we deduce the relation
\begin{align*}
\tilde\nabla^{A}\tilde\nabla^{B}h^*_{AB} &= \tilde\Box S^* + (1-n)r^{-1}r^{A}\tilde\nabla_{A}S^* + (n-1)r^{-2}r^{B}r_{B}S^* \\
&-\frac{n}{2}r^{-1}(\tilde\Box r)S^* + (n-2)(n-1)r^{-2}\left(r^{A}r^{B}h^*_{AB}\right).
\end{align*}
Rewriting the double divergence term in the first equation in this way, contracting with $r^{A}r^{B}$, and applying the previous relation between $S^*$ and $r^{A}r^{B}h^*_{AB}$, we find an autonomous equation for $r^{A}r^{B}h^*_{AB}$
\[ r^{-1}r^{C}\tilde\nabla_{C}\left(r^{A}r^{B}h^*_{AB}\right) + r^{-2}\left(1+(n-1)r^{C}r_{C}\right)\left(r^{A}r^{B}h^*_{AB}\right) = 0,\]
with general solution
\[ r^{A}r^{B}h^*_{AB} = d(t)\mu r^{-1}(1-\mu)^{-1/(n-1)}.\]
Contracting the first equation with $T^{A}r^{B}$ instead, we find
\[ nr^{-1}T^{C}\tilde\nabla_{C}\left(r^{A}r^{B}h^*_{AB}\right) + nr^{-2}\left(T^{A}r^{B}h^*_{AB}\right) = 0.\]
Finally, contracting the second equation with $r^{A}$, we have
\[ \tilde\nabla^{B}\left(r^{A}h^*_{AB}\right) -\frac{\tilde\Box r}{2} S^* + (n-2)r^{-1}\left(r^{A}r^{B}h^*_{AB}\right) + r^{-1}r^{A}r_{A}S^* - r^{A}\tilde\nabla_{A}S^* = 0.\]
Exercising our residual freedom by the choice $c(t) = -\frac{1}{2}d(t)$  above, with associated co-vector field $\bar{X}$, and letting $X = X' + \bar{X}$, the normalized solution $ h_1^{**} = h_1 - \pi_{X}$ has vanishing $r^{A}r^{B}h^{**}_{AB}$ component.  The equations above immediately imply vanishing of the component $T^{A}r^{B}h^{**}_{AB}$.  It remains to consider the component $S^{**}$, which satisfies
\[ rr^{A}\tilde\nabla_{A}S^{**} - 2r^{A}r_{A}S^{**} + r(\tilde\Box r)S^{**} + S^{**} = 0,\]
\[ -\frac{\tilde\Box r}{2} S^{**} + r^{-1}r^{A}r_{A}S^{**} - r^{A}\tilde\nabla_{A}S^{**} = 0.\]
Taken together, the two equations imply $S^{**} = 0$.  In this way, we have shown that $h_1^{**} = 0$; that is, $h_1 = \pi_{X}$ is a pure gauge solution.

%%%%%%%%%%%%%%%%%%%%%%%%%%%%%%%%%%%%%%
\subsection{The $\ell = 1$ Co-Vector Mode}

The co-vector portion amounts to
\[ h_2 = 2\hat{h}_{A\alpha}dx^{A}dx^{\alpha},\]
and the linearized Ricci tensor takes the form
\begin{align*}
2\delta R_{A\alpha} &= -r^{-n}\epsilon_{AB}\tilde\nabla^{B}\left(r^{n+2}\epsilon^{CD}\tilde\nabla_{D}\left(r^{-2}\hat{h}_{C\alpha}\right)\right).
\end{align*}

Given a co-vector $X = \hat{X}_{\alpha}dx^{\alpha}$, with $\hat{X}_{\alpha}$ satisfying the divergence-free condition $\mathring\nabla^{\alpha}\hat{X}_{\alpha} = 0,$ the associated pure gauge solution $\pi_{X}$ modifies $h_2$ as
\begin{align}
\begin{split}
h_2 &\rightarrow h_2 - \pi_{X},\\
\hat{h}_{A\alpha} &\rightarrow \hat{h}_{A\alpha} - r^2\tilde\nabla_{A}\left(r^{-2}\hat{X}_{\alpha}\right).
\end{split}
\end{align}

We choose $X'$ to eliminate $r^{A}\hat{h}_{A\alpha}$, such that
\begin{equation}
r^{A}\tilde\nabla_{A}\hat{X}_{\alpha} -2r^{-1}r^{A}r_{A}\hat{X}_{\alpha} =  r^{A}\hat{h}_{A\alpha}.
\end{equation}
In accordance with the homogeneous solutions of the equation above, there remains residual gauge freedom in the form $\hat{X}_{\alpha} = \bar{c}_{m_v(n,1)}(t)r^2 X_{\alpha}^{1 m_v(n,1)}$.

We consider the normalized $h_2^{*} = h_2 - \pi_{X'}$ and decompose $\hat{h}^{*}_{A\alpha} = \sum_{m_v(n,1)}\hat{h}_{A\alpha}^{*1 m_v(n,1)}$.  Contracting the linearized Einstein equation
\[ -r^{-n}\epsilon_{AB}\tilde\nabla^{B}\left(r^{n+2}\epsilon^{CD}\tilde\nabla_{D}\left(r^{-2}\hat{h}_{C\alpha}^{*1m_v(n,1)}\right)\right) = 0\]
with $r^{A}$ and $T^{A}$, we deduce
\[ r^{n+2}\epsilon^{CD}\tilde\nabla_{D}\left(r^{-2}\hat{h}_{C\alpha}^{*1 m_v(n,1)}\right) = d_{m_v(n,1)},\]
with $d_{m_v(n,1)}$ a constant.  As $r^{A}\hat{h}^{*1m_v(n,1)}_{A\alpha} = 0$, the above reduces to 
\[ r^{n+2}r^{A}\tilde\nabla_{A}\left(r^{-2}T^{B}\hat{h}_{B\alpha}^{*1 m_v(n,1)}\right) = d_{m_v(n,1)}r^{B}r_{B},\]
with general solution
\[ T^{A}\hat{h}_{A\alpha}^{*1 m_v(n,1)} = c_{m_v(n,1)}(t)r^{2}X_{\alpha}^{1 m_v(n,1)} + \frac{d_{m_v(n,1)}}{r^{n-1}}.\] 
Taking $\bar{X} = \bar{c}_{m_v(n,1)}(t)r^2X_{\alpha}^{1m_v(n,1)}$ with $\bar{c}_{m_v(n,1)}'(t) = c_{m_v(n,1)}(t)$ and letting $X = X' + \bar{X}$, we have shown
\[ h_2 = \pi_{X} + \sum_{m = 1}^{\frac{1}{2}n(n+1)}d_{m}K_{m},\]
with $K_{m}$ being the Myers-Perry solutions of Definition \ref{linMPBasis} and $d_{m} = d_{m_v(n,1)}$.  Note that there remains gauge freedom in the form $\bar{c}(t) \equiv \bar{c}$; such transformations correspond to scalar multiples of the angular Killing fields $\Omega_{i}$.  

%%%%%%%%%%%%%%%%%%%%%%%%%%%%%%%%%%%%
\subsection{Proof of Theorem 1}
Combining the results of the subsections above, we have a proof of Theorem 1.  In particular, adding the various linearized solutions, we obtain a smooth co-vector $X^{\ell <2}$ on the Schwarzschild-Tangherlini background, unique modulo Killing fields, and constants $c, d_{m}$ such that
\begin{equation}
\delta g^{\ell <2} = \pi_{X^{\ell <2}} + cK + \sum_{m = 1}^{\frac{1}{2}n(n+1)} d_{m}K_{m},
\end{equation}
where $K, K_{m}$ are the basis solutions for the linearized Myers-Perry family of Definition \ref{linMPBasis}.

The remainder of the paper concerns the identification and analysis of the gauge-invariant master quantities of the higher angular frequency portion $\delta g^{\ell \geq 2}$.

%%%%%%%%%%%%%%%%%%%%%%%%%%%%%%%%%%%%%%
\section{Analysis of Regge-Wheeler Type Equations}

The remainder of this work is concerned with the higher angular modes, encoded by $\delta g^{\ell \geq 2}$.  For each portion of $\delta g^{\ell \geq 2}$, we decouple gauge-invariant quantities satisfying Regge-Wheeler type equations, the analysis of which is expected to provide an avenue towards proving decay of the solution.

To eliminate redundancy, we present in this section a general theory for the analysis of such equations, specialized as necessary in subsequent sections.  We consider solutions within the sub-bundle $\mathcal{L}(-2)$ of symmetric traceless two-tensors on the spheres of symmetry, with such solutions either being divergence-free or possessing scalar potentials or divergence-free co-vector potentials.  In this language, we define a solution of a Regge-Wheeler type equation as follows.

\begin{definition}\label{RWdefinition}
Let $\Psi$ be a symmetric traceless two-tensor, regarded as a section of $\mathcal{L}(-2)$.  We say that $\Psi$ is a solution of a Regge-Wheeler type equation with potential $V$ if $\Psi$ satisfies
\begin{equation}\label{RWequation}
\slashed{\Box}_{\mathcal{L}(-2)} \Psi = V \Psi.
\end{equation}
We further assume that $V$ is a radial function, bounded on the exterior region, and comparable to $\frac{1}{r^2}$ at spatial infinity.
\end{definition}

We remark that $V$ need not be non-negative.

%%%%%%%%%%%%%%%%%%%%%%%%%%%%%%%%%%%%%%
\subsection{Stress-Energy Tensors}

We consider the natural stress-energy tensor
\begin{equation}\label{stressTensor}
T_{ab}[\Psi] := \slashed{\nabla}_a \Psi\cdot \slashed{\nabla}_b \Psi - \frac{1}{2}g_{ab}((\slashed{\nabla}\Psi)^2 + V|\Psi|^2),
\end{equation}
where we emphasize that
\begin{align*}
&\slashed{\nabla}_{a}\Psi\cdot\slashed{\nabla}_{b}\Psi = g^{\alpha\beta}g^{\gamma\delta}(\slashed\nabla_{a}\Psi)_{\alpha\gamma}(\slashed\nabla_{b}\Psi)_{\beta\delta},\\
&|\Psi|^2 = g^{\alpha\beta}g^{\gamma\delta}\Psi_{\alpha\gamma}\Psi_{\beta\delta},
\end{align*}
and
\begin{align*}
(\slashed\nabla\Psi)^2 &= g^{ab}g^{\alpha\beta}g^{\gamma\delta}(\slashed\nabla_{a}\Psi)_{\alpha\gamma}(\slashed\nabla_{b}\Psi)_{\beta\delta},\\
&= (\tilde{\slashed\nabla}\Psi)^2  + r^{-2}|\mathring{\slashed\nabla}\Psi|^2,\\
(\tilde{\slashed\nabla}\Psi)^2 &= g^{AB}g^{\alpha\beta}g^{\gamma\delta}(\slashed\nabla_{A}\Psi)_{\alpha\gamma}(\slashed\nabla_{B}\Psi)_{\beta\delta},\\
|\mathring{\slashed\nabla}\Psi|^2 &= \mathring\sigma^{\eta\nu}g^{\alpha\beta}g^{\gamma\delta}(\slashed\nabla_{\eta}\Psi)_{\alpha\gamma}(\slashed\nabla_{\nu}\Psi)_{\beta\delta}.
\end{align*}

In addition, we will use the virtual stress-energy tensor
\begin{equation}\label{virtualStressTensor}
\check{T}_{ab}[\Psi] := \slashed{\nabla}_a \Psi\cdot \slashed{\nabla}_b \Psi - \frac{1}{2}g_{ab}(\slashed{\nabla}^{c}\Psi\cdot\slashed{\nabla}_{c}\Psi).
\end{equation}

Estimates are obtained by contracting with a vector-field multiplier $X^{b}$ and applying the spacetime Stokes' theorem
\begin{equation}
\int_{\partial \mathcal{D}} T_{ab}[\Psi]n_{\partial \mathcal{D}}^{a}X^{b} = \int_{\mathcal{D}} \nabla^{a}(T_{ab}[\Psi]X^{b}),
\end{equation}
over a spacetime region $\mathcal{D}$ with boundary $\partial \mathcal{D}$.  A similar identity holds for the virtual stress tensor \eqref{virtualStressTensor}.

We remark that the natural stress-energy tensor \eqref{stressTensor} has non-trivial divergence
\begin{equation}\label{stressDivergence}
\nabla^{a}T_{ab}[\Psi] = -\frac{1}{2}\nabla_{b}V|\Psi|^2 + \slashed{\nabla}^{a}\Psi[\slashed{\nabla}_{a},\slashed{\nabla}_{b}]\Psi,
\end{equation}
with the commutator $[\slashed{\nabla}_{a},\slashed{\nabla}_{b}]$ vanishing when contracted with a multiplier invariant under the angular Killing fields.  In particular, all such multipliers considered in the analysis below have this property.

Although the virtual stress-energy tensor satisfies a positive energy condition, while this is not necessarily true for the natural stress-energy tensor.  Indeed, it will often be the case that the stress-energy tensor fails to satisfy a pointwise positivity condition, owing to the potential $V$, at which point it becomes necessary to incorporate integral estimates. 

%%%%%%%%%%%%%%%%%%%%%%%%%%%%%%%%%%%%
\subsection{Comparisons}

Given $A$ and $B$, we use the notation
\[ A \approx B \]
if the two quantities are comparable up to constants depending upon the orbit sphere dimension $n$ and the mass $M$.  That is, there exists $C(n,M)$ such that
\[ \frac{1}{C(n,M)}A \leq B \leq C(n,M)A.\]
Likewise, we use 
\[ A \gtrsim B \]
for one-sided comparisons up to constants with such dependence.

%%%%%%%%%%%%%%%%%%%%%%%%%%%%%%%%%%%%%
\subsection{Decay Foliation}\label{decayFoliation}
Recalling the Eddington-Finkelstein double null coordinates
\begin{align*}
u &= \frac{1}{2}(t - r_{*}),\\
v &= \frac{1}{2}(t+r_{*}),
\end{align*}
we further define
\begin{align}
\begin{split}
L &:= \frac{\partial}{\partial v} = \frac{\partial}{\partial t} + \frac{\partial}{\partial r_{*}},\\
\underline{L} &:= \frac{\partial}{\partial u} = \frac{\partial}{\partial t} - \frac{\partial}{\partial r_{*}}.
\end{split}
\end{align}
Fixing radii $r_1$ and $R_1$ satisfying $r_{h} < r_1 < r_{P} < R_1$, we choose a Lipschitz hypersurface $\Sigma_0$ with the following properties:
\begin{itemize}
\item $\Sigma_0$ intersects the future horizon $\mathcal{H}^+$ transversely and $\Sigma_0$ is spacelike in $r_{h}\leq r\leq R_1$,
\item $\Sigma_0\cap \{r_1\leq r\leq R_1\}=\{t=0\}\cap \{r_1\leq r\leq R_1\}$,
\item $\Sigma_0\cap \{R_1\leq r\}=\{u=-\frac{1}{2}\left(R_1\right)_{*}\}\cap \{R_1\leq r\}$,
\end{itemize}
where $(R_1)_*$ is given by \eqref{RWcoordinate} with normalization \eqref{RWnormalization}.

Flowing along the static Killing vector field $T$, we construct our decay foliation $\Sigma_{\tau} := \phi_{\tau}(\Sigma_0)$.  We define  $n_{\Sigma_\tau}$ as the unit timelike normal vector for $\Sigma_\tau$ for $r\leq R_1$ and $n_{\Sigma_\tau}:=L$ for $r> R_1$.  Further, we denote the volume form on $\Sigma_\tau$ corresponding to $n_{\Sigma_\tau}$ by $dVol_{\Sigma_\tau}$, and the volume form on the unit round $n$-sphere $S^n$ by $dVol_{S^n}$. 

We denote by $D(\tau_1, \tau_2)$ the spacetime region between the hypersurfaces $\Sigma_{\tau_1}$ and $\Sigma_{\tau_2}$, with $\tau_1 \leq \tau_2$.  To be more precise, we define
\begin{equation}
D(\tau_1, \tau_2) := J^{+}(\Sigma_{\tau_1}) \cap J^{-}(\Sigma_{\tau_2}),
\end{equation}
with volume form $dVol$.  We denote the null hypersurface bounding $D(\tau_1, \tau_2)$ at the future event horizon by
\begin{equation}
\mathcal{H}^{+}(\tau_1, \tau_2) := D(\tau_1, \tau_2) \cap \mathcal{H}^{+},
\end{equation}
with null tangential chosen to be $n_{\mathcal{H}^+} := T$ and the associated volume form denoted $dVol_{\mathcal{H}^{+}}$.
In addition, we define the null hypersurfaces
\begin{equation}
C_{v_0}(u_1, u_2) := \{ v = v_0, u_1 \leq u \leq u_2 \},
\end{equation}
with the choice of null tangential $\underline{L}^{a}$ and the associated volume form $r^n du dVol_{S^n}$.  Intersecting the constant $v$-hypersurface with the spacetime region $D(\tau_1, \tau_2)$, we find
\[ C_{v}(-\infty, \infty) \cap D(\tau_1, \tau_2) = C_{v}(\tau_1 - \frac{1}{2}\left(R_1\right)_*, \tau_2 - \frac{1}{2}\left(R_1\right)_*),\]
and define the limit of such hypersurfaces
\begin{equation}
\mathcal{I}^{+}(\tau_1, \tau_2) := \lim_{v \to \infty} C_{v}(\tau_1 - \frac{1}{2}\left(R_1\right)_*, \tau_2 - \frac{1}{2}\left(R_1\right)_*).
\end{equation}

Often we write boundary integrals of the form
\[ \int_{\mathcal{I}^{+}(\tau_1, \tau_2)} J^{X}_{a}[\Psi]\underline{L}^{a} \left(r^n du dVol_{S^n}\right),\]
where it is understood that we are evaluating the limit as $v$ approaches infinity of the boundary integrals
\[ \int_{C_{v}(\tau_1 - \frac{1}{2}\left(R_1\right)_*, \tau_2 - \frac{1}{2}\left(R_1\right)_*)} J^{X}_{a}[\Psi]\underline{L}^{a} \left(r^n du dVol_{S^n}\right).\]

%%%%%%%%%%%%%%%%%%%%%%%%%%%%%%%%%%%%%%
\subsection{Poincar\'{e} Inequalities}

The spherical Laplacian operator acts as an endomorphism on each of the aforementioned sub-bundles of $\mathcal{L}(-2)$, with spectra described in (\ref{scalarSpectraThree},\ \ref{coVectorSpectraTwo},\ \ref{twoTensorSpectra}).  

We assume that, in acting on the sub-bundle associated with $\Psi$, the spherical Laplacian $\mathring\Delta$ has least eigenvalue $\lambda$.  The identity
\begin{equation}
\mathring\Delta |\Psi|^2 = 2\mathring{\slashed\Delta}\Psi \cdot \Psi + 2 |\mathring{\slashed\nabla} \Psi|^2
\end{equation}
yields the Poincar\'{e} inequality
\begin{equation}\label{PoincarePsi}
\int_{S^n} |\mathring{\slashed\nabla}\Psi|^2 \geq \lambda\int_{S^n}|\Psi|^2.
\end{equation}

For those sections of $\mathcal{L}(-2)$ with scalar potential, we have $\lambda = 2$, whereas those with divergence-free co-vector potential have $\lambda = n$.   Finally, on the sub-bundle of divergence-free symmetric traceless two-tensors we have $\lambda = 2n$.

%%%%%%%%%%%%%%%%%%%%%%%%%%%%%%%%%%%%%%
\subsection{Hardy Inequalities}

We adapt the Hardy inequalities found in the work of Andersson-Blue \cite{AnderssonBlue} and the earlier work of Blue-Soffer \cite{BlueSoffer}.
\begin{lemma}\label{HardyEstimate}
Suppose $A$ and $W$ are smooth functions on $[s_0, s_1]$,  with $A$ non-negative.  Further, assume that the ODE
\begin{equation}
-\frac{d}{ds}\left(A\frac{dg}{ds}\right) + Wg = 0
\end{equation}
has a smooth, positive solution $g$ on $[s_0,s_1]$.   Given any smooth function $f$, as long as
\[ f^2A\frac{d}{ds}( \log g)\]
vanishes at $s_0$ and $s_1$, we have the estimate
\begin{equation}
\int^{s_1}_{s_0} \left[A\left(\frac{df}{ds}\right)^2 + Wf^2 \right]ds \geq 0.
\end{equation}
\end{lemma}

\begin{proof}
Let $h=f/g$. Then
\begin{align*}
&\int_{s_0}^{s_1} \left[A\left(\frac{df}{ds}\right)^2+Wf^2\right] ds\\
=&\int_{s_0}^{s_1} \left[Ag^2\left(\frac{d h}{ds}\right)^2+Ah^2\left(\frac{d g}{ds}\right)^2+2Agh\frac{d g}{ds}\frac{d h}{ds}+Wg^2h^2\right]  ds\\
=&\int_{s_0}^{s_1} \Bigg[Ag^2 \left(\frac{d h}{ds}\right)^2+Ah^2 \left(\frac{d g}{ds}\right)^2+ Wg^2h^2+\frac{d}{ds}\left( Ah^2g\frac{dg}{ds} \right)\\
&-gh^2\frac{d}{ds}\left( A\frac{dg}{ds} \right)-Ah^2\left(\frac{dg}{ds}\right)^2\Bigg]  ds\\
=&\int_{s_0}^{s_1} Ag^2\left(\frac{dh}{ds}\right)^2 ds+\left(Ah^2g\frac{dg}{ds}\right)\Bigg|^{s_1}_{s_0}\\
=&\int_{s_0}^{s_1} Ag^2\left(\frac{dh}{ds}\right)^2 ds\geq 0.\\
\end{align*}
\end{proof}

Often we will assume that $s_0 = 2M$ and $s_1 = \infty$, with $A$ vanishing at this boundary and $f$ compactly supported.  In applying the estimate, the primary difficulty lies in finding a positive solution to the associated ODE.  To this end, we transform the ODE into a hypergeometric form, from which a well-known positive solution can be constructed.

%%%%%%%%%%%%%%%%%%%%%%%%%%%%%%%%%%%%%%
\subsection{Hypergeometric Differential Equations}

We consider the hypergeometric ODE
\begin{equation}\label{hyperGeoODE}
(1-z)z\frac{d^2\tilde{g}}{dz^2}+\Big( c-(a+b+1) \Big)\frac{d \tilde{g}}{dz}-ab \tilde{g}=0,
\end{equation}
with $z < 1$ and $0 < b < c$.  With such constraints, the hypergeometric function $F(a,b;c;z)$, defined by
\begin{equation}\label{hyperGeoFunction}
F(a,b;c;z)=\frac{\Gamma(c)}{\Gamma(a)\Gamma(b)}\int_0^1 t^{b-1}(1-t)^{c-b-1}(1-zt)^{-a} dt,
\end{equation}
is a solution of \eqref{hyperGeoODE}.  We can easily divide out the terms involving the Gamma function to obtain a positive solution of \eqref{hyperGeoODE}.

%%%%%%%%%%%%%%%%%%%%%%%%%%%%%%%%%%%%%%
\subsection{The $T$-Energy}

\subsubsection{Definition and Monotonicity}
We define the $T$-current
\begin{equation}
J^{T}_{a}[\Psi] := T_{ab}[\Psi]T^{b},
\end{equation}
and the $T$-energy
\begin{align}\label{Tenergy}
\begin{split}
E^{T}_{\Psi}(\Sigma_{\tau}) &:= \int_{\Sigma_{\tau}} J^{T}_{a}[\Psi]n_{\Sigma_\tau}^{a}dVol_{\Sigma_\tau}\\
&=\int_{\Sigma_{\tau}} T_{ab}[\Psi]n_{\Sigma_\tau}^{a}T^{b}dVol_{\Sigma_\tau}.
\end{split}
\end{align}

Applying the static Killing multiplier $T$ over the spacetime region bounded by the decay foliation hypersurfaces $\Sigma_{\tau_1}$ and $\Sigma_{\tau_2}$ , where $0 \leq \tau_1 < \tau_2$, we have
\begin{align*}
&E^{T}_{\Psi}(\Sigma_{\tau_2}) + \int_{\mathcal{H}^{+}(\tau_1,\tau_2)} J^{T}_{a}[\Psi]n_{\mathcal{H}^{+}}^{a}dVol_{\mathcal{H}^{+}}\\
&+ \int_{\mathcal{I}^{+}(\tau_1,\tau_2)} J^{T}_{a}[\Psi]\underline{L}^{a}(r^ndudVol_{S^n}) =  E^{T}_{\Psi}(\Sigma_{\tau_1}),
\end{align*}
with the terms of the density
\[ \Div J^{T}[\Psi] = (\nabla^{a}T_{ab}[\Psi])T^{b} + T_{ab}[\Psi]\nabla^{(a}T^{b)}\]
vanishing according to \eqref{stressDivergence} and Killing condition on $T$, 
\[(^{T}\pi)^{ab} :=   \mathcal{L}_{T}g_{M}^{ab} = 2\nabla^{(a}T^{b)} = 0.\]

Non-negativity of the boundary integral along the event horizon follows from $T$ being null tangential, whereas non-negativity of the boundary integral along future null infinity follows from the radial decay of the potential $V$, such that
\begin{align}
\begin{split}
&\int_{\mathcal{H}^{+}(\tau_1,\tau_2)} J^{T}_{a}[\Psi]n_{\mathcal{H}^{+}}^{a}dVol_{\mathcal{H}^{+}} \geq 0,\\
&\int_{\mathcal{I}^{+}(\tau_1,\tau_2)} J^{T}_{a}[\Psi]\underline{L}^{a}(r^ndudVol_{S^n}) \geq 0.
\end{split}
\end{align}

Together, the two imply monotonicity of the $T$-energy:
\begin{equation}\label{Tmonotonicity}
E_{\Psi}^{T}(\Sigma_{\tau_2}) \leq E_{\Psi}^{T}(\Sigma_{\tau_1}).
\end{equation}

Before proceeding, we define the virtual $T$-energy:
\begin{equation}\label{virtualTenergy}
\check{E}^{T}_{\Psi}(\Sigma_{\tau}) := \int_{\Sigma_{\tau}} \check{T}_{ab}[\Psi]n_{\Sigma_\tau}^{a}T^{b}.
\end{equation}

\subsubsection{The Adapted Hardy Estimate}
The natural stress-energy tensors under consideration often fail to satisfy pointwise positive energy conditions; as a consequence, the $T$-energies above are not obviously positive-definite.  To address this issue, we rely upon the Poincar\'{e} inequality \eqref{PoincarePsi} and an adapted Hardy estimate, in the spirit of Lemma \ref{HardyEstimate}.

To begin, we regard $R = r$ as a function on the hypersurface $\Sigma_{\tau}$, and consider the coordinate system $(\tau, R, x^{A})$.  Written in this form, the integrand of the virtual $T$-energy \eqref{virtualTenergy} takes the form
\begin{align*}
&\check{T}_{ab}[\Phi]n^a_{\Sigma_\tau}T^b dVol_{\Sigma_\tau}\\
=&\frac{1}{2}\left( (1-\mu)^{-1}\cosh^{-2}x|\nablas_{\tau}\Phi|^2+(1-\mu)|\nablas_R\Phi|^2+R^{-2}|\mathring{\nablas}\Phi|^2 \right)R^n dRdVol_{S^n},
\end{align*}
where $x(r)$ is specified by
\[(1-\mu)^{-1/2}\cosh x=-\left\langle n_{\Sigma_\tau},T \right\rangle,\] 
as follows:

By straightforward computation,
\begin{align*}
\frac{\partial}{\partial \tau}&=\frac{\partial}{\partial t},\\
\frac{\partial}{\partial R}&=\frac{\partial}{\partial r}+(1-\mu)^{-1}\tanh x \frac{\partial}{\partial t}.
\end{align*}
Therefore $g_{ab}$ and $g^{ab}$ under the $\tau,R$ coordinate are of the form
\begin{align*}
g_{ab}=
\left[ {\begin{array}{ccc}
   -(1-\mu)  & -\tanh x & 0 \\
   -\tanh x & (1-\mu)^{-1}\cosh^{-2} x & 0\\
   0 & 0 & R^2\mathring{\sigma}_{\alpha\beta}
  \end{array} } \right],
\end{align*}
\begin{align*}
g^{ab}=
\left[ {\begin{array}{ccc}
   -(1-\mu)^{-1}\cosh^{-2} x  & -\tanh x & 0 \\
   -\tanh x & (1-\mu) & 0\\
   0 & 0 & R^{-2}\mathring{\sigma}^{\alpha\beta}
  \end{array} } \right].
\end{align*}
In particular,
\begin{align*}
n_{\Sigma_\tau}=(1-\mu)^{-1/2}\cosh^{-1}x\frac{\partial}{\partial\tau}+(1-\mu)^{1/2}\sinh x\frac{\partial}{\partial R},
\end{align*}
\begin{align*}
dVol_{\Sigma_\tau}=(1-\mu)^{-1/2}\cosh^{-1}x R^n dRdVol_{S^n},
\end{align*}
and
\begin{align*}
(\nablas\Phi)^2=-(1-\mu)^{-1}\cosh^{-2}x|\nablas_\tau\Phi|^2+(1-\mu)|\nablas_R\Phi|^2-2\tanh x\nablas_\tau\Phi\cdot\nablas_R\Phi+R^{-2}|\mathring{\nablas}\Phi|^2.
\end{align*}
Hence,
\begin{align*}
\check{T}_{ab}[\Phi]n^a_{\Sigma_\tau}T^b =&\check{T}_{ab}[\Phi] \left((1-\mu)^{-1/2}\cosh^{-1} x\frac{\partial}{\partial\tau}+(1-\mu)^{1/2}\sinh x\frac{\partial}{\partial R}\right)^a  \left(\frac{\partial}{\partial \tau}\right)^b\\
&=\left( (1-\mu)^{-1/2}\cosh^{-1} x|\nablas_\tau\Phi|^2+(1-\mu)^{1/2}\sinh x\nablas_\tau\Phi\cdot \nablas_R\Phi \right)\\
&-\frac{1}{2}\Big( -(1-\mu)^{-1}\cosh^{-2}x|\nablas_\tau\Phi|^2+(1-\mu)|\nablas_R\Phi|^2\\
&-2\tanh x\nablas_\tau\Phi\cdot\nablas_R\Phi+R^{-2}|\mathring{\nablas}\Phi|^2 \Big)\left(-(1-\mu)^{1/2}\cosh x\right)\\
&=\frac{1}{2}(1-\mu)^{-1/2}\cosh^{-1}x |\nablas_\tau\Phi|^2 +\frac{1}{2}(1-\mu)^{3/2}\cosh x |\nablas_R\Phi|^2\\
&+\frac{1}{2}(1-\mu)^{1/2}\cosh x R^{-2}|\mathring{\nablas}\Phi|^2.
\end{align*}

Incorporating the potential $V$, the $T$-energy \eqref{Tenergy} has the form
\begin{align*}
E_{\Psi}^T(\Sigma_{\tau})&=\frac{1}{2}\int_{S^n}\int_{r_{h}}^\infty \left[(1-\mu)|\slashed{\nabla}_{R}\Psi|^2+V|\Psi|^2\right]R^{n}drdVol_{S^n}\\
&+\frac{1}{2}\int_{S^n}\int_{r_{h}}^\infty \left[\cosh^{-2}x(1-\mu)^{-1}|\nablas_{\tau}\Phi|^2+R^{-2}|\mathring{\nablas}\Phi|^2\right]  R^{n}dr dVol_{S^n}.
\end{align*}

With the change of variable $s = R^{n-1}$, the radial coefficient naturally takes the form $A = s(s-2M)$, in the notation of Lemma \ref{HardyEstimate}.  Further choosing $s_0 = 2M$ and $s_1 = \infty$, the following Hardy estimate holds for a large class of potentials $V$:

\begin{lemma}\label{HardyTEstimate}
Let $f(s)$ be a function defined on $[2M,\infty)$ and $E,F\geq 0$ be two nonnegative numbers with $|2F-E|\leq 1$ and $E>1$.  Then
\begin{align}\label{2}
\int_{2M}^\infty \left[A\left(\frac{d f}{ds}\right)^2+V(E,F)f^2 \right] ds\geq 0,
\end{align}
where
\begin{align*}
A &= s(s-2M),\\
V(E,F) &=\frac{1}{4}(E^2-1)-\frac{2M F^2}{s}.
\end{align*}
\end{lemma}

\begin{proof}
We first assume $f(s)$ has compact support in $[2M,\infty)$. From Lemma \ref{HardyEstimate}, the estimate follows from the existence a positive function $g(s)$ defined on $[2M,\infty)$ and satisfying the equation
\begin{align}\label{1}
-\frac{d}{ds}\left(A \frac{d g}{ds} \right)+V(E,F)g=0.
\end{align}
One can obtain an explicit positive solution of (\ref{1}) by using hypergeometric functions.  Letting $\alpha=\pm F$ and $\tilde{g}=s^{-\alpha}g$, the equation \eqref{1} becomes
\begin{align*}
-s(s-2M)\frac{d^2\tilde{g}}{ds^2}+\Big(-2(s-2M)\alpha-(2s-2M)\Big)\frac{d\tilde{g}}{ds}-\frac{1}{4}\Big((2\alpha+1)^2-E^2\Big)\tilde{g}(z)=0.
\end{align*}
Performing the change of variable $z=1-\frac{s}{2M}$, we have
\begin{align}\label{re}
(1-z)z\frac{d^2\tilde{g}}{dz^2}+\Big( 1-(2\alpha+2) \Big)\frac{d \tilde{g}}{dz}-\frac{1}{4}\Big((2\alpha+1)^2-E^2\Big)\tilde{g}(z)=0.
\end{align}
Comparing \eqref{re} to hypergeometric ODE \eqref{hyperGeoODE}, we have a hypergeometric equation with $z \leq 0$, $c=1$, and $\{a,b\}=\frac{1}{2}\Big(1+2\alpha\pm E\Big)=\frac{1}{2}\Big(1\pm 2F \pm E\Big)$.  The associated hypergeometric function $F(a,b;c;z)$ has integral representation \eqref{hyperGeoFunction} assuming
\begin{align*}
&0<\frac{1}{2}\Big(1+2F+E\Big)<1\ \textup{or}\ 0<\frac{1}{2}\Big(1-2F+E\Big)<1 \\
\textup{or}\ &0<\frac{1}{2}\Big(1+2F-E\Big)<1\ \textup{or}\ 0<\frac{1}{2}\Big(1-2F-E\Big)<1.
\end{align*}
In particular, when $E,F\geq 0$, the above condition is equivalent to $|2F-E|<1$.\\
We approximate a general $f(s)$ with compactly supported functions.  Without loss of generality, we assume
\begin{align*}
\int_{2M}^\infty \left[s^2\left(\frac{d f}{ds}\right)^2+f^2\right] ds<\infty,
\end{align*}
as the left hand side of \eqref{2} is infinity otherwise.  For any large number $s_0>>2M$ we consider the cut-off function $\eta$ such that
\begin{align*}
\eta(s)=\left\{\begin{array}{cc}
1 & s\leq s_0\\
0 & s\geq 2es_0,
\end{array}\right.
\end{align*}
and $|\eta'(s)|\leq s^{-1}$. Let  $f_1=\eta f$ and $f_2=(1-\eta)f$. Then
\begin{align*}
\int_{2M}^\infty \left[A\left(\frac{df}{ds}\right)^2+Vf^2\right]ds=&\int_{2M}^\infty \left[A\left(\frac{df_1}{ds}\right)^2+Vf_1^2\right]ds+\int_{2M}^\infty \left[A\left(\frac{df_2}{ds}\right)^2+Vf_2^2\right]ds\\
                                                     +&2\int_{2M}^\infty \left[A\left(\frac{df_1}{ds}\right)\left(\frac{df_2}{ds}\right)+Vf_1f_2\right]ds.
\end{align*}
The first line is non-negative from the previous discussion.  The integrand of the second line is non-vanishing only in $[s_0,2es_0]$.  Hence
\begin{align*}
\left|\int_{2M}^\infty \left[A\left(\frac{df_1}{ds}\right)\left(\frac{df_2}{ds}\right)+Vf_1f_2\right] ds\right|\lesssim \int_{s_0}^{2es_0} \left[s^2\left(\frac{df}{ds}\right)^2+f^2\right] ds,
\end{align*}
which goes to zero as $s_0$ goes to infinity by the dominated convergence theorem.
\end{proof}

\subsubsection{The $T$-Energy Comparison}\label{TStrategy}

We briefly describe a typical application of the Hardy estimate above.  First, we borrow from the angular term using the Poincar\'{e} inequality \eqref{PoincarePsi}, and find the underestimate
\begin{align*}
E_{\Psi}^T(\Sigma_{\tau})&=\frac{1}{2}\int_{S^n}\int_{r_{h}}^\infty \left[(1-\mu)|\slashed{\nabla}_{R}\Psi|^2+V|\Psi|^2\right]R^{n}drdVol_{S^n}\\
&+\frac{1}{2}\int_{S^n}\int_{r_{h}}^\infty \left[\cosh^{-2}x(1-\mu)^{-1}|\nablas_{\tau}\Phi|^2+R^{-2}|\mathring{\nablas}\Phi|^2\right]  R^{n}dr dVol_{S^n}\\
&\geq \frac{1}{2}\int_{S^n}\int_{r_{h}}^\infty \left[(1-\mu)|\slashed{\nabla}_{R}\Psi|^2+\left(V + (\lambda - \delta)R^{-2}\right)|\Psi|^2\right]R^{n}drdVol_{S^n}\\
&+\frac{1}{2}\int_{S^n}\int_{r_{h}}^\infty \left[\cosh^{-2}x(1-\mu)^{-1}|\nablas_{\tau}\Phi|^2+\delta R^{-2}|\mathring{\nablas}\Phi|^2\right]  R^{n}dr dVol_{S^n},
\end{align*}
where $\delta > 0$ is a small residual angular coefficient.

Isolating the term
\[\int_{r_{h}}^\infty \left[(1-\mu)|\slashed{\nabla}_{R}\Psi|^2+\left(V + (\lambda - \delta)R^{-2}\right)|\Psi|^2\right]R^{n}dr,\]
we perform the change of variables $s = R^{n-1}$ and determine a $V(E,F)$ in Lemma \ref{HardyTEstimate} underestimating the $s$-dependent potential associated with $\left(V + (\lambda - \delta)R^{-2}\right)$ pointwise on the exterior region.  The lemma then implies
 \[\int_{r_{h}}^\infty \left[(1-\mu)|\slashed{\nabla}_{R}\Psi|^2+\left(V + (\lambda - \delta)R^{-2}\right)|\Psi|^2\right]R^{n}dr \geq 0,\]
 and we can choose $\epsilon(\delta) > 0$ such that
 \begin{align*}
&E_{\Psi}^T(\Sigma_{\tau})\\
 &\geq \frac{1}{2}\int_{S^n}\int_{r_{h}}^\infty \left[(1-\mu)|\slashed{\nabla}_{R}\Psi|^2+\left(V + (\lambda - \delta)R^{-2}\right)|\Psi|^2\right]R^{n}drdVol_{S^n}\\
&+\frac{1}{2}\int_{S^n}\int_{r_{h}}^\infty \left[\cosh^{-2}x (1-\mu)^{-1}|\nablas_{\tau}\Phi|^2+\delta R^{-2}|\mathring{\nablas}\Phi|^2\right]  R^{n}dr dVol_{S^n}\\
&\geq \frac{1}{2}\int_{S^n}\int_{r_{h}}^\infty \epsilon \left[(1-\mu)|\slashed{\nabla}_{R}\Psi|^2+\left(V + (\lambda - \delta)R^{-2}\right)|\Psi|^2\right]R^{n}drdVol_{S^n}\\
&+\frac{1}{2}\int_{S^n}\int_{r_{h}}^\infty \left[\cosh^{-2} x(1-\mu)^{-1}|\nablas_{\tau}\Phi|^2+\delta R^{-2}|\mathring{\nablas}\Phi|^2\right]  R^{n}dr dVol_{S^n}\\
&\gtrsim \check{E}_{\Psi}^{T}(\Sigma_\tau),
\end{align*}
with $\epsilon(\delta)$ chosen small enough to absorb possible negative contributions from $\left(V + (\lambda - \delta)R^{-2}\right)$ into the $\delta$-small residual angular term, again via application of \eqref{PoincarePsi}.

We remark that, owing to the characteristics of the potential $V$ and the Poincar\'{e} inequality \eqref{PoincarePsi}, we have the other direction in the comparison
\[ \check{E}^{T}_{\Psi}(\Sigma_{\tau}) \gtrsim E^{T}_{\Psi}(\Sigma_\tau).\]

In what follows, we assume the $T$-energy comparison
\begin{equation}\label{Tcomparison}
\check{E}^{T}_{\Psi}(\Sigma_{\tau}) \approx E^{T}_{\Psi}(\Sigma_\tau).
\end{equation}

Owing to the spacetime Stokes' theorem, the comparison \eqref{Tcomparison} implies
\begin{align}\label{boundaryTEstimates}
\begin{split}
&\int_{\mathcal{H}^{+}(\tau_1,\tau_2)} J^{T}_{a}[\Psi]n_{\mathcal{H}^{+}}^{a}dVol_{\mathcal{H}^{+}} \lesssim \check{E}^{T}_{\Psi}(\Sigma_{\tau_1}) ,\\
&\int_{\mathcal{I}^{+}(\tau_1,\tau_2)} J^{T}_{a}[\Psi]\underline{L}^{a}(r^ndudVol_{S^n}) \lesssim \check{E}^{T}_{\Psi}(\Sigma_{\tau_1}) .
\end{split}
\end{align}

%%%%%%%%%%%%%%%%%%%%%%%%%%%%%%%%%%%%%%
\subsection{The $N$-Energy}

We describe the red-shift vector field, introduced in Dafermos-Rodnianski \cite{DR} in four spacetime dimensions and generalized to the higher-dimensional setting by Schlue \cite{Schlue}.

For convenience, we calculate in the inward-directed Eddington-Finkelstein coordinates
\begin{align*}
\bar{v}&=t+r_*,\\
R&=r,
\end{align*}
in which the background metric takes the form
\begin{align*}
-(1-\mu)d\bar{v}^2+2d\bar{v}dR+R^2\mathring\sigma_{\alpha\beta}dx^{\alpha}dx^{\beta}.
\end{align*}

Let $Y$ be a smooth vector, specified by
\begin{align*}
 Y|_{r=r_{h}}&=-\frac{\partial}{\partial R},\\
 \mathcal{L}_T Y&= 0,\\
 \mathcal{L}_{\partial_{x^\alpha}}Y &=0,\\
 \left(\nabla_Y Y\right)_{r=r_{h}}&=-\sigma (Y+T),
\end{align*}
where $\sigma$ is a positive number to be determined. At the event horizon $R = r_{h}$ we compute
\begin{align*}
\nabla_{\bar{v}} Y&=\nabla_Y\partial_{\bar{v}}=\frac{(n-1)M}{R^{n}}\partial_R = \kappa_{n}\partial_{R},\\
\nabla_{\partial_R} Y&=\sigma (-\partial_R+\partial_{\bar{v}}),\\
\nabla_{\partial_{x^\alpha}} Y&=-\frac{1}{R}\partial_{x^\alpha},
\end{align*}
where $\kappa_n =\frac{(n-1)M}{r_{h}^{n}}$ is the surface gravity.  From these calculations we deduce
\[\nabla_{a}Y^a=-\sigma-nR^{-1}\]
and 
\begin{align*}
\nabla^{(a}Y^{b)}=
\left[ {\begin{array}{ccc}
   \sigma & -\frac{\sigma}{2} & 0 \\
   -\frac{\sigma}{2} & \kappa_n & 0\\
   0 & 0 & -\frac{1}{R^3} \mathring{\sigma}^{\alpha\beta}
  \end{array} } \right].
\end{align*}

Contracting with the virtual stress-energy tensor, we estimate
\begin{align*}
\check{T}_{ab}[\Psi]\nabla^{(a}Y^{b)}&=\sigma |\slashed{\nabla}_{\bar{v}}\Psi|^2+\kappa_n |\slashed{\nabla}_{R}\Psi|^2-\sigma \slashed{\nabla}_{\bar{v}}\Psi \cdot\slashed{\nabla}_{R}\Psi-\frac{1}{R^3}|\mathring{\slashed{\nabla}}\Psi|^2\\
                                  &-\frac{1}{2}\left( 2\slashed{\nabla}_{\bar{v}}\Psi \cdot\slashed{\nabla}_{R}\Psi+\frac{1}{R^2}|\mathring{\slashed{\nabla}}\Psi|^2 \right)(-\sigma-nR^{-1})\\
                                  &=\sigma |\slashed{\nabla}_{\bar{v}}\Psi|^2+\kappa_n |\slashed{\nabla}_{R}\Psi|^2 +\frac{n}{R}\slashed{\nabla}_{\bar{v}}\Psi\cdot \slashed{\nabla}_{R}\Psi\\
                                  &+\frac{1}{R^3}\left( \frac{\sigma R}{2}+\frac{n}{2}-1 \right)|\mathring{\slashed{\nabla}}\Psi|^2\\
                                  &\gtrsim \left( \sigma |\slashed{\nabla}_{\bar{v}}\Psi|^2+\kappa_n |\slashed{\nabla}_{R}\Psi|^2+\frac{\sigma}{R^2}|\mathring{\slashed{\nabla}}\Psi|^2 \right),
\end{align*}
and
\begin{align}\label{coerciveEstimate}
\begin{split}
&\nabla^a (\check{T}_{ab}[\Psi]Y^b)=(\nabla^a \check{T}_{ab}[\Psi]) Y^b+\check{T}_{ab}\nabla^{(a}Y^{b)}\\ 
                               &\gtrsim \left( \sigma |\slashed{\nabla}_{\bar{v}}\Psi|^2+\kappa_n |\slashed{\nabla}_{R}\Psi|^2+\frac{\sigma}{R^2}|\mathring{\slashed{\nabla}}\Psi|^2 \right)+V\nabla_Y \Psi\cdot \Psi\\
                               &\gtrsim \left( \sigma |\slashed{\nabla}_{\bar{v}}\Psi|^2+\kappa_n |\slashed{\nabla}_{R}\Psi|^2+\frac{\sigma}{R^2}|\mathring{\slashed{\nabla}}\Psi|^2 \right)
\end{split}
\end{align}
using the Cauchy-Schwarz and Poincar\'{e} inequalities and a choice of large $\sigma$.  Note that such a choice can be made for any bounded radial potential $V$.  

We extend $Y$ to the exterior region such that $Y$ is causal and $Y=0$ for $r\geq r_1$.  By continuity, the coercive estimate \eqref{coerciveEstimate} is satisfied in $r_{h}\leq r\leq r_0$ for some $r_0<r_1$. 

We define the strictly timelike red-shift multiplier $N:=T+Y$, in addition to the current
\begin{align}\label{Ncurrent}
\begin{split}
J^{N}_{a}[\Psi] &:=T_{ab}[\Psi]T^{b} + \check{T}_{ab}[\Psi]Y^{b},
\end{split}
\end{align}
and energies
\begin{align}\label{Nenergies}
\begin{split}
E^{N}_{\Psi}(\Sigma_{\tau}) &:= \int_{\Sigma_{\tau}} J^{N}_{a}[\Psi]n^{a}_{\Sigma_\tau}dVol_{\Sigma_\tau},\\
\check{E}^{N}_{\Psi}(\Sigma_{\tau}) &:= \int_{\Sigma_{\tau}} \check{T}_{ab}[\Psi]n^{a}_{\Sigma_\tau}N^{b}dVol_{\Sigma_\tau}.
\end{split}
\end{align}

\begin{lemma}\label{Nestimates}
Assume the $T$-energy comparison \eqref{Tcomparison} holds.  Defining the red-shift multiplier $N$ as above, we have the pointwise density estimate
\begin{align*}
\Div J^{N}[\Psi]=\nabla^a\left(\check{T}_{ab}[\Psi]Y^b\right) &\gtrsim \check{T}_{ab}[\Psi] n_{\Sigma_\tau}^aN^b\ \textup{in}\ r\leq r_0\\
                                            &\gtrsim -\check{T}_{ab}[\Psi] n_{\Sigma_\tau}^aN^b\ \textup{in}\ r\in [r_0,r_1],
\end{align*}
in addition to the energy comparison
\[E^{N}_{\Psi}(\Sigma_\tau) \approx \check{E}^{N}_{\Psi}(\Sigma_\tau)\]
and the boundary estimates
\begin{align*}
&\int_{\mathcal{H}^+(0,\tau)} J^{N}_{a}[\Psi] n_{\mathcal{H}^+} dVol_{\mathcal{H}^+} \geq 0,\\
&\int_{\mathcal{I}^+(0,\tau)} J^{N}_{a}[\Psi] \underline{L}^{a} (r^ndudVol_{S^n})\geq 0.
\end{align*}
\end{lemma}

%%%%%%%%%%%%%%%%%%%%%%%%%%%%%%%%%%%%%%
\subsection{Uniform Boundedness Estimate}

Applying the divergence theorem together with the $T$-energy comparison \eqref{Tcomparison} and Lemma \ref{Nestimates}, we have uniform boundedness of the non-degenerate $N$-energy.
\begin{theorem}\label{uniformBoundedness}
Suppose $\Psi$ is a solution of the Regge-Wheeler type equation \eqref{RWequation} as in Definition \ref{RWdefinition}.  Further, assume that $\Psi$ satisfies the $T$-energy comparison \eqref{Tcomparison}.  Then $\Psi$ satisfies the uniform boundedness estimate
\begin{equation}
\check{E}^{N}_{\Psi}(\Sigma_{\tau}) \lesssim \check{E}^{N}_{\Psi}(\Sigma_{0}),
\end{equation}
with $0 \leq \tau$.
In addition, we have the boundary estimates
\begin{align}
\begin{split}
&\int_{\mathcal{H}^+(0,\tau)} \check{T}_{ab}[\Psi]n_{\mathcal{H}^+}^aT^b\ dVol_{\mathcal{H}^+}\lesssim  \check{E}^{N}_{\Psi}(\Sigma_0),\\
&\int_{\mathcal{I}^+(0,\tau)} T_{ab}[\Psi]\underline{L}^aT^b\ \left(r^n dudVol_{S^n} \right)\lesssim   \check{E}^{N}_{\Psi}(\Sigma_0).
\end{split}
\end{align}
\end{theorem}

\begin{proof}
Application of the divergence theorem to the $N$-current \eqref{Ncurrent}, along with the $T$-energy comparison and those of Lemma \ref{Nestimates}, leads to
\begin{align*}
&\check{E}^{N}_{\Psi}(\Sigma_\tau) +\int_0^{\tau} \int_{\Sigma_{\tau'}\cap\{r\leq r_0\}} \check{T}_{ab}[\Psi] n_{\Sigma_{\tau'}}^aN^b  dVol_{\Sigma_{\tau'}} d\tau'\\
&+\int_{\mathcal{H}^+(0,\tau)} \check{T}_{ab}[\Psi] n^a_{\mathcal{H}^+}N^bdVol_{\mathcal{H}^+}+\int_{\mathcal{I}^+(0,\tau)} T_{ab}[\Psi]\underline{L}^aT^b (r^ndudVol_{S^n})\\
&\lesssim \check{E}^{N}_{\Psi}(\Sigma_0) +\int_0^{\tau} \int_{\Sigma_{\tau'}\cap\{r_0\leq r\leq r_1\}} \check{T}_{ab}[\Psi] n_{\Sigma_{\tau'}}^aT^b    dVol_{\Sigma_{\tau'}} d\tau',
\end{align*}
where we note the comparison
\[ \check{T}_{ab}[\Psi] n_{\Sigma_{\tau'}}^aT^b \approx \check{T}_{ab}[\Psi] n_{\Sigma_{\tau'}}^a N^b\]
in the region $r \geq r_0$, away from the event horizon.  We remark that the volume comparison $dVol \approx dVol_{\Sigma_{\tau'}} d\tau'$ is used in the spacetime integrals above.

Introducing the shorthand
\begin{align*}
f(\tau):=\check{E}^{T}_{\Psi}(\Sigma_\tau),\\
g(\tau):=\check{E}^{N}_{\Psi}(\Sigma_{\tau}),
\end{align*}
monotonicity of the standard $T$-energy \eqref{Tmonotonicity} and the $T$-energy comparison \eqref{Tcomparison} imply
\[f(\tau_2)\lesssim f(\tau_1)\]
for any $\tau_2 > \tau_1 \geq 0$. 

Owing to positivity of the boundary terms along the event horizon and null infinity, we find
\begin{align*}
g(\tau_2)+\int_{\tau_1}^{\tau_2} g(\tau) d\tau&\lesssim g(\tau_1)+\int_{\tau_1}^{\tau_2} f(\tau)d\tau\\
                                        &\lesssim g(\tau_1)+(\tau_2-\tau_1) f(0).
\end{align*}
To be more precise, there exists a constant $C_0\geq 1$ such that
\begin{align*}
g(\tau_2)+\int_{\tau_1}^{\tau_2} g(\tau) d\tau\leq C_0 g(\tau_1)+C_0f(0)(\tau_2-\tau_1).
\end{align*}
Now define
\begin{align*}
I:=\{\tau>0 \Big| g(\tau)>g(0)+C_0f(0) \}.
\end{align*}
By definition, $g(\tau)\leq  g(0)+C_0f(0)$ for $\tau\notin I$. For $\tau_0\in I$, $\tau_0\in (a,b)\subset I$ for some $a>0$, $a\notin I$ and $b\in\mathbb{R}\cup\{\infty\}$. In particular, $g(a)=g(0)+C_0f(0)$. Then
\begin{align*}
g(\tau_0)+C_0f(0)(\tau_0-a)\leq &g(\tau_0)+\int_{a}^{\tau_0} g(\tau) d\tau\\
                           \leq &C_0g(a)+C_0f(0)(\tau_0-a),
\end{align*}
such that
\begin{align*}
g(\tau_0)&\leq C_0(g(0)+C_0f(0))\\
&\leq C_0(1+C_0) g(0),
\end{align*}
with the last inequality following from 
\[f(\tau)\leq g(\tau).\]

The estimates for the remaining boundary terms follow from \eqref{boundaryTEstimates} and the aforementioned inequality.

\end{proof}

%%%%%%%%%%%%%%%%%%%%%%%%%%%%%%%%%%%%%%
\subsection{The Morawetz Estimate}

We assume the existence of a Morawetz estimate, first adapted to curved four-dimensional spacetime backgrounds by Blue-Soffer \cite{BlueSoffer2} and Dafermos-Rodnianski \cite{DR} and later to the higher-dimensional setting by Schlue \cite{Schlue}.  The requirements of such a multiplier are described in the following lemma:
\begin{lemma}\label{Morawetz}
There exists a current $J_{a}^{\mathfrak{M}}[\Psi]$, with density $K^{\mathfrak{M}}[\Psi]= \Div J^{\mathfrak{M}}[\Psi]$, satisfying the bulk estimates
\begin{align*}
&\textup{in}\ r\geq r_0,\\
&\int_{S^n} K^{\mathfrak{M}}[\Psi] dVol_{S^n} \gtrsim \int_{S^n}\frac{1}{r^n}\left( \left(1-\frac{(n+1)M}{r^{n-1}}\right)^2\left(|\slashed{\nabla}_{t}\Psi|^2+|\mathring{\slashed{\nabla}}\Psi|^2 \right)+|\slashed{\nabla}_{r}\Psi|^2 \right)dVol_{S^n},\\
&\textup{in}\ r\leq r_0,\\
&\int_{S^n}K^{\mathfrak{M}}[\Psi] dVol_{S^n}\gtrsim \int_{S^n}\check{T}_{ab}[\Psi]N^aN^b dVol_{S^n},
\end{align*}
and the boundary estimates
\begin{align*}
\int_{S^n}|J^{\mathfrak{M}}_a[\Psi]n^a_{\Sigma_\tau}|dVol_{S^n}&\lesssim \int_{S^n}\check{T}_{ab}[\Psi]n^a_{\Sigma_\tau}N^b dVol_{S^n},\\
\int_{S^n}\left( J^{\mathfrak{M}}_a[\Psi]n^a_{\mathcal{H}^+}\right)_{-} dVol_{S^n}&\lesssim \int_{S^n}\check{T}_{ab}[\Psi]n^a_{\mathcal{H}^+}T^b dVol_{S^n},\\
\int_{S^n}|J^{\mathfrak{M}}_a[\Psi]\underline{L}^{a}| dVol_{S^n}&\lesssim \int_{S^n}\check{T}_{ab}[\Psi]\underline{L}^{a}T^b dVol_{S^n},
\end{align*}
\end{lemma}
where  $\left( J^{\mathfrak{M}}_a[\Psi]n^a_{\mathcal{H}^+}\right)_{-}$ is the negative part of  $\left( J^{\mathfrak{M}}_a[\Psi]n^a_{\mathcal{H}^+}\right)$.

Applying the divergence theorem and Theorem \ref{uniformBoundedness}, we have the following result, essential in proving uniform decay:
\begin{theorem}\label{MorawetzPlus}
Assuming the estimates of Lemma \ref{Morawetz}, along with the conclusions of Theorem \ref{uniformBoundedness} following from the $T$-energy comparison \eqref{Tcomparison}, we have the further density estimate
\begin{equation}\label{Morawetz_int}
\int_0^\infty \int_{\Sigma_\tau\cap\{r\leq R_1\}} K^{\mathfrak{M}}[\Psi]  dVol_{\Sigma_\tau} d\tau \lesssim \check{E}^{N}_{\Psi}(\Sigma_0).
\end{equation}
\end{theorem}
\begin{proof}
By the divergence theorem
\begin{align*}
&\int_{\Sigma_\tau}J^{\mathfrak{M}}_a[\Psi]n^a_{\Sigma_\tau} dVol_{\Sigma_\tau}+\int_{\mathcal{H}^+(0,\tau)}J^{\mathfrak{M}}_a[\Psi]n^a_{\mathcal{H}^+} dVol_{\mathcal{H}^+}\\
+&\int_{\mathcal{I}^
+(0,\tau)}J^{\mathfrak{M}}_a[\Psi]\underline{L}^a \left(r^n dudVol_{S^n} \right)+\int_{D(0,\tau)}K^{\mathfrak{M}}[\Psi]dVol\\
=&\int_{\Sigma_\tau}J^{\mathfrak{M}}_a[\Psi]n^a_{\Sigma_\tau} dVol_{\Sigma_\tau}
\end{align*}
The boundary terms are controlled by the right hand side of \eqref{Morawetz_int} by Lemma \ref{Morawetz} and Theorem \ref{uniformBoundedness}. The result follows from the volume comparison $dVol\approx dVol_{\Sigma_\tau} d\tau$.
\end{proof}

Briefly we discuss the construction of candidate Morawetz currents and the calculation of their densities, deferring details on the verification of the estimates of Lemma \ref{Morawetz} to the cases at hand.  We denote $\frac{\partial f}{\partial r_*}$ by $f'$, and consider the vector field $X=f(r)\partial_{r_*}$.  We calculate
\begin{align*}
{T}_{ab}[\Psi]\nabla^a X^b=&\frac{f'}{1-\mu}|\slashed{\nabla}_{r_*}\Psi|^2+\frac{f}{r^3}\left( 1-\frac{(n+1)M}{r^{n-1}} \right)|\mathring{\slashed{\nabla}}\Psi|^2-\frac{1}{2}(f'+\frac{n(1-\mu)}{r}f)(\nablas \Psi)^2.
\end{align*}
Letting $\omega^X=f'+\frac{n(1-\mu)}{r}f$. With $\nabla_a X^a=f'+(-\mu_r+\frac{n(1-\mu)}{r})f=\omega^X-\mu_r f$, we introduce the currents
\begin{align}
\begin{split}
J^{X}_{a}[\Psi] &:= T_{ab}[\Psi]X^{b},\\
J^{X,\omega^{X}}_{a}[\Psi] &:= J^X_a[\Psi]+\frac{1}{4}\omega^X\nabla_a |\Psi|^2-\frac{1}{4}\nabla_a\omega |\Psi|^2,
\end{split}
\end{align}
and calculate their densities:
\begin{align}\label{Xdensity}
\begin{split}
\Div J^X[\Psi] =&\frac{f'}{1-\mu}|\slashed{\nabla}_{r_*}\Psi|^2+\frac{f}{r^3}\left( 1-\frac{(n+1)M}{r^{n-1}} \right)|\mathring{\slashed{\nabla}}\Psi|^2\\
                     -&\frac{1}{2}\omega^X(\nablas \Psi)^2-\frac{1}{2}fV'\Psi^2-\frac{1}{2}(\omega^X-\mu_r f)V|\Psi|^2,
\end{split}
\end{align}
\begin{align}\label{Xweighted}
\begin{split}
\Div J^{X,\omega^X}[\Psi] =&\frac{f'}{1-\mu}|\slashed{\nabla}_{r_*}\Psi|^2+\frac{f}{r^3}\left( 1-\frac{(n+1)M}{r^{n-1}} \right)|\mathring{\slashed{\nabla}}\Psi|^2\\
                        +&\left(-\frac{1}{4}\Box\omega^X-\frac{1}{2}V' f+\frac{1}{2}\mu_r V f\right)|\Psi|^2.
\end{split}
\end{align}
For any radial function $\beta (r)$,
\begin{align*}
\textup{div} \left(\frac{f'}{1-\mu}\beta |\Psi|^2\partial_{r_*}\right)=&\frac{f'}{1-\mu}\beta |\Psi|^2\left(-\frac{\mu'}{1-\mu}+\frac{n(1-\mu)}{r}\right) + \frac{f''}{1-\mu}\beta |\Psi|^2\\
&+\frac{f'\mu'}{(1-\mu)^2}\beta |\Psi|^2+\frac{f'}{1-\mu}\beta' |\Psi|^2+\frac{2f'}{1-\mu}\beta \Psi\cdot\slashed{\nabla}_{r_*}\Psi\\
                                                                   =&\frac{f''\beta}{1-\mu}|\Psi|^2+f'\left( \frac{\beta'}{1-\mu}+\frac{n\beta}{r} \right)|\Psi|^2+\frac{f'}{1-\mu}2\beta \Psi\cdot\slashed{\nabla}_{r_*}\Psi.
\end{align*}
Hence the current
\begin{equation}
J^{X,\omega^{X},\beta}_{a}[\Psi] := J^{X,\omega^X}_{a}[\Psi]+f'\beta |\Psi|^2dr_{*},
\end{equation}
has density
\begin{align}\label{Xweightedweighted}
\begin{split}
&\Div \left( J^{X,\omega^{X},\beta}[\Psi]\right)=\\
&\frac{f'}{1-\mu}|\slashed{\nabla}_{r_*}\Psi+\beta\Psi|^2 +\frac{f}{r^3}\left( 1-\frac{(n+1)M}{r^{n-1}} \right)|\mathring{\slashed{\nabla}}\Psi|^2\\
& +\left(-\frac{1}{4}\Box\omega^X-\frac{1}{2}V' f+\frac{1}{2}\mu_r V f  \right)|\Psi|^2\\
+&\left( \frac{f''\beta}{1-\mu}+f'\left( \frac{\beta'}{1-\mu}+\frac{n\beta}{r}-\frac{\beta^2}{1-\mu} \right) \right)|\Psi|^2.\\
                       =&\frac{f'}{1-\mu}|\slashed{\nabla}_{r_*}\Psi+\beta\Psi|^2+\frac{f}{r^3}\left( 1-\frac{(n+1)M}{r^{n-1}} \right)|\mathring{\slashed{\nabla}}\Psi|^2\\
                       +&\mathring{W}(f,\beta)|\Psi|^2,
\end{split}
\end{align}
where
\[\mathring{W}(f,\beta)=                     \left(-\frac{1}{4}\Box\omega^X-\frac{1}{2}V' f+\frac{1}{2}\mu_r V f  \right)+\left( \frac{f''\beta}{1-\mu}+f'\left( \frac{\beta'}{1-\mu}+\frac{n\beta}{r}-\frac{\beta^2}{1-\mu} \right) \right).\]

%%%%%%%%%%%%%%%%%%%%%%%%%%%%%%%%%%%%%%
\subsection{The $r^p$-Hierarchy}

We adapt the $r^p$-hierarchy appearing in the work of Schlue \cite{Schlue}, generalizing the earlier ideas of Dafermos-Rodnianski \cite{DR2}.

We define 
\begin{equation}\label{PhiDef}
\Phi:=r^{n/2}\Psi
\end{equation}
and the current
\begin{equation}
J_b[\Phi]:= \frac{1}{r^{n}(1-\mu)}T_{ab}[\Phi]L^a-\frac{1}{8r^{n+2}(1-\mu)}\left( (n^2-2n)+\frac{2n^2M}{r^{n-1}} \right)|\Phi|^2L_b,
\end{equation}
from which we obtain the following $r^p$ density estimate:
\begin{proposition}\label{rpDensity}
Suppose the potential $V$ has the form
\begin{align*}
V=\frac{a_0}{r^2}+\frac{b_0M}{r^{n+1}}+O\left(\frac{M^2}{r^{2n}}\right),\ a_0>0.
\end{align*}
Then for $1\leq p\leq 2$, $R_1$ large, and an appropriate choice of $k$,
\begin{align*}
\nabla^b\left(\frac{r^p J_b[\Phi]}{(1-\mu)^k}\right)\gtrsim \left( r^{p-n-1}|\slashed{\nabla}_L\Phi|^2+(2-p)r^{p-n-3}|\mathring{\slashed{\nabla}}\Phi|^2 \right)\ \ \textup{as}\ \ r\geq R_1. 
\end{align*}
\end{proposition}
\begin{proof}
From $\Box r=\frac{nr^{n-1}-2M}{r^n}$, we have 
\begin{align*}
\Box r^{n/2}=\frac{1}{4}r^{n/2-2}\left( (3n^2-2n)-\frac{2n^2M}{r^{n-1}} \right),
\end{align*}
and
\begin{align*}
\slashed{\Box}_{\mathcal{L}(-2)} \Phi=\frac{1}{4}r^{-2}\left( (n^2-2n)+\frac{2n^2M}{r^{n-1}} \right)\Phi+\frac{n}{r}(1-\mu)\slashed{\nabla}_{r}\Phi+V\Phi.
\end{align*}
By direct calculation
\begin{align*}
\nabla^aT_{ab}[\Phi] =\frac{n}{r}(1-\mu)\slashed{\nabla}_{r}\Phi\cdot\slashed{\nabla}_b \Phi+\frac{1}{4}r^{-2}\left( (n^2-2n)+\frac{2n^2M}{r^{n-1}} \right)\Phi\cdot\slashed{\nabla}_b \Phi-\frac{1}{2}|\Phi|^2\nabla_b V.
\end{align*}
Denoting 
\begin{equation}
\mathring{J}_a[\Phi] := \frac{1}{r^n(1-\mu)}T_{ab}[\Phi]L^b,
\end{equation}
we compute its divergence:
\begin{align*}
&\nabla^a\left(\frac{1}{r^n(1-\mu)}T_{ab}[\Phi]L^b\right)\\
&= \nabla^a\left( \frac{1}{r^n(1-\mu)} \right)T_{ab}[\Phi]L^b +\frac{1}{r^n(1-\mu)}\left( \nabla^a T_{ab}[\Phi]\right)L^b + \frac{1}{r^n(1-\mu)}T_{ab}[\Phi]\nabla^a L^b.
\end{align*}

The first term is
\begin{align*}
I_1=&-\frac{(nr^{n-1}-2M)}{r^{2n}(1-\mu)^2}\left( (1-\mu)\slashed{\nabla}_{r}\Phi\cdot\slashed{\nabla}_L \Phi -\frac{1}{2}(1-\mu)(\nablas\Phi)^2-\frac{1}{2}(1-\mu)V|\Phi|^2 \right)\\
   =&-\frac{(nr^{n-1}-2M)}{r^{2n}(1-\mu)}\slashed{\nabla}_{r}\Phi\cdot\slashed{\nabla}_L \Phi+\frac{(nr^{n-1}-2M)}{2r^{2n}(1-\mu)}(\nablas\Phi)^2+\frac{(nr^{n-1}-2M)}{2r^{2n}(1-\mu)}V|\Phi|^2.
\end{align*}
The second term is
\begin{align*}
I_2=&\frac{1}{r^n(1-\mu)}\left[ \frac{n}{r}(1-\mu)\slashed{\nabla}_{r}\Phi\cdot\slashed{\nabla}_b \Phi+\frac{1}{4}r^{-2}\left( (n^2-2n)+\frac{2n^2M}{r^{n-1}} \right)\Phi\cdot\slashed{\nabla}_b \Phi-\frac{1}{2}|\Phi|^2\nabla_b V  \right]L^{b}\\
   =&\frac{n}{r^{n+1}}\slashed{\nabla}_{r}\Phi\cdot\slashed{\nabla}_L \Phi+\frac{1}{4r^{n+2}(1-\mu)}\left((n^2-2n)+\frac{2n^2M}{r^{n-1}}\right)\Phi\cdot\slashed{\nabla}_L \Phi-\frac{1}{2r^n(1-\mu)}\nabla_L V|\Phi|^2.
\end{align*}
From $\nabla^{(a}L^{b)}=\frac{(n-1)M}{r^n}\tilde{g}^{AB}+\frac{1-\mu}{r^3}\mathring{\sigma}^{\alpha\beta}$, the last term is
\begin{align*}
I_3=&\frac{1}{r^n(1-\mu)}\left[ \frac{(n-1)M}{r^{n}}(\tilde{\slashed{\nabla}}\Phi)^2+\frac{1-\mu}{r^3}|\mathring{\slashed{\nabla}}\Phi|^2-\frac{1}{2}\left( \frac{nr^{n-1}-2M}{r^n} \right)((\slashed{\nabla}\Phi)^2+V|\Phi|^2) \right]\\
=&\frac{(n-1)M}{r^{2n}(1-\mu)}(\tilde{\slashed{\nabla}}\Phi)^2+\frac{1}{r^{n+3}}|\mathring{\slashed{\nabla}}\Phi|^2-\frac{nr^{n-1}-2M}{2r^{2n}(1-\mu)}((\slashed{\nabla}\Phi)^2+V|\Phi|^2).
\end{align*}
Together with $\slashed{\nabla}_{r}\Phi=\frac{1}{2(1-\mu)}(\slashed{\nabla}_L \Phi-\slashed{\nabla}_{\underline{L}}\Phi)$ and $(\tilde{\slashed{\nabla}}\Phi)^2=-\frac{1}{(1-\mu)}\slashed{\nabla}_L \Phi\slashed{\nabla}_{\underline{L}}\Phi$, we have
\begin{align*}
\nabla^a \mathring{J}_a[\Phi]=&-\frac{(n-1)M}{r^{2n}(1-\mu)^2}|\slashed{\nabla}_L \Phi|^2+\frac{1}{r^{n+3}}|\mathring{\slashed{\nabla}}\Phi|^2+ \frac{1}{r^n(1-\mu)}W \Phi\cdot\slashed{\nabla}_L \Phi-\frac{1}{2r^n(1-\mu)}\nabla_L V|\Phi|^2,
\end{align*}
where 
\begin{align*}
W=\frac{1}{r^{2}}\left[ \frac{n^2-2n}{4} +\frac{n^2 M}{2r^{n-1}} \right].
\end{align*}
Noting that
\begin{align*}
\nabla^a\left( -\frac{1}{2r^n (1-\mu)}W|\Phi|^2L_a \right)=&-\frac{1}{2r^n (1\mu)}W\Phi\cdot\slashed{\nabla}_L \Phi-\frac{1}{2r^n (1-\mu)} \nabla_L W|\Phi|^2,
\end{align*}
we deduce that $J_a[\Phi]=\mathring{J}_a[\Phi]-\frac{1}{2}W|\Phi|^2L_a$ satisfies
\begin{align*}
\nabla^b J_b[\Phi]=&-\frac{(n-1)M}{r^{2n}(1-\mu)^2}|\slashed{\nabla}_L \Phi|^2+\frac{1}{r^{n+3}}|\mathring{\slashed{\nabla}}\Phi|^2-\frac{1}{2r^n(1-\mu)}\nabla_L (V+W) |\Phi|^2.
\end{align*}
From 
\begin{align*}
T_{ab}[\Phi]\nabla^a r L^b=&\frac{1}{2}T_{ab}[\Phi]L^b(L-\underline{L})^a\\
                               =&\frac{1}{2}|\slashed{\nabla}_L \Phi|^2-\frac{1-\mu}{2r^2}|\mathring{\slashed{\nabla}}\Phi|^2-\frac{1-\mu}{2}V|\Phi|^2,
\end{align*}
we have
\begin{align*}
\nabla^a r J_a[\Phi] =\frac{1}{2r^{n}(1-\mu)}|\slashed{\nabla}_L \Phi|^2-\frac{1}{2r^{n+2}}|\mathring{\slashed{\nabla}}\Phi|^2-\frac{1}{2r^n}(V+W)|\Phi|^2.
\end{align*}
Thus
\begin{align*}
\nabla^a (r^p J_a[\Phi])=&r^p\nabla^a J_a[\Phi]+pr^{p-1}\nabla^a r J_a[\Phi]\\
                      =&\frac{r^{p-n-1}}{(1-\mu)^2}\left( \frac{p}{2}(1-\mu)-\frac{(n-1)M}{r^{n-1}} \right)|\slashed{\nabla}_L \Phi|^2+\left(1-\frac{p}{2}\right)r^{p-n-3}|\mathring{\slashed{\nabla}}\Phi|^2\\
                      -&\frac{1}{2}r^{p-n}\left(\frac{1}{1-\mu}\nabla_L (V+W)-\frac{p}{r}(V+W)\right)|\Phi|^2.
\end{align*}
From the assumption on $V$,\\
\begin{align*}
V+W=\frac{\bar{a}_0}{r^2}+\frac{ \bar{b}_0M}{r^{n+1}}+O\left(\frac{M^2}{r^{2n}}\right),
\end{align*}
where $\bar{a}_0=a_0+\frac{n^2-2n}{4}$ and $\bar{b}_0=b_0+\frac{n^2}{2}$, the coefficient of $|\Phi|^2$ has the form
\begin{align*}
r^{p-n-3} \left( \bar{a}_0\left(1-\frac{p}{2}\right)+ \frac{\bar{b}_0}{2}(n+1-p)\frac{M}{r^{n-1}}+O\left(\frac{M^2}{r^{2n-2}}\right) \right).
\end{align*}
The leading positive term $\bar{a}_0\left(1-\frac{p}{2}\right)$ vanishes as $p=2$. In order to make the coefficient positive, we consider
\begin{align*}
&\nabla^a\left( \frac{r^p J_a[\Phi]}{(1-\mu)^{k-1}} \right)\\
&=\frac{r^{p-n-1}}{(1-\mu)^{k+1}}\left[ \frac{p}{2}(1-\mu)-\frac{k(n-1)M}{r^{n-1}} \right]|\slashed{\nabla}_L \Phi|^2\\
&+\frac{r^{p-3-n}}{(1-\mu)^{k-1}}\left[ \left(1-\frac{p}{2}\right)+(k-1)(n-1)(1-\mu)^{-1}\frac{M}{r^{n-1}} \right]|\mathring{\slashed{\nabla}}\Phi|^2\\
&+\frac{r^{p-n-3}}{(1-\mu)^{k-1}}\left[ \bar{a}_0\left(1-\frac{p}{2} \right) + \Big(\frac{\bar{b}_0}{2}(n+1-p)+\bar{a}_0(k-1)(n-1)\Big)\frac{M}{r^{n-1}}  +O\left(\frac{M^2}{r^{2n-2}}\right) \right]|\Phi|^2.
\end{align*}
By taking $k$ large such that $\frac{\bar{b}_0}{2}(n+1-p)+\bar{a}_0(k-1)(n-1)>0$, the coefficient of $|\Phi|^2$ above is non-negative for $r \geq R_1$ sufficiently large.  Rechoosing $R_1$ as necessary, comparison of the coefficient of $|\slashed{\nabla}_{L}\Phi|^2$ above with that of the right-hand side in the proposition holds, and the proposition follows.  
\end{proof}
We remark that this lemma will be applied to the Regge-Wheeler type equations \eqref{RW1}, \eqref{RW3}, and \eqref{RW4}.  Since the potentials $V^{(+)}_\ell$ in \eqref{RW4} satisfy $V^{(+)}_\ell\to V^{(+)}$ as $\ell\to\infty$, we can choose $R_1$ and $k$ large which hold for all the equations  \eqref{RW1}, \eqref{RW3}, and \eqref{RW4}.

Defining
\begin{align}
\begin{split}
E^p_{\Psi}(\Sigma_\tau)&:=\int_{\Sigma_\tau\cap\{r\geq R_1\}} \frac{r^{p-n}}{(1-\mu)^k} \frac{|\slashed{\nabla}_L\Phi|^2}{(1-\mu)} r^n dv dVol_{S^n},\\
A^p_{\Psi}(\Sigma_{\tau})&:=\int_{\Sigma_\tau\cap\{r\geq R_1\}} r^{p-n}|r^{-2}\mathring{\slashed{\nabla}}\Phi|^2 r^n dv dVol_{S^n},
\end{split}
\end{align}
with $\Phi$ specified in terms of $\Psi$ by \eqref{PhiDef}, the Morawetz estimate of Lemma \ref{Morawetz} and the $r^p$ density estimate of Proposition \ref{rpDensity} give the $r^p$-estimate:
\begin{theorem}\label{rpEstimate}
Suppose $\Psi$ is a solution of the Regge-Wheeler type equation \eqref{RWequation} as in Definition \ref{RWdefinition}.  Further, assume that $\Psi$ satisfies the $T$-energy comparison \eqref{Tcomparison} and the Morawetz estimate of Lemma \ref{Morawetz}, with choices of $k$ and $R_1$ made such that Proposition \ref{rpDensity} holds.  For $1\leq p\leq 2$ and $\tau_1<\tau_2$, we have the $r^p$-estimate:
\begin{align}
\begin{split}
&E^p_{\Psi}(\Sigma_{\tau_2})+\frac{1}{C}\int_{\tau_1}^{\tau_2}\Big( E^{p-1}_{\Psi}(\Sigma_\tau)+(2-p)A^{p-1}_{\Psi}(\Sigma_\tau)\Big) d\tau\\
&\leq E^p_{\Psi}(\Sigma_{\tau_1})+C\check{E}^N_{\Psi}(\Sigma_{\tau_1}).
\end{split}
\end{align}
\end{theorem}
\begin{proof}
Let $J_{a}^{p}[\Phi]=\eta(r) \frac{r^p J_{a}[\Phi]}{(1-\mu)^k}$, where $\eta$ is a radial cut-off function with $\chi=1$ for $r\geq R_1$ and $\chi=0$ for $r\leq R_1-1$, with $k$ and $R_1$ chosen to ensure that Proposition \ref{rpDensity} holds.  Let $J_{a}^{\mathfrak{M}}[\Psi]$ be a current satisfying the Morawetz estimates in Lemma \ref{Morawetz}.  Using Proposition \ref{rpDensity} and the coercivity of $K^{\mathfrak{M}}[\Psi] := \Div J^{\mathfrak{M}}[\Psi]$ in $R_1-1\leq r\leq R_1$ , we have for $C>>1$
\begin{align*}
\textup{div}(J^p[\Phi]+CJ^{\mathfrak{M}}[\Psi])&\geq 0 \ as\ r\leq R_1,\\
\textup{div}(J^p[\Phi]+CJ^{\mathfrak{M}}[\Psi])&\gtrsim r^{p-n-1}|\nablas_{v}\Psi|^2+(2-p)r^{p-n-3}|\mathring{\nablas}\Psi|^2\ as\ r\geq R_1.
\end{align*}
Owing to the $T$-energy comparison \eqref{Tcomparison}, implying Theorem \ref{uniformBoundedness}, we have estimates on many of the boundary terms in terms of the initial $N$-energy.
The result then follows from the divergence theorem.
\end{proof}
In addition, we have the lemma:
\begin{lemma}\label{smallLemma}
\begin{align*}
E^0_{\Psi}(\Sigma_\tau)+A^0_{\Psi}(\Sigma_\tau)\gtrsim \int_{\Sigma_{\tau}\cap \{r\geq R_1\}} T_{ab}[\Psi]N^an^b_{\Sigma_\tau} dVol_{\Sigma_\tau}
\end{align*}
\end{lemma}
\begin{proof}
The integrand of $E_0$ is $|\slashed{\nabla}_L(r^{n/2}\Psi)|^2\approx |r^{n/2}\slashed{\nabla}_L\Psi+\frac{n}{2}r^{n/2-1}\Psi|^2$. Borrowing some of the $|\Psi|^2$ term from $A_0$ and applying a Poincar\'{e} inequality yields the estimate.
\end{proof}

%%%%%%%%%%%%%%%%%%%%%%%%%%%%%%%%%%%%%%
\subsection{Uniform Decay Estimates}

We introduce the notation
\begin{align}
\begin{split}
E^2_{\Psi, \mathcal{L}_{K}\Psi}(\Sigma_{\tau}) &:= E^2_{\Psi}(\Sigma_{\tau}) + E^2_{\mathcal{L}_{T}\Psi}(\Sigma_{\tau}) + E^2_{\mathcal{L}_{\Omega}\Psi}(\Sigma_{\tau}),\\
E^N_{\Psi, \mathcal{L}_{K}\Psi}(\Sigma_{\tau}) &:= E^N_{\Psi}(\Sigma_{\tau}) + E^N_{\mathcal{L}_{T}\Psi}(\Sigma_{\tau}) + E^N_{\mathcal{L}_{\Omega}\Psi}(\Sigma_{\tau}),
\end{split}
\end{align}
with $K := \{ T, \Omega_{i}\}$ ranging over the background Killing fields and $\Omega := \{ \Omega_{i} \}$ ranging over the angular Killing fields.

Assuming the $T$-energy comparison \eqref{Tcomparison} and the Morawetz estimate of Lemma \ref{Morawetz}, we have the following theorem on uniform decay.

\begin{theorem}
Suppose $\Psi$ is a solution of the Regge-Wheeler type equation \eqref{RWequation} as in Definition \ref{RWdefinition}.  Further, assume that $\Psi$ satisfies the $T$-energy comparison \eqref{Tcomparison} and the Morawetz estimate of Lemma \ref{Morawetz}, with choices of $k$ and $R_1$ made such that Proposition \ref{rpDensity} holds.  Then $\Psi$ satisfies the uniform decay estimate
\begin{equation}
\check{E}^{N}_{\Psi}(\Sigma_{\tau}) \lesssim \frac{I_{\Psi}(\Sigma_0)}{\tau^2},
\end{equation}
where 
\begin{equation} 
I_{\Psi}(\Sigma_0):=E^2_{\Psi, \mathcal{L}_{K}\Psi}(\Sigma_0)+E^N_{\Psi, \mathcal{L}_{K}\Psi, \mathcal{L}^2_{K}\Psi}(\Sigma_0).
\end{equation}
\end{theorem}
\begin{proof}
Note that the $T$-energy comparison \eqref{Tcomparison} implies the conclusions of Theorem \ref{uniformBoundedness}.  In addition, these conclusions and the Morawetz estimates of Lemma \ref{Morawetz} imply Theorem \ref{MorawetzPlus}.  Finally, these results, along with the density estimate of Proposition \ref{rpDensity}, imply the $r^p$-estimate of Theorem \ref{rpEstimate}.

By setting $\tau_k=2^k$ and applying the $r^p$-estimate in $(\tau_{k-1},\tau_k)$ with $p=2$, we have
\begin{align*}
E^2_{\Psi}(\Sigma_{\tau_k})+\frac{1}{C}\int_{\tau_{k-1}}^{\tau_k} E^1_{\Psi}(\Sigma_\tau) d\tau\leq E^2_{\Psi}(\Sigma_{\tau_{k-1}})+C\check{E}^N_{\Psi}(\Sigma_{\tau_{k-1}}).
\end{align*}
From the mean value theorem, there exist $\tilde{\tau_k}\in [\tau_{k-1},\tau_k]$ such that
\begin{align*}
E^1_{\Psi}(\Sigma_{\tilde{\tau}_k})&\leq \frac{C}{\tau_k} \left( E^2_{\Psi}(\Sigma_{\tau_{k-1}})+\check{E}^N_{\Psi}(\Sigma_{\tau_{k-1}}) \right)\\
                   &\leq \frac{C}{\tau_k} \left( E^2_{\Psi}(\Sigma_{\tau_{0}})+\check{E}^{N}_{\Psi}(\Sigma_{\tau_0})\right),
\end{align*}
where we used $r^p$-estimate in $(0,\tau_{k-1})$ and uniform boundedness of the $N$-energy.  Applying the $r^p$-estimate in $(\tilde{\tau}_{k-2},\tilde{\tau}_k)$ with $p=1$,
\begin{align*}
E^1_{\Psi}(\Sigma_{\tilde{\tau}_k})+\frac{1}{C}\int_{\tilde{\tau}_{k-2}}^{\tilde{\tau}_k}\Big( E^0_{\Psi}(\Sigma_\tau)+A^0_{\Psi}(\Sigma_\tau)\Big) d\tau\leq E^1_{\Psi}(\Sigma_{\tilde{\tau}_{k-2}})+CE^N_{\Psi}(\Sigma_{\tilde{\tau}_{k-2}}).
\end{align*}
Together with Lemma \ref{smallLemma}, we have 
\begin{align*}
E^1_{\Psi}(\Sigma_{\tilde{\tau}_k})+\frac{1}{C}\int_{\tilde{\tau}_{k-2}}^{\tilde{\tau}_k} \int_{\Sigma_\tau\cap\{r\geq R_1\}} T_{ab}[\Psi]N^an^b_{\Sigma_\tau}  d\tau\leq E^1_{\Psi}(\Sigma_{\tilde{\tau}_{k-2}})+CE^N_{\Psi}(\Sigma_{\tilde{\tau}_{k-2}}).
\end{align*}

Applying Theorem \ref{MorawetzPlus} to $\Psi, \mathcal{L}_T \Psi, \mathcal{L}_\Omega \Psi$,
\begin{align*}
&\int_{{\tau}}^{\infty} \int_{\Sigma_{\tau'}\cap\{r\leq R_1\}} T_{ab}[\Psi]N^an^b_{\Sigma_{\tau'}}d\tau'\\
&\lesssim\int_{D(\tau,\infty)} \left(K^{\mathfrak{M}}[\Psi] + K^{\mathfrak{M}}[\mathcal{L}_T \Psi] + K^{\mathfrak{M}}[\mathcal{L}_\Omega \Psi]\right)dVol\\
&\lesssim E^{N}_{\Psi,\mathcal{L}_T \Psi,\mathcal{L}_\Omega \Psi}(\Sigma_\tau).
\end{align*}
By adding above two inequalities, we get
\begin{align*}
E^1_{\Psi}(\Sigma_{\tilde{\tau}_k})+\frac{1}{C}\int_{\tilde{\tau}_{k-2}}^{\tilde{\tau}_k} E^N_{\Psi}(\Sigma_\tau) d\tau&\leq E^1_{\Psi}(\Sigma_{\tilde{\tau}_{k-2}})+CE^N_{\Psi,\mathcal{L}_T \Psi,\mathcal{L}_\Omega \Psi}(\Sigma_{\tilde{\tau}_{k-2}})\\
&\lesssim \frac{1}{\tau_k} (E^2_{\Psi}(\Sigma_{\tau_0})+E^N_{\Psi}(\Sigma_{\tau_0}))+E^N_{\Psi,\mathcal{L}_T \Psi,\mathcal{L}_\Omega \Psi}(\Sigma_{\tilde{\tau}_{k-2}})
\end{align*}
From the mean value theorem, there exists $\hat{\tau}_k\in (\tilde{\tau}_{k-2},\tilde{\tau}_{k})$ such that
\begin{align}\label{end}
\begin{split}
E^N_{\Psi}(\Sigma_{\hat{\tau}_k})&\lesssim \frac{1}{\tau_k^2} (E^2_{\Psi}(\Sigma_{\tau_0})+E^N_{\Psi}(\Sigma_{\tau_0}))+\frac{1}{\tau_k} E^N_{\Psi,\mathcal{L}_T \Psi,\mathcal{L}_\Omega \Psi}(\Sigma_{\tilde{\tau}_{k-2}})\\
&\lesssim \frac{1}{\tau_k^2} (E^2_{\Psi}(\Sigma_{\tau_0})+E^N_{\Psi,\mathcal{L}_K\Psi}(\Sigma_{\tau_0}))+\frac{1}{\tau_k} E^N_{\mathcal{L}_K \Psi}(\Sigma_{\tilde{\tau}_{k-2}}),
\end{split}
\end{align}
where we have used
\begin{align*}
E^{N}_{\Psi}(\Sigma_{\tilde{\tau}_{k-2}}) &\lesssim E^{N}_{\Psi}(\Sigma_{\hat{\tau}_{k-2}})\\
& \lesssim \frac{1}{\tau_k^2} (E^2_{\Psi}(\Sigma_{\tau_0})+E^N_{\Psi}(\Sigma_{\tau_0}))+\frac{1}{\tau_k} E^N_{\Psi,\mathcal{L}_T \Psi,\mathcal{L}_\Omega \Psi}(\Sigma_{\tilde{\tau}_{k-2}}).
 \end{align*}
Finally, applying \eqref{end} for $\mathcal{L}_K \Psi$, 
\begin{align*}
E^N_{\mathcal{L}_K \Psi}(\Sigma_{\hat{\tau}_k})&\lesssim \frac{1}{\tau_k^2}(E^2_{\mathcal{L}_K \Psi}(\Sigma_{\tau_0})+E^N_{\mathcal{L}_K \Psi,\mathcal{L}_K^2 \Psi}(\Sigma_{\tau_0}))+\frac{1}{\tau_k}E^N_{\mathcal{L}_K^2 \Psi}(\Sigma_{\tilde{\tau}_{k-2}})\\
&\lesssim \frac{1}{\tau_k^2}(E^2_{\mathcal{L}_K \Psi}(\Sigma_{\tau_0})+E^N_{\mathcal{L}_K \Psi,\mathcal{L}_K^2 \Psi}(\Sigma_{\tau_0}))+\frac{1}{\tau_k}E^N_{\mathcal{L}_K^2 \Psi}(\Sigma_{\tau_0}),
\end{align*}
and plugging in to \eqref{end} leads to  
\begin{align}
 E^N_{\Psi}(\Sigma_{\hat{\tau}_{k-2}})\lesssim \frac{1}{\tau_k^2}\Big( E^2_{\Psi,\mathcal{L}_K \Psi}(\Sigma_{\tau_0})+E^N_{\Psi,\mathcal{L}_K \Psi,\mathcal{L}_K^2 \Psi}(\Sigma_{\tau_0}) \Big).
\end{align}
\end{proof}

%%%%%%%%%%%%%%%%%%%%%%%%%%%%%%%%%%%%%%
\section{The Two-Tensor Portion}

%%%%%%%%%%%%%%%%%%%%%%%%%%%%%%%%%%%%%%
\subsection{The Linearized Einstein Equations}

There is just a single wave-type equation involving the two-tensor portion.  Namely, from \eqref{tracelessTwo} we have
\begin{align*}
&-r^{2}\tilde\Box\left(r^{-2}\hat{h}_{\alpha\beta}\right) - nr^{-1}r^{A}\tilde\nabla_{A}\hat{h}_{\alpha\beta} -r^{-2}\mathring\Delta \hat{h}_{\alpha\beta}\\
&+2r^{-2}\left(n+r^{A}r_{A}-r\tilde\Box r\right)\hat{h}_{\alpha\beta} = 0.
\end{align*}
%%%%%%%%%%%%%%%%%%%%%%%%%%%%%%%%%%%%%%%
\subsection{Master Equation for the Two-Tensor Portion}
Expanding the first term in the equation above and noting
\[ -r^2\tilde\Box\left(r^{-2}\right) = 2r^{-1}\tilde\Box r -6r^{-2}r^{A}r_{A},\]
we deduce 
\[ \slashed{\Box}_{\mathcal{L}(-2)} \hat{h}_{\alpha\beta} = 2r^{-2}\left(n +(1-n)r^{A}r_{A}-r\tilde\Box r\right)\hat{h}_{\alpha\beta},\]
or
\begin{equation}\label{RW1}
\slashed{\Box}_{\mathcal{L}(-2)} \hat{h}_{\alpha\beta} = U\hat{h}_{\alpha\beta},
\end{equation}
with
\begin{equation}
U := 2r^{-2}.
\end{equation}
In deriving \eqref{RW1}, we have made use of \eqref{SchwarzFormulae} and \eqref{TwoTensorWave}.  The equation is analogous to that first discovered in Kodama-Ishibashi-Seto \cite{KI1}.

%%%%%%%%%%%%%%%%%%%%%%%%%%%%%%%%%%%%%%%%%
\subsection{Analysis of the Master Equation}

In contrast with the Regge-Wheeler type equations to be considered later in this work, the potential in \eqref{RW1} is non-negative throughout the exterior region.  As a consequence, the natural stress-energy tensor \eqref{stressTensor} satisfies a positive energy condition, dramatically simplifying the analysis.

\subsubsection{The $T$-Energy Comparison}
Owing to the positive energy condition, the $T$-energy comparison \eqref{Tcomparison} is easily satisfied.  This, in turn, implies the uniform boundedness result of Theorem \ref{uniformBoundedness}.

\subsubsection{The Morawetz Estimate}

The calculation of the weighted density of the current $J_{a}^{X,\omega^{X}}[\hat{h}]$ \eqref{Xweighted} agrees with that of the two-tensor wave equation, up to the presence of the potential terms
\begin{align*}
&\left[\frac{1}{2}\partial_{r}\mu f U - \frac{1}{2}f(1-\mu)\partial_{r}U\right]|\hat{h}|^2\\
&=\frac{2f}{r^3}\left(1-\frac{n+1}{2}\mu\right)|\hat{h}|^2,
\end{align*}
with coefficient proportional to that of the angular gradient.  Choosing a function $f$ increasing on the exterior and vanishing at the photon sphere, the extra terms above are manifestly non-negative.  Such a choice of $f$ appears in the work of Schlue \cite{Schlue}. Then we can define $J_{a}^{\mathfrak{M}}[\hat{h}]$ as
\begin{equation}
J_{a}^{\mathfrak{M}}[\hat{h}]:=J_{a}^{X,\omega^{X}}[\hat{h}]+\frac{1}{4}(\nabla_a g)|\hat{h}|^2-\frac{1}{4}h\nabla_a g |\hat{h}|^2+\epsilon_3\check{T}_{ab}[\hat{h}]Y^b,
\end{equation}
with suitable choices of $g(r)$ and $\epsilon_3>0$ to obtain the density estimate in Lemma \ref{Morawetz}. See subsection \ref{covectorMorawetz} for details. In addition, the function $f$ constructed in \cite{Schlue} is uniformly bounded on the exterior region, ensuring that the boundary terms associated with the $\mathfrak{M}$-current are dominated by those of the $N$-current as required in Lemma \ref{Morawetz}.  

\subsection{Uniform Boundedness and Decay of the Master Quantity}

With the $T$-comparison and the Morawetz estimates in place, we have uniform boundedness and decay of solutions to the Regge-Wheeler equation \eqref{RW1} in all spacetime dimensions:

\begin{theorem}
Let $\delta g$ be a smooth, symmetric two-tensor on a Schwarzschild-Tangherlini spacetime, satisfying the linearized Einstein equation \eqref{linEinstein}.  There exists a gauge-invariant master quantity $\hat{h}_{\alpha\beta}$ in the two-tensor portion $h_3$ of $\delta g$ satisfying the Regge-Wheeler type equation \eqref{RW1}.  As a solution of \eqref{RW1}, $\hat{h}_{\alpha\beta}$ satisfies the uniform boundedness estimate
\begin{equation}
\check{E}^{N}_{\hat{h}}(\Sigma_\tau) \lesssim \check{E}^{N}_{\hat{h}}(\Sigma_0),
\end{equation}
and the uniform decay estimate
\begin{equation}
\check{E}^{N}_{\hat{h}}(\Sigma_{\tau}) \lesssim \frac{I_{\hat{h}}(\Sigma_0)}{\tau^2},
\end{equation}
where 
\begin{equation} 
I_{\hat{h}}(\Sigma_0):=E^2_{\hat{h}, \mathcal{L}_{K}\hat{h}}(\Sigma_0)+E^N_{\hat{h}, \mathcal{L}_{K}\hat{h}, \mathcal{L}^2_{K}\hat{h}}(\Sigma_0)
\end{equation}
and $\tau \geq 0.$
\end{theorem}

We emphasize that the relevant constants in the comparisons depend only upon the orbit sphere dimension $n$ and the mass $M > 0$.

%%%%%%%%%%%%%%%%%%%%%%%%%%%%%%%%%%%%%%%%%%
\section{The Co-Vector Portion}

%%%%%%%%%%%%%%%%%%%%%%%%%%%%%%%%%%%%%%%%%%
\subsection{The Linearized Einstein Equations}
The cross-term and the traceless portion of the pure angular term of the linearized Ricci tensor above yield the co-vector equations

\begin{align}
\begin{split}
&-\tilde\Box \hat{h}_{A\alpha} - r^{-2}\mathring\Delta\hat{h}_{A\alpha} + (2-n)r^{-1}r^{B}\tilde\nabla_{B}\hat{h}_{A\alpha}\\
&+(n-1)r^{-2}\hat{h}_{A\alpha} +\tilde\nabla^{B}\tilde\nabla_{A}\hat{h}_{B\alpha} -2r^{-1}\left(\tilde\nabla_{A}\tilde\nabla^{B}r\right)\hat{h}_{B\alpha}\\
&+2(1-n)r^{-2}r_{A}r^{B}\hat{h}_{B\alpha}-2r^{-1}r_{A}\tilde\nabla^{B}\hat{h}_{B\alpha}+nr^{-1}r^{B}\tilde\nabla_{A}\hat{h}_{B\alpha}\\
&+(n-1)\tilde\nabla_{A}\left(r^{-2}\hat{h}_{\alpha}\right) + \tilde\nabla_{A}\left(r^{-2}\mathring\Delta\hat{h}_{\alpha}\right) = 0,
\end{split}
\end{align}

\begin{align}
\begin{split}
&\tilde\nabla^{A}\hat{h}_{A\beta} +(n-2)r^{-1}r^{A}\hat{h}_{A\beta} -r^{2}\tilde\Box\left(r^{-2}\hat{h}_{\beta}\right)\\
&-nr^{-1}r^{A}\tilde\nabla_{A}\hat{h}_{\beta} +2(n-1)r^{-2}\hat{h}_{\beta}+2r^{-2}r^{A}r_{A}\hat{h}_{\beta}\\
&-2r^{-1}\left(\tilde\Box r\right)\hat{h}_{\beta} = 0.
\end{split}
\end{align}

We introduce the connection-level co-vector quantities
\begin{align}
P_{\alpha} &:= r^{3}\epsilon^{AB}\tilde\nabla_{B}\left(r^{-2}\hat{h}_{A\alpha}\right),\\
Q_{\alpha\beta A} &:= \mathring\nabla_{\beta}\hat{h}_{A\alpha} +\mathring\nabla_{\alpha}\hat{h}_{A\beta}- r^{2}\tilde\nabla_{A}\left(r^{-2}\left(\mathring\nabla_{\beta}\hat{h}_{\alpha}+\mathring\nabla_{\alpha}\hat{h}_{\beta}\right)\right),
\end{align}
which are moreover gauge-invariant.  The co-vector equations above can be rewritten in terms of these gauge-invariant quantities as

\begin{align}
-r^{-n}\epsilon_{AB}\tilde\nabla^{B}\left(r^{n-1}P_{\alpha}\right) -r^{-2}\mathring\nabla^{\beta}Q_{\alpha\beta A} &= 0,\label{one}\\
r^{2-n}\tilde\nabla^{A}\left(r^{n-2}Q_{\alpha\beta A}\right) &= 0\label{two}.
\end{align}

In addition, the two are related by
\begin{equation}\label{three}
\epsilon^{AB}\tilde\nabla_{B}\left(r^{-2}Q_{\alpha\beta A}\right) - r^{-3}\left(\mathring\nabla_{\alpha}P_{\beta} + \mathring\nabla_{\beta} P_{\alpha}\right)= 0.
\end{equation}

%%%%%%%%%%%%%%%%%%%%%%%%%%%%%%%%%%%%%
\subsection{Master Equations for the Co-Vector Portion}

\subsubsection{Decoupling of $P_{\alpha}$}
The decoupling of $P_{\alpha}$ proceeds as follows.  Applying the operator $\epsilon^{AB}\tilde\nabla_{B}$ to \eqref{one} and rewriting the result with \eqref{three}, we find
\begin{align*}
&\epsilon^{AB}\tilde\nabla_{B}\left(r^{-n}\epsilon_{AC}\tilde\nabla^{C}\left(r^{n-1}P_{\alpha}\right)\right) +\epsilon^{AB}\tilde\nabla_{B}\left(r^{-2}\mathring\nabla^{\beta}Q_{\alpha\beta A}\right) = 0,\\
&g^{BC}\left(\tilde\nabla_{B}\left(r^{-n}\tilde\nabla_{C}\left(r^{n-1}P_{\alpha}\right)\right)\right) + r^{-3}\mathring\Delta P_{\alpha} + (n-1)r^{-3}P_{\alpha} = 0.
\end{align*}

Expanding the first term and multiplying through by $r$, we have
\begin{align*}
&\tilde\Box P_{\alpha} + (n-2)r^{-1}r^{B}\tilde\nabla_{B}P_{\alpha} + r^{-2}\mathring\Delta P_{\alpha}\\  
&+\left((n-1)r^{-2}- 2r^{-2}(n-1)r^{B}r_{B} + (n-1)r^{-1}\left(\tilde\Box r\right)\right)P_{\alpha} = 0.
\end{align*}

Noting the formula for the spin-$1$ d'Alembertian \eqref{CovectorWave}, along with the background formulae \eqref{SchwarzFormulae}, the equation reduces to
\begin{equation}\label{RW2}
\slashed{\Box}_{\mathcal{L}(-1)} P_{\alpha} = W P_{\alpha},
\end{equation}
where
\begin{equation}
W := \frac{1}{r^2}-\frac{2Mn^2}{r^{n+1}}.
\end{equation}
We remark that $P_{\alpha}$ is the higher-dimensional analog of the Cunningham-Moncrief-Price quantity \cite{CMP} in four spacetime dimensions.  

\subsubsection{Decoupling of $Q_{\alpha\beta A}r^{A}$}
Multiplying \eqref{three} by $r^{n+2}$, and applying the operator $\epsilon_{AB}\tilde\nabla^{B}$ to the result, we find
\begin{align*}
&\epsilon_{AB}\tilde\nabla^{B}\left(r^{n+2}\epsilon^{CD}\tilde\nabla_{D}\left(r^{-2}Q_{\alpha\beta C}\right)\right)\\
&-\epsilon_{AB}\tilde\nabla^{B}\left(r^{n-1}\mathring\nabla_{\alpha}P_{\beta}\right) - \epsilon_{AB}\tilde\nabla^{B}\left(r^{n-1}\mathring\nabla_{\beta}P_{\alpha}\right) = 0,
\end{align*}
or, applying \eqref{one},
\begin{align*}
&r^{-n}\epsilon_{AB}\epsilon^{CD}\tilde\nabla^{B}\left(r^{n+2}\tilde\nabla_{D}\left(r^{-2}Q_{\alpha\beta C}\right)\right)\\
&+r^{-2}\mathring\nabla_{\alpha}\mathring\nabla^{\gamma}Q_{\beta\gamma A} + r^{-2}\mathring\nabla_{\beta}\mathring\nabla^{\gamma}Q_{\alpha\gamma A} = 0.
\end{align*}
The first term above can be expanded by appealing to the relation
\[ \epsilon_{AB}\epsilon^{CD}P^{B}_{DC} = P^{B}_{BA} - P^{B}_{AB},\]
valid for tensors on the two-dimensional quotient space.  Applying this result, and contracting the equation with $r^{A}$, we find
\begin{align*}
&\left(\tilde\Box Q_{\alpha\beta A}\right)r^{A} - 2r^{-1}r^{A}r^{B}\tilde\nabla_{B}Q_{\alpha\beta A} - r^{-1}\left(\tilde\Box r\right)Q_{\alpha\beta A}r^{A}\\
&+2r^{-1}(r^{B}r_{B})\tilde\nabla^{B}Q_{\alpha\beta B} - r^{A}\left(\tilde\nabla^{B}\tilde\nabla_{A} Q_{\alpha\beta B}\right)\\
&+r^{-2}\mathring\nabla_{\alpha}\mathring\nabla^{\gamma}\left(Q_{\beta\gamma A}r^{A}\right) + r^{-2}\mathring\nabla_{\beta}\mathring\nabla^{\gamma}\left(Q_{\alpha\gamma A}r^{A}\right) = 0.
\end{align*}

Commuting the covariant derivative, and applying \eqref{two}, we rewrite the term
\begin{align*}
&r^{A}\left(\tilde\nabla^{B}\tilde\nabla_{A}Q_{\alpha\beta B}\right)\\
&= (n-2)r^{-2}r^{B}r_{B}Q_{\alpha\beta A}r^{A} - (n-2)r^{-1}r^{A}\left(\tilde\nabla_{A}\tilde\nabla^{C} r\right)Q_{\alpha\beta C}\\
&-(n-2)r^{-1}r^{A}r^{C}\tilde\nabla_{C}Q_{\alpha\beta A} + \tilde{K}Q_{\alpha\beta A}r^{A}.
\end{align*}

With this, and application of \eqref{two} to the preceding divergence term, our equation takes the form
\begin{align*}
&\left(\tilde\Box Q_{\alpha\beta A}\right)r^{A} + (n-4)r^{-1}r^{A}r^{B}\tilde\nabla_{B}Q_{\alpha\beta B} - r^{-1}\left(\tilde\Box r\right)Q_{\alpha\beta A}r^{A}\\
&-2(n-2)r^{-2}r^{B}r_{B}Q_{\alpha\beta A}r^{A} - \tilde{K}Q_{\alpha\beta A}r^{A} - (n-2)r^{-2}r^{B}r_{B}Q_{\alpha\beta A}r^{A}\\
&+(n-2)r^{-1}r^{A}\left(\tilde\nabla_{A}\tilde\nabla^{B} r\right)Q_{\alpha\beta B}\\
&+r^{-2}\left(\mathring\nabla_{\alpha}\mathring\nabla^{\gamma}Q_{\beta\gamma A}r^{A}\right) + r^{-2}\left(\mathring\nabla_{\beta}\mathring\nabla^{\gamma}Q_{\alpha\gamma A}r^{A}\right) = 0.
\end{align*}

Comparing this expression with the spin-$2$ d'Alembertian \eqref{TwoTensorWave} applied to $Q_{\alpha\beta A}r^{A}$,
\begin{align*}
&\slashed{\Box}_{\mathcal{L}(-2)}\left(Q_{\alpha\beta A}r^{A}\right) = \left(\tilde\Box Q_{\alpha\beta A}\right)r^{A} + \left(\tilde\Box r^{A}\right)Q_{\alpha\beta A}\\
&+2\left(\tilde\nabla^{A}\tilde\nabla^{A} r\right)\tilde\nabla_{B}Q_{\alpha\beta A} + (n-4)r^{-1}r^{A}\left(\tilde\nabla_{A}\tilde\nabla^{B} r\right)Q_{\alpha\beta B}\\
&+(n-4)r^{-1}r^{A}r^{B}\tilde\nabla_{B}Q_{\alpha\beta B} + r^{-2}\mathring\Delta\left(Q_{\alpha\beta A}r^{A}\right)\\
&+(6-2n)r^{-2}r^{B}r_{B}Q_{\alpha\beta A}r^{A} - 2r^{-1}\left(\tilde\Box r\right)Q_{\alpha\beta A}r^{A},
\end{align*}
a lengthy reduction, using also the background calculations \eqref{SchwarzFormulae} and the commutation relation
\[ \mathring\nabla_{\alpha}\mathring\nabla^{\gamma}Q_{\beta\gamma A} + \mathring\nabla_{\beta}\mathring\nabla^{\gamma}Q_{\alpha \gamma A} - \mathring\Delta Q_{\alpha\beta A} = -2Q_{\alpha\beta A},\]
yields the equation
\begin{equation}\label{RW3}
\slashed\Box_{\mathcal{L}(-2)}\left(Q_{\alpha\beta A}r^{A}\right) = V^{(-)}\left(Q_{\alpha\beta A}r^{A}\right),
\end{equation}
with
\begin{equation}\label{RWPotential}
V^{(-)} := \frac{n+2}{r^2} - \frac{2Mn^2}{r^{n+1}}.
\end{equation}
The quantity and associated equation are analogous to those first discovered in Kodama-Ishibashi-Seto \cite{KI1}, although their derivation is quite different in that work.

Subsequently, we employ the shorthand
\begin{equation}\label{QMinusDef}
Q^{(-)}_{\alpha\beta} := Q_{\alpha\beta A}r^{A}.
\end{equation}

\subsubsection{Spin-Raising of $P_{\alpha}$}

We denote by $\mathcal{D}$ the symmetrized gradient operation, and consider the quantity
\begin{equation}\label{SDef}
S_{\alpha\beta} := r\left(\mathcal{D}P\right)_{\alpha\beta} := r\left(\mathring\nabla_{\alpha}P_{\beta} + \mathring\nabla_{\beta}P_{\alpha}\right).
\end{equation}

Expanding with definition \eqref{TwoTensorWave}, we find
\begin{align*}
\slashed\Box_{\mathcal{L}(-2)}S_{\alpha\beta} &= \tilde\Box S_{\alpha\beta} + (n-4)r^{-1}r^{A}\tilde\nabla_{A}S_{\alpha\beta} \\
&+r^{-2}\mathring\Delta S_{\alpha\beta} + (6-2n)r^{-2}r^{A}r_{A}S_{\alpha\beta} -2r^{-1}\left(\tilde\Box r\right)S_{\alpha\beta}\\
&= r\left(\mathcal{D}\slashed\Box_{\mathcal{L}(-1)}P\right)_{\alpha\beta} + r^{-2}(n+1)S_{\alpha\beta}\\
&= \left(W + (n+1)r^{-2}\right)S_{\alpha\beta},\end{align*}
where we have used
\[ \mathring\Delta \mathcal{D} P = \mathcal{D}\mathring\Delta P + (n+1)\mathcal{D} P,\]
applied the definition \eqref{CovectorWave}, and used the wave equation \eqref{RW2}.  That is, the spin-raised quantity $S_{\alpha\beta}$ satisfies
\begin{equation}\label{spinraised}
\slashed\Box_{\mathcal{L}(-2)} S_{\alpha\beta} = V^{(-)} S_{\alpha\beta},
\end{equation}
with $V^{(-)}$ defined by \eqref{RWPotential}.

%%%%%%%%%%%%%%%%%%%%%%%%%%%%%%%%%%%%%%%%%%%
\subsection{Analysis of the Master Equation}
We analyze solutions $\Psi$ of the Regge-Wheeler type equation \eqref{RW3}, including both $S_{\alpha\beta}$ and $Q^{(-)}_{\alpha\beta}.$

%%%%%%%%%%%%%%%%%%%%%%%%%%%%%%%%%%%%%%%%%%%
\subsubsection{The $T$-Energy Comparison}
The $T$-energy has the form
\begin{align*}
E_{\Psi}^T(\Sigma_{\tau})&=\frac{1}{2}\int_{S^n}\int_{r_{h}}^\infty \left[(1-\mu)|\slashed{\nabla}_{R}\Psi|^2+V^{(-)}|\Psi|^2\right]R^{n}drdVol_{S^n}\\
&+\frac{1}{2}\int_{S^n}\int_{r_{h}}^\infty \left[\cosh^{-2} x(1-\mu)^{-1}|\nablas_{\tau}\Phi|^2+R^{-2}|\mathring{\nablas}\Phi|^2\right]  R^{n}dr dVol_{S^n}.
\end{align*}

Performing the change of variables $s = R^{n-1}$, we evaluate the quantity
\[ \int_{r_{h}}^{\infty} \left[s(s-2M)\left(\frac{df}{ds}\right)^2 + \left(\frac{n+2}{(n-1)^2}-\frac{2Mn^2}{(n-1)^2s}\right)f^2\right]ds.\]
Choosing $F = \frac{n}{n-1}$ and $E = 2F - 1 = \frac{n+1}{n-1}$, and noting $\frac{n+2}{(n-1)^2} \geq \frac{1}{4}(E^2-1),$  Lemma \ref{HardyTEstimate} implies positivity of the expression above.

Choosing $\epsilon(n) > 0$ small, we find the lower bound
\begin{align*}
E_{\Psi}^T(\Sigma_{\tau})&\geq \frac{1}{2}\int_{S^n}\int_{r_{h}}^\infty \epsilon(n)\left[(1-\mu)|\slashed{\nabla}_{R}\Psi|^2+V^{(-)}|\Psi|^2\right]R^{n}drdVol_{S^n}\\
&+\frac{1}{2}\int_{S^n}\int_{r_{h}}^\infty \left[\cosh^{-2} x(1-\mu)^{-1}|\nablas_{\tau}\Phi|^2+R^{-2}|\mathring{\nablas}\Phi|^2\right]  R^{n}dr dVol_{S^n}\\
&\gtrsim \check{E}^{T}_{\Psi}(\Sigma_{\tau}),
\end{align*}
with the reversed comparison in \eqref{Tcomparison} following trivially.  We conclude that $\Psi$ satisfies the uniform boundedness estimate of Theorem \ref{uniformBoundedness} in all spacetime dimensions.

%%%%%%%%%%%%%%%%%%%%%%%%%%%%%%%%%%%%%%%%%%%%%%%
\subsubsection{The Morawetz Estimate}\label{covectorMorawetz}

Borrowing from the angular term via the Poincar\'{e} inequality \eqref{PoincarePsi}, with least eigenvalue $\lambda = n$, and applying \eqref{Xweightedweighted}, we find
\[\textup{div} \left( J^{X,\omega^X}[\Psi]+ \frac{f'}{1-\mu}\beta |\Psi|^2\partial_{r_*}\right)\geq \frac{f'}{1-\mu}|\slashed{\nabla}_{r_*}\Psi+\beta\Psi|^2+W|\Psi|^2,\]
where
\begin{align*}
W(f,\beta)&=\left(-\frac{1}{4}\Box\omega^X-\frac{1}{2}V^{(-)'} f+\frac{1}{2}\mu_r V^{(-)} f  \right)\\
 &+\left( \frac{f''\beta}{1-\mu}+f'\left( \frac{\beta'}{1-\mu}+\frac{n\beta}{r}-\frac{\beta^2}{1-\mu} \right) \right)\\
 &+\frac{nf}{r^3}\left( 1-\frac{(n+1)M}{r^{n-1}} \right).
\end{align*}
For $n=3$, choosing
\begin{align*}
f&=\left( 1-\frac{(n+1)M}{r^{n-1}} \right)\left( 1-\frac{M}{r^{n-1}}+\frac{4M^2}{5r^{2n-2}} \right)\\
 &=\left( 1-\frac{4M}{r^{2}} \right)\left( 1-\frac{M}{r^{2}}+\frac{4M^2}{5r^{4}} \right),\\
\beta&=\frac{1}{2},
\end{align*}
gives $f'\geq 0$ and $W(f,\beta)>0$. For $n = 4$, choices of
\begin{align*}
f&=\left( 1-\frac{(n+1)M}{r^{n-1}} \right)\left( 1-\frac{2M}{r^{n-1}}+\frac{8M^2}{5r^{2n-2}} \right)\\
 &=\left( 1-\frac{5M}{r^{3}} \right)\left( 1-\frac{2M}{r^{3}}+\frac{8M^2}{5r^{6}} \right),\\
\beta&=1+\frac{M}{r^{3}},
\end{align*}
do the same.  Furthermore, $W\approx r^{-3}$ near infinity. 

Let $J[\Psi]=J^{X,\omega}[\Psi]+\frac{f'}{1-\mu}\beta |\Psi|^2\frac{\partial}{\partial r_*}$. For any $\epsilon_0>0$ we have
\begin{align*}
\textup{div}J[\Psi] &=\frac{f'}{1-\mu}|\slashed{\nabla}_{r_*}\Psi+\beta \Psi|^2+\frac{f}{r^3}\left(1-\frac{(n+1)M}{r^{n-1}}\right)|\mathring{\slashed{\nabla}}\Psi|^2+\mathring{W}|\Psi|^2\\
             &=\frac{\epsilon_0 f'}{1-\mu}|\slashed{\nabla}_{r_*}\Psi|^2+\frac{f'}{1-\mu}|(1-\epsilon_0)^{1/2}\slashed{\nabla}_{r_*}\Psi+(1-\epsilon_0)^{-1/2}\beta \Psi|^2\\
             &+\frac{f}{r^3}\left(1-\frac{(n+1)M}{r^{n-1}}\right)|\mathring{\slashed{\nabla}}\Psi|^2+\mathring{W}|\Psi|^2-\frac{\epsilon_0}{1-\epsilon_0}\frac{\beta^2f'}{1-\mu}|\Psi|^2\\
             &\geq \frac{\epsilon_0 f'}{1-\mu}|\slashed{\nabla}_{r_*}\Psi|^2+\frac{f'}{1-\mu}|(1-\epsilon_0)^{1/2}\slashed{\nabla}_{r_*}\Psi+(1-\epsilon_0)^{-1/2}\beta \Psi|^2\\
             &+\frac{1}{2}\left( \frac{f}{r^3}\left(1-\frac{(n+1)M}{r^{n-1}}\right)|\mathring{\slashed{\nabla}}\Psi|^2+\mathring{W}|\Psi|^2 \right)\\
             &+\left(\frac{1}{2} W-\frac{\epsilon_0}{1-\epsilon_0}\frac{\beta^2f'}{1-\mu} \right)|\Psi|^2.
\end{align*}
For $r \geq r_0$, $W \gtrsim r^{-3}$  and $\frac{\beta f'}{1-\mu} \approx r^{-n}$; by choosing $\epsilon_0$ small enough, the last term is positive in this region. To control the angular derivative, notice that
\begin{align*}
&\frac{f}{r^3}\left(1-\frac{(n+1)M}{r^{n-1}}\right)|\mathring{\slashed{\nabla}}\Psi|^2+\mathring{W}|\Psi|^2\\
&\geq \epsilon_1\frac{f}{r^3}\left(1-\frac{(n+1)M}{r^{n-1}}\right)|\mathring{\slashed{\nabla}}\Psi|^2+W|\Psi|^2-\epsilon_1\frac{nf}{r^3}\left(1-\frac{(n+1)M}{r^{n-1}}\right)|\Psi|^2.
\end{align*}
Again, $W \gtrsim r^{-3}$ for $r\geq r_0/2$ and 	$\frac{nf}{r^3}\left(1-\frac{(n+1)M}{r^{n-1}}\right) \approx r^{-3}$.  By choosing $\epsilon_1$ small, we have as $r\geq r_0/2$
\begin{align*}
&\frac{f}{r^3}\left(1-\frac{(n+1)M}{r^{n-1}}\right)|\mathring{\slashed{\nabla}}\Psi|^2+\mathring{W}|\Psi|^2\\
&\geq \epsilon_1\frac{f}{r^3}\left(1-\frac{(n+1)M}{r^{n-1}}\right)|\mathring{\slashed{\nabla}}\Psi|^2+\frac{1}{2} W|\Psi|^2.
\end{align*}
Hence for $r\geq r_0$,
\begin{align*}
\textup{div}J[\Psi]\geq \frac{\epsilon_0 f'}{1-\mu}|\slashed{\nabla}_{r_*}\Psi|^2+\frac{\epsilon_1 f}{2r^3}\left(1-\frac{(n+1)M}{r^{n-1}}\right)|\mathring{\slashed{\nabla}}\Psi|^2+\frac{1}{4}W |\Psi|^2.
\end{align*}
To obtain a $|\slashed{\nabla}_{t}\Psi|^2$ term, let $\chi$ be a cut-off function with $\chi=0$ as $r\leq r_0$ and $\chi=1$ as $r\geq (r_0+r_1)/2$. Define
\[g(r):=\epsilon_2\frac{\chi}{r^n}\left(1-\frac{(n+1)M}{r^{n-1}} \right)^2,\]
where $\epsilon_2>0$ is a constant to be determined. Denoting 
\[\tilde{J}_{a}[\Psi] :=J_{a}[\Psi]+\frac{1}{4}(\nabla_{a} g)|\Psi|^2-\frac{1}{4}g\nabla_{a} |\Psi|^2,\]
we calculate
\begin{align*}
\textup{div}\tilde{J}[\Psi]&=\textup{div} J[\Psi]+\frac{1}{4}\Box g |\Psi|^2-\frac{1}{2} g (\nablas \Psi)^2-\frac{1}{2} g V^{(-)} |\Psi|^2\\
                     &=\textup{div} J[\Psi]+\frac{1}{2}g\left( \frac{1}{1-\mu}|\slashed{\nabla}_{t}\Psi|^2-\frac{1}{1-\mu}|\slashed{\nabla}_{r_*}\Psi|^2-\frac{1}{r^2}|\mathring{\slashed{\nabla}}\Psi|^2 \right)\\
                     &+\left( \frac{1}{4}\Box g-\frac{1}{2}gV^{(-)} \right)|\Psi|^2.
\end{align*}
Together with previous estimate, for $r\geq r_0$
\begin{align*}
\textup{div}\tilde{J}[\Psi] &\geq \frac{g}{2(1-\mu)}|\slashed{\nabla}_{t}\Psi|^2+ \left(\frac{\epsilon_0 f'}{1-\mu}-\frac{g}{2(1-\mu)}\right)|\slashed{\nabla}_{r_*}\Psi|^2\\
                        &+\left(\frac{\epsilon_1 f}{r^3}\left(1-\frac{(n+1)M}{r^{n-1}}\right)-\frac{g}{2r^2}   \right) |\mathring{\slashed{\nabla}}\Psi|^2\\
                        &+\left(\frac{1}{4}W +\frac{1}{4}\Box g-\frac{1}{2}gV^{(-)} \right)|\Psi|^2.
\end{align*}
For large radii, $f' \approx g \approx r^{-n}$. By choosing $\epsilon_2$ small, the coefficient of $|\slashed{\nabla}_{r_*}\Psi|^2$ is greater than $\frac{\epsilon_0 f'}{2(1-\mu)}$. Both $\frac{\epsilon_1 f}{r^2}\left(1-\frac{(n+1)M}{r^{n-1}}\right)$ and $g$ vanish quadratically at the photon sphere, and the former, comparable to $r^{-3}$, decays no slower than the latter, comparable to $r^{-n}$. Hence the coefficient of $|\mathring{\slashed{\nabla}}\Psi|^2$ can also be made positive with $\epsilon_2$ small.  The term $\frac{1}{4}\Box g-\frac{1}{2}gV^{(-)}$ behaves like $r^{-n-2}$ toward infinity; yet again, positivity of the coefficient for $|\Psi|^2$ is ensured by a choice of small $\epsilon_2$.  With such a choice, we have for $r\geq r_0$
\begin{align*}
\textup{div}\tilde{J}[\Psi]\gtrsim\frac{1}{r^n}\left( \left(1-\frac{(n+1)M}{r^{n-1}}\right)^2\left(|\slashed{\nabla}_{t}\Psi|^2+|\mathring{\slashed{\nabla}}\Psi|^2 \right)+|\slashed{\nabla}_{r}\Psi|^2 \right).
\end{align*}
In $r\leq r_0$, $\textup{div}\tilde{J}[\Psi]=\textup{div}{J}[\Psi]\geq 0$. By adding $\epsilon_3 \check{T}_{ab}[\Psi]Y^b$ into $\tilde{J}_{a}[\Psi]$, we obtain a Morawetz current
\begin{equation}
\begin{split}
J^{\mathfrak{M}}_a[\Psi]&:=\tilde{J}_a[\Psi]+\epsilon_3\check{T}_{ab}[\Psi]Y^b\\
     &=T_{ab}[\Psi]X^b-\frac{1}{2}V^{(-)}|\Psi|^2X_a+\frac{1}{4}(\omega^X-g)\nabla_a |\Psi|^2\\
     &-\frac{1}{4}\nabla_a (\omega^X-g)|\Psi|^2 +f'\beta|\Psi|^2dR_{*}+\epsilon_3\check{T}_{ab}[\Psi]Y^b.
\end{split}
\end{equation}
By choosing $\epsilon_3$ small enough, $K^{\mathfrak{M}}[\Psi] := \Div J^{\mathfrak{M}}[\Psi]$ has the desired positivity. To estimate $J^{\mathfrak{M}}_a[\Psi] n^a_{\Sigma_\tau}$, note that $X$ and $Y$ are regular vector fields, such that
\[|(T_{ab}[\Psi]X^b-\frac{1}{2}V^{(-)}|\Psi|^2X_a+\epsilon_3 \check{T}_{ab}[\Psi]Y^b)n^a_{\Sigma_\tau}|\lesssim T_{ab}[\Psi]n^a_{\Sigma_\tau}N^b.\]
From $\omega^X-g\approx r^{-1}$ for large radii, the third and the fourth terms are controlled. The estimate for the fifth term follows from regularity of $\frac{f'}{1-\mu}=\frac{\partial f}{\partial r}$ at the horizon in addition to its comparability with $r^{-n}$ for large radii. The estimate for $J^{\mathfrak{M}}_a[\Psi] \underline{L}^a$ is similar.  For $J^{\mathfrak{M}}_a[\Psi] n^a_{\mathcal{H}^+}$, we note that on the horizon $\frac{\partial}{\partial r_*}$ and $n^a_{\mathcal{H}^+}$ are proportional to $T$ and that $\omega^X-g=0$.  Except for $\epsilon_3 \check{T}_{ab}[\Psi]Y^b n^a_{\mathcal{H}^+}$, satisfying $\check{T}_{ab}[\Psi]Y^b n^a_{\mathcal{H}^+}\geq 0$, the terms in $J^{\mathfrak{M}}_a[\Psi] n^a_{\mathcal{H}^+}$ are bounded by $\check{T}_{ab}[\Psi]T^b n^a_{\mathcal{H}^+}$, from which Lemma \ref{Morawetz} follows.

We remark that the construction of a current $\tilde{J}_{a}[\Psi]$ with positivity in $r\leq r_0$, an essential feature in the construction above, is impossible for $n \geq 5$, i.e. in spacetime dimension seven and above.  In this regime, a more refined analysis is needed to form a Morawetz current satisfying the necessary density estimate.  This difficulty is the reason for our restriction to dimensions six and below in our decay estimates.

\subsection{Uniform Boundedness and Decay of the Master Quantities}

With the $T$-comparison, we have uniform boundedness of solutions to the Regge-Wheeler equation \eqref{RW2} in all spacetime dimensions; in addition, the Morawetz estimate for $n \leq 4$ gives uniform decay of solutions in six and fewer spacetime dimensions.

\begin{theorem}
Let $\delta g$ be a smooth, symmetric two-tensor on a Schwarzschild-Tangherlini spacetime, satisfying the linearized Einstein equation \eqref{linEinstein}.  There exists a gauge-invariant master quantities $Q^{(-)}_{\alpha\beta}$ and $S_{\alpha\beta}$ in the co-vector portion $h_2$ of $\delta g$ satisfying the Regge-Wheeler type equation \eqref{RW2}.  As solutions of \eqref{RW2}, $Q^{(-)}_{\alpha\beta}$ and $S_{\alpha\beta}$ satisfy the uniform boundedness estimate
\begin{align}
\begin{split}
\check{E}^{N}_{Q^{(-)}}(\Sigma_\tau) \lesssim \check{E}^{N}_{Q^{(-)}}(\Sigma_0),\\
\check{E}^{N}_{S}(\Sigma_\tau) \lesssim \check{E}^{N}_{S}(\Sigma_0),
\end{split}
\end{align}
in all spacetime dimensions.  In six and fewer spacetime dimensions, $Q^{(-)}_{\alpha\beta}$ and $S_{\alpha\beta}$ satisfy the uniform decay estimate
\begin{align}
\begin{split}
\check{E}^{N}_{Q^{(-)}}(\Sigma_{\tau}) \lesssim \frac{I_{Q^{(-)}}(\Sigma_0)}{\tau^2},\\
\check{E}^{N}_{S}(\Sigma_{\tau}) \lesssim \frac{I_{S}(\Sigma_0)}{\tau^2},\\
\end{split}
\end{align}
where 
\begin{equation} 
I_{Q^{(-)}}(\Sigma_0):=E^2_{Q^{(-)}, \mathcal{L}_{K}Q^{(-)}}(\Sigma_0)+E^N_{Q^{(-)}, \mathcal{L}_{K}Q^{(-)}, \mathcal{L}^2_{K}Q^{(-)}}(\Sigma_0),
\end{equation}
and similarly for $S$.  Here we assume $\tau \geq 0$.
\end{theorem}

We emphasize that the relevant constants in the comparisons depend only upon the orbit sphere dimension $n$ and the mass $M > 0$.

%%%%%%%%%%%%%%%%%%%%%%%%%%%%%%%%%%%%%%%%%%%%%

\section{The Scalar Portion}

%%%%%%%%%%%%%%%%%%%%%%%%%%%%%%%%%%%%%%
\subsection{The Linearized Einstein Equations}

In this subsection we reduce the linearized vacuum Einstein equations for the scalar solution, modifying the approach of Kodama-Ishibashi \cite{KI2}.  

The authors consider the linearized Einstein tensor
\begin{equation}
\delta E_{ab} dx^{a}dx^{b} = \delta E_{AB}dx^{A}dx^{B} + 2\delta E_{A\alpha} dx^{A}dx^{\alpha} + \delta E_{\alpha\beta} dx^{\alpha}dx^{\beta},
\end{equation}
admitting the same Hodge decomposition as the linearized Ricci tensor above.  Defining 
\begin{align}
e_{A} &:= H_{A} - \frac{1}{2}r^2\tilde\nabla_{A}\left(r^{-2}H_2\right),\\
\tilde{h}_{AB} &:= h_{AB} - \tilde\nabla_{A}e_{B} - \tilde\nabla_{B}e_{A},\\
\tilde{H} &:= \frac{1}{2r^2}\left(H-\frac{1}{n}\mathring\Delta H_2 - 2rr^{A}e_{A}\right),
\end{align}
the linearized vacuum Einstein equations for the scalar portion can be expressed in terms of the gauge-invariant quantities $\tilde{h}_{AB}$ and $\tilde{H}$.  Following \cite{KI2}, we further rescale the quantities as
\begin{align}
\begin{split}
\delta\check{E}_{ab} &:= r^{n-2}\delta E_{ab},\\
\check{h}_{AB} &:= r^{n-2}\tilde{h}_{AB},\\
\check{H} &:= r^{n-2}\tilde{H}.
\end{split}
\end{align}
As noted in the same work, owing to the Bianchi identities, it suffices to consider the rescaled linearized Einstein tensor components
\[ \delta\check{E}_{A\alpha},\ \delta\check{E}_{\alpha\beta} - \frac{1}{n}\left(\mathring\sigma^{\gamma\delta}\delta\check{E}_{\gamma\delta}\right)\mathring\sigma_{\alpha\beta},\ \delta\check{E}_{tr},\ \delta\check{E}_{rr}.\]

With the aim of rewriting these remaining linearized equations, we define the further gauge-invariants
\begin{align}
\begin{split}
X &= \check{h}_{t}^{t} - 2\check{H},\\
Y &= \check{h}_{r}^{r} - 2\check{H},\\
Z & = \check{h}_{t}^{r}.
\end{split}
\end{align}

Using the algebraic relation provided by the traceless equation
\[ \delta\check{E}_{\alpha\beta} - \frac{1}{n}\left(\mathring\sigma^{\gamma\delta}\delta\check{E}_{\gamma\delta}\right)\mathring\sigma_{\alpha\beta} = 0,\]
we can invert these relations and find
\begin{align}\label{inversion}
\begin{split}
\check{h}_{t}^{t} &= \frac{(n-1)X - Y}{n},\\
\check{h}_{r}^{r} &= \frac{-X + (n-1)Y}{n},\\
\check{h}_{t}^{r} &= Z,\\
\check{H} &= -\frac{X+Y}{2n}.
\end{split}
\end{align}

We rewrite the linearized Einstein equations using scalar spherical harmonic expansion with indices $\ell \geq 2$ and $m_s(n,\ell) \in \{1,\hdots ,d_s(n,\ell)\}$, with the dimension of the eigenspaces $d_s(n,\ell)$ given by \eqref{scalarEigenspace}, and the inversion \eqref{inversion}.  First, the cross-term equations $\delta\check{E}_{A\alpha} = 0$ imply
\begin{align}\label{crossEinstein}
\begin{split}
\partial_{r}Z_{\ell m_s(n,\ell)} &= -\partial_{t}X_{\ell m_s(n,\ell)},\\
\partial_{r}Y_{\ell m_s(n,\ell)} &= \frac{\partial_{r}f}{2f}(X_{\ell m_s(n,\ell)} - Y_{\ell m_s(n,\ell)}) + \frac{1}{f^2}\partial_{t}Z_{\ell m_s(n,\ell)},
\end{split}
\end{align}
where we adopt the notation $f := 1-\mu$ of \cite{KI2}.

Substituting \eqref{crossEinstein}, the quotient equation $\delta\check{E}^{r}_{r} = 0$ can be rewritten as
\begin{align}\label{quotEinsteinOne}
\begin{split}
&\partial_{r}X_{\ell m_s(n,\ell)} = \frac{2}{\partial_{r}f}\Bigg(\frac{n-1}{r^2}(f-1)+\frac{(n+2)\partial_{r}f}{2r}+\frac{\partial_{r}^2f}{n}\\
&-\left(\frac{\partial_{r}f}{2}+\frac{nf}{r}\right)\frac{\partial_{r}f}{2f}\Bigg)X_{\ell m_s(n,\ell)} +\frac{2}{\partial_{r}f}\Bigg(\frac{1-f}{r^2} - \frac{(3n-2)}{2r}\partial_{r}f \\
&- \frac{(n-1)}{n}\partial_{r}^2f + \frac{\ell(\ell+n-1)-n}{r^2}+\left(\frac{\partial_{r}f}{2}+\frac{nf}{r}\right)\frac{\partial_{r}f}{2f}\Bigg)Y_{\ell m_s(n,\ell)}\\
&+\frac{2}{\partial_{r}f}\left(\frac{2n}{rf} - \frac{1}{f^2}\left(\frac{\partial_{r}f}{2}+\frac{nf}{r}\right)\right)\partial_{t}Z_{\ell m_s(n,\ell)}\\
&+ \frac{2}{f\partial_{r}f}\left(\partial_{t}^2X_{\ell m_s(n,\ell)} + \partial_{t}^2Y_{\ell m_s(n,\ell)}\right).
\end{split}
\end{align}

Finally, the quotient equation $\delta\check{E}^{r}_{t} = 0$ has the form
\begin{align}\label{quotEinsteinTwo}
\begin{split}
&\left(\frac{\ell(\ell+n-1)}{r^2} - \partial_{r}^2f - \frac{n\partial_{r}f}{r}\right)Z_{\ell m_s(n,\ell)}\\
& + f\left(\partial_{t}\partial_{r}X_{\ell m_s(n,\ell)} + \partial_{t}\partial_{r}Y_{\ell m_s(n,\ell)}\right)\\
&-\left(\frac{(n-2)}{r}f + \frac{\partial_{r}f}{2}\right)\partial_{t}X_{\ell m_s(n,\ell)} + \left(\frac{2f}{r} - \frac{\partial_{r}f}{2}\right)\partial_{t}Y_{\ell m_s(n,\ell)} = 0.
\end{split}
\end{align}

%%%%%%%%%%%%%%%%%%%%%%%%%%%%%%%%%%%%%%
\subsection{Master Equation for the Scalar Portion}

Let us define the further gauge-invariant quantity
\begin{align}
\begin{split}
\tilde{Z}_{\ell m_s(n,\ell)} &:=  \frac{1}{\left(\frac{\ell(\ell+n-1)}{r^2} - \partial_{r}^2f - \frac{n\partial_{r}f}{r}\right)}\Bigg(f\left(\partial_{r}X_{\ell m_s(n,\ell)} + \partial_{r}Y_{\ell m_s(n,\ell)}\right)\\
&-\left(\frac{(n-2)}{r}f + \frac{\partial_{r}f}{2}\right)X_{\ell m_s(n,\ell)} + \left(\frac{2f}{r} - \frac{\partial_{r}f}{2}\right)Y_{\ell m_s(n,\ell)}\Bigg),
\end{split}
\end{align}
in addition to
\begin{equation}
\Phi_{\ell m_s(n,\ell)} = \frac{n\tilde{Z}_{\ell m_s(n,\ell)} - r(X_{\ell m_s(n,\ell)} + Y_{\ell m_s(n,\ell)})}{r^{n/2-1}(\ell(\ell+n-1)-n+\frac{1}{2}n(n+1)\mu)},
\end{equation}
slightly modifying those definitions provided in \cite{KI2}.  Substituting \eqref{quotEinsteinOne} and \eqref{crossEinstein} into $\tilde{Z}_{\ell m_s(n,\ell)}$, we regard $\Phi_{\ell m_s(n,\ell)}$ as an expression in $X_{\ell m_s(n,\ell)},$ $Y_{\ell m_s(n,\ell)},$ $\partial_{t}Z_{\ell m_s(n,\ell)},$ $\partial_{t}^2 X_{\ell m_s(n,\ell)},$ $\partial_{t}^2 Y_{\ell m_s(n,\ell)}$.

As described below, $\Phi_{\ell m_s(n,\ell)}$ satisfies the following scalar wave equation:
\begin{equation}
\tilde\Box \Phi_{\ell m_s(n,\ell)} = \tilde{V}^{(+)}_{\ell} \Phi_{\ell m_s(n,\ell)},
\end{equation}
where
\begin{equation}
\tilde{V}^{(+)}_{\ell} := \frac{Q_{\ell}}{16r^2\left(q+n(n+1)\mu/2\right)^2},
\end{equation}
with
\[ q := \ell(\ell + n - 1) - n\]
and
\begin{align}
\begin{split}
Q_{\ell} &:= n^4(n+1)^2\mu^3 + n(n+1)\Big[4(2n^2-3n+4)q\\
&+n(n-2)(n-4)(n+1)\Big]\mu^2-12n\Big[(n-4)q\\
&+n(n+1)(n-2)\Big]q\mu +16q^3+4n(n+2)q^2.
\end{split}
\end{align}
Note that the equation reduces to the well-known Zerilli equation in the case $n = 2$.  

Suppressing the angular harmonic dependence where possible, the equation in standard Schwarzschild coordinates is tantamount to
\[ \partial_{r}\left(f\partial_{r}\Phi_{\ell}\right) = \tilde{V}^{(+)}_{\ell}\Phi_{\ell} + \frac{1}{f}\partial_{t}^2\Phi_{\ell}.\]
Again, we regard $\Phi_{\ell}$ as an expression in $X, Y, \partial_{t}Z, \partial_{t}^2 X, \partial_{t}^2 Y$, and apply the linearized Einstein equations (\ref{crossEinstein},\ \ref{quotEinsteinOne},\ \ref{quotEinsteinTwo}) to the left-hand side of the equation above as follows.  Differentiating in $r$, we replace radial derivatives on $X, Y, \partial_{t}Z$ using \eqref{quotEinsteinOne} and \eqref{crossEinstein}, respectively.  Differentiation of the terms $\partial_{t}^2X$ and $\partial_{t}^2Y$ results in mixed partials, with the same coefficient on each quantity; we replace these terms using \eqref{quotEinsteinTwo}.  After this first differentiation and substitution, we are again left with an expression in $X, Y, \partial_{t}Z, \partial_{t}^2 X, \partial_{t}^2 Y$.  Multiplying by $f$ and differentiating the product in $r$, we again replace radial derivatives on $X, Y, \partial_{t}Z$ using \eqref{quotEinsteinOne} and \eqref{crossEinstein}.  The situation for the mixed partial terms is more subtle.  Regarding the mixed term obtained via radial derivatives on $\partial_{t}^2X$, we first substitute \eqref{quotEinsteinOne} for an appropriate portion in order to match the $\partial_{t}^4X + \partial_{t}^4Y$ term in the right-hand side of the equation above.  For the remainder of the term, we add and subtract appropriate radial derivatives on $\partial_{t}^2Y$ and apply \eqref{quotEinsteinTwo} to the matching mixed partial terms in $X$ and $Y$.  Finally, the residual radial derivatives on $\partial_{t}^2Y$ are handled via substitution with \eqref{crossEinstein}.  The resulting quantity, an expression in $X, Y, \partial_{t}Z, \partial_{t}^2 X, \partial_{t}^2 Y$ and their second time derivatives, is equal to the right-hand side of the equation above, as can be verified by direct calculation.

%%%%%%%%%%%%%%%%%%%%%%%%%%%%%%%%%%%%%%%%%
\subsection{The Spin-Raised Equation}

We spin-raise the equation by associating $\Phi_{\ell m_s(n,\ell)}$ with a symmetric traceless two-tensor $Q^{(+)}_{\ell m_s(n,\ell), \alpha\beta}$, specified by
\begin{equation}\label{QPlusDef}
Q^{(+)}_{\ell m_s(n,\ell), \alpha\beta} := \left(r^{k}\Phi_{\ell m_s(n,\ell)}\right)Y_{\alpha\beta}^{\ell m_s(n,\ell)}dx^{\alpha}dx^{\beta},
\end{equation}
where 
\[k := \frac{4-n}{2}\]
and $Y_{\alpha\beta}^{\ell m_s(n,\ell)}$ are the tensor spherical harmonics \eqref{scalarTensorHarmonics}.

We calculate
\begin{align*}
\slashed\Box_{\mathcal{L}(-2)} Q^{(+)}_{\ell m_s(n,\ell), \alpha\beta} &= \Big[r^k\tilde\Box \Phi_{\ell m_s(n,\ell)} + \tilde\Box(r^k)\Phi_{\ell m_s(n,\ell)}\\
& + (n-4)kr^{-2}r^{A}r_{A}r^k\Phi_{\ell m_s(n,\ell)}\Big]Y_{\alpha\beta}^{\ell m_s(n,\ell)}\\
&+ r^{-2}r^k\Phi_{\ell m_s(n,\ell)}\left(2n-\ell(\ell+n-1)\right)Y_{\alpha\beta}^{\ell m_s(n,\ell)}\\
&+(6-2n)r^{-2}r^{A}r_{A}Q^{(+)}_{\ell m_s(n,\ell), \alpha\beta} -2r^{-1}\left(\tilde\Box r \right)Q^{(+)}_{\ell m_s(n,\ell), \alpha\beta}\\
&=\Big[\tilde{V}^{(+)}_{\ell} + k(k-1)r^{B}r_{B}r^{-2}+kr^{-1}\tilde\Box r \\
&+ (n-4)r^{-2}r^{B}r_{B} + \left(2n-\ell(\ell + n-1)\right)r^{-2}\\
&+(6-2n)r^{-2}r^{B}r_{B}-2r^{-1}\left(\tilde\Box r\right)\Big]Q^{(+)}_{\ell m_s(n,\ell), \alpha\beta},
\end{align*}
where we have used
\[\tilde\Box (r^k) = k(k-1)r^{k-2}r^{A}r_{A} + kr^{k-1}\tilde\Box r.\]

Simplifying, we obtain the master equation
\begin{equation}\label{RW4}
\slashed{\Box}_{\mathcal{L}(-2)} Q^{(+)}_{\ell m_s(n,\ell), \alpha\beta} = V^{(+)}_{\ell}Q^{(+)}_{\ell m_s(n,\ell), \alpha\beta},
\end{equation}
where 
\[V^{(+)}_{\ell}= V_{1,n\ell}+V_{2,n\ell}+V_{3,n\ell}\]
is given by
\begin{align}
\begin{split}
r^2V_{1,n\ell}&=-(\ell (\ell+n-1)-2n),\\
r^2V_{2,n\ell}&=\left(\frac{n^2-10n+16}{4}-\frac{3n^2-12n+16}{4}\left(\frac{2M}{r^{n-1}}\right)\right),
\end{split}
\end{align}
and
\begin{align}
\begin{split}
4r^2V_{3,n\ell}=\frac{1}{D_{n\ell}^2}&\Bigg[ (\ell-1)^2(n+\ell)^2n(n+2)+4(\ell-1)^3(n+\ell)^3\\
&-\Big(6(\ell-1)(n+\ell)(n-2)n^2(n+1)M\\
&+6(\ell-1)^2(n+\ell)^2(n-4)nM\Big)r^{-n+1}\\
&+\Big(4(\ell-1)(n+\ell)n(n+1)(2n^2-3n+4)M^2 \\
&+(n-4)(n-2)n^2(n+1)^2M^2\Big)r^{-2n+2} \\
&+2n^4(n+1)^2M^3r^{-3n+3}\Bigg],
\end{split}
\end{align}
with
\begin{equation}
D_{n\ell}=n(n+1)Mr^{-n+1}+(\ell-1)(n+\ell).
\end{equation}

Note that the potentials converge as
\begin{equation}\label{limitPotential}
\lim_{\ell \to \infty} V_{\ell}^{(+)} = V^{(+)} := \frac{1}{2r^2}\left((n^2-2n+8)-(5n^2-10n+8)\frac{2M}{r^{n-1}}\right).
\end{equation}

\subsection{Analysis of the Master Equation}

\subsubsection{The $T$-Energy Comparison}

The $T$-energy comparison for the Zerilli potential is more involved than the earlier equations.  The argument naturally splits into two regimes, $n\geq 4$ and $n \leq 3.$

Turning to the higher dimensional analysis, with $n \geq 4$, we make use of the following inequalities:

For $\alpha,\beta$, we have
\[ (\alpha+\beta)^2 \geq 4\alpha\beta,\]
and for $x, y$ and $\lambda_1, \lambda_2 \geq 0$, 
\[\frac{\lambda_1 x^2 + \lambda_2 y^2}{(x+y)^2} \geq \frac{\lambda_1 \lambda_2}{\lambda_1 + \lambda_2}.\]
Moreover, we let $s = r^{n-1}$.

Grouping terms in $V^{3}_{n\ell}$,
\begin{align*}
&4r^2V^3_{n\ell}=\frac{1}{D_{n\ell}^2}\Bigg[ \underbrace{(\ell-1)^2(n+\ell)^2n(n+2)}_{\textup{I}}+\underbrace{4(\ell-1)^3(n+\ell)^3}_{\textup{II}}+2n^4(n+1)^2M^3r^{-3n+3}\\
                                 &+\Big( \underbrace{(n-4)(n-2)n^2(n+1)^2M^2}_{\textup{I}}+\underbrace{4(\ell-1)(n+\ell)n(n+1)(2n^2-3n+4)M^2}_{\textup{II}} \Big)r^{-2n+2}\\
                                 &-\Big(\underbrace{6(\ell-1)(n+\ell)(n-2)n^2(n+1)M}_{\textup{III}}+6(\ell-1)^2(n+\ell)^2(n-4)nM \Big)r^{-n+1}\Bigg]
\end{align*}
we estimate each of the pieces as follows.

For I and II, we apply the second inequality to deduce
\begin{align*}
\textup{I}/D_{n\ell}^2 &= \frac{\overbrace{(\ell-1)^2(n+\ell)^2}^{x^2}\overbrace{n(n+2)}^{\lambda_1} + \overbrace{n^2(n+1)^2M^2s^{-2}}^{y^2}\overbrace{(n-4)(n-2)}^{\lambda_2}}{((\ell-1)(n+\ell)+n(n+1)Ms^{-1})^2}\\
&\geq \frac{n(n+2)(n-4)(n-2)}{n(n+2)+(n-4)(n-2)} = \frac{n(n+2)(n-4)(n-2)}{2(n^2-2n+4)}.
\end{align*}
and
\begin{align*}
\textup{II}/D_{n\ell}^2 &= \frac{(n+\ell)^2(\ell-1)^2(4(n+\ell)(\ell-1)) + n^2(n+1)^2M^2s^{-2}(\frac{4(n+\ell)(\ell-1)(2n^2-3n+4)}{n(n+1)})}{((\ell-1)(n+\ell)+n(n+1)Ms^{-1})^2}\\
&\geq 4(\ell-1)(n+\ell)\frac{2n^2-3n+4}{3n^2-2n+4}.
\end{align*}

For III, we apply the second inequality:
\begin{align*}
-\textup{III}/D_{n \ell}^2&=-\frac{6(\ell-1)(n+\ell)(n-2)n^2(n+1)Ms^{-1}}{\Big(n(n+1)Ms^{-1}+(\ell-1)(n+\ell)\Big)^2}\\
                                            &\geq -\frac{6(\ell-1)(n+\ell)(n-2)n^2(n+1)Ms^{-1}}{4(\ell-1)(n+\ell)n(n+1)Ms^{-1}}\\
                                            &=-\frac{3}{2}n(n-2).
\end{align*}

After this first round of estimates, we have
\begin{align*}
r^2V^3_{n\ell} &\geq \frac{n(n+2)(n-2)(n-4)}{8(n^2-2n+4)} + (\ell - 1)(n+\ell)\frac{2n^2-3n+4}{3n^2-2n+4}-\frac{3}{8}n(n-2)\\
&+\frac{1}{4D_{n\ell}^2}\left[2n^4(n+1)^2M^3s^{-3}-6(\ell-1)^2(n+\ell)^2(n-4)nMs^{-1}\right].
\end{align*}

Defining the quadratic polynomial
\[ P(n) = an^2 + bn + c,\]
with $a, b, c$ as yet unchosen, we rewrite the estimate above as
\begin{align*}
r^2V^3_{n\ell} &\geq \frac{n(n+2)(n-2)(n-4)}{8(n^2-2n+4)} + (\ell - 1)(n+\ell)\frac{2n^2-3n+4}{3n^2-2n+4}-\frac{3}{8}n(n-2)\\
&+\frac{1}{4D_{n\ell}^2}\Big[\underbrace{2n^4(n+1)^2M^3s^{-3}+P(n)(\ell-1)^2(n-\ell)^2Ms^{-1}}_{\textup{IV}}\\
&\underbrace{-P(n)(\ell-1)^2(n+\ell)^2Ms^{-1}-6(\ell-1)^2(n+\ell)^2(n-4)nMs^{-1}}_{\textup{V}}\Big].
\end{align*}

Using the second inequality, we estimate
\begin{align*}
\textup{IV}/D_{n\ell}^2 &= Ms^{-1}\frac{(\ell - 1)^2(n-\ell)^2P(n) + n^2(n+1)^2M^2s^{-2}(2n^2)}{((\ell-1)(n+\ell)+n(n+1)Ms^{-1})^2}\\
&\geq Ms^{-1}\frac{2n^2P(n)}{2n^2+P(n)} = \left(\frac{n^2P(n)}{2n^2+P(n)}\right)\frac{2M}{s},
\end{align*}
where we have assumed that $P(n)$ is non-negative.

Next, we control V with the first inequality:
\begin{align*}
\textup{V}/D_{n\ell}^2 &= \frac{(\ell-1)^2(n+\ell)^2(-6n^2+24n -P(n))Ms^{-1}}{((\ell-1)(n+\ell)+n(n+1)Ms^{-1})^2}\\
&\geq (\ell - 1)(n+\ell)\frac{-6n^2+24n-P(n)}{4n(n+1)}.
\end{align*}

After this second round of estimates, we have
\begin{align*}
r^2V^3_{n\ell} &\geq \frac{n(n+2)(n-2)(n-4)}{8(n^2-2n+4)} -\frac{3}{8}n(n-2) + \frac{1}{4}\left(\frac{n^2P(n)}{2n^2+P(n)}\right)\frac{2M}{s}\\
&+ (\ell - 1)(n+\ell)\left(\frac{2n^2-3n+4}{3n^2-2n+4}+\frac{-6n^2+24n-P(n)}{16n(n+1)}\right)\\
&\geq \frac{n(n+2)(n-2)(n-4)}{8(n^2-2n+4)} -\frac{3}{8}n(n-2) + \frac{1}{4}\left(\frac{n^2P(n)}{2n^2+P(n)}\right)\frac{2M}{s}\\
&+ (n+2)\left(\frac{2n^2-3n+4}{3n^2-2n+4}+\frac{-6n^2+24n-P(n)}{16n(n+1)}\right),
\end{align*}
assuming that 
\[\frac{2n^2-3n+4}{3n^2-2n+4}+\frac{-6n^2+24n-P(n)}{16n(n+1)} \geq 0.\]

As $V^{1}_{n\ell}$ can be accounted for by borrowing from the angular gradient, it remains to consider
\begin{align*}
&\left[r^2V_{n\ell}^{2}+r^2V_{n\ell}^{3}\right]\frac{1}{(n-1)^2}\\
&\geq \Bigg[\frac{n^2-10n+16}{4}+\frac{n(n+2)(n-2)(n-4)}{8(n^2-2n+4)} -\frac{3}{8}n(n-2)\\
&+ (n+2)\left(\frac{2n^2-3n+4}{3n^2-2n+4}+\frac{-6n^2+24n-P(n)}{16n(n+1)}\right)\Bigg]\frac{1}{(n-1)^2}\\
&-\Bigg[\frac{3n^2-12n+16}{4}-\frac{1}{4}\frac{n^2P(n)}{2n^2+P(n)}\Bigg]\frac{1}{(n-1)^2}\frac{2M}{s}.
\end{align*}

Setting $F_n^2 = \Bigg[\frac{3n^2-12n+16}{4}-\frac{1}{4}\frac{n^2P(n)}{2n^2+P(n)}\Bigg]\frac{1}{(n-1)^2}$ and $E_n = 2F_n - 1$, we must check that
\begin{align*}
&\Bigg[\frac{n^2-10n+16}{4}+\frac{n(n+2)(n-2)(n-4)}{8(n^2-2n+4)} -\frac{3}{8}n(n-2)\\
&+ (n+2)\left(\frac{2n^2-3n+4}{3n^2-2n+4}+\frac{-6n^2+24n-P(n)}{16n(n+1)}\right)\Bigg]\frac{1}{(n-1)^2}\geq \frac{1}{4}(E_n^2-1).
\end{align*}

Choosing $P(n) = 2n^2 - 5n + 5$, fulfilling both of the conditions on $P(n)$ above, we obtain positivity for $n \geq 7$ and $n = 4$.  We remark that we can always perturb the estimate slightly and retain $\epsilon(n+\ell)^3(\ell-1)^3$ in II; phrased another way, we can keep a small piece of the angular gradient, as necessary to demonstrating positivity of the $T$-energy.

To extend the estimate to all $n\geq 4$, we need a refinement of the second inequality above.  Namely, given $\lambda_1, \lambda_2 > 0$ and $x, y > 0$, with some $\alpha > 0$ such that $0 < y < \alpha < \frac{\lambda_1}{\lambda_2}x,$  we have
\[\frac{\lambda_1 x^2 + \lambda_2 y^2}{(x+y)^2} \geq (1+\epsilon)\frac{\lambda_1 \lambda_2}{\lambda_1 + \lambda_2},\]
where
\[ \epsilon = \frac{(\lambda_{1}x-\lambda_{2}\alpha)^2}{\lambda_1\lambda_2(x+\alpha)^2}.\]

Specializing to I, with $x = (\ell - 1)(n + \ell), y = n(n+1)Ms^{-1}, \lambda_1 = n(n+2), \lambda_2 = (n-4)(n-2)$, we note that $Ms^{-1} < \frac{1}{2}$, such that $0 < y < \frac{1}{2}n(n+1) < \frac{1}{2}n(n+\ell)(\ell-1)$.  The relation
\[ y < \alpha = \frac{1}{2}n(n+\ell)(\ell-1) < \frac{n(n+2)}{(n-4)(n-2)}(n+\ell)(\ell-1) = \frac{\lambda_1}{\lambda_2}x \]
holds for $5 \leq n \leq 7$, and we obtain the improvement 
\begin{align*}
r^2V^3_{n\ell} &\geq (1+\epsilon)\frac{n(n+2)(n-2)(n-4)}{8(n^2-2n+4)} -\frac{3}{8}n(n-2) + \frac{1}{4}\left(\frac{n^2P(n)}{2n^2+P(n)}\right)\frac{2M}{s}\\
&+ (n+2)\left(\frac{2n^2-3n+4}{3n^2-2n+4}+\frac{-6n^2+24n-P(n)}{16n(n+1)}\right),
\end{align*}
with
\[ \epsilon = \frac{n(n^2-8n+4)^2}{(n-4)(n-2)(n+2)^3}.\]
Using the same choices of $P(n), F_n$ and $E_n$ as above, the result extends to include $n = 5$ and $n = 6$.

Finally, we consider the low dimensions $n = 2$ and $n = 3$.  Here we group 
\begin{align*}
&4r^2V^3_{n\ell}=\frac{1}{D_{n\ell}^2}\Bigg[(\ell-1)^2(n+\ell)^2n(n+2)+\underbrace{4(\ell-1)^3(n+\ell)^3}_{\textup{II}}+2n^4(n+1)^2M^3r^{-3n+3}\\
                                 &+\Big( \underbrace{(n-4)(n-2)n^2(n+1)^2M^2}_{\textup{I}}+\underbrace{4(\ell-1)(n+\ell)n(n+1)(2n^2-3n+4)M^2}_{\textup{II}} \Big)r^{-2n+2}\\
                                 &-\Big(\underbrace{6(\ell-1)(n+\ell)(n-2)n^2(n+1)M}_{\textup{III}}+\underbrace{6(\ell-1)^2(n+\ell)^2(n-4)nM}_{\textup{I}} \Big)r^{-n+1}\Bigg].
\end{align*}

The terms II and III are handled just as before.  The terms in I change sign for these low dimensions, and we estimate
\begin{align*}
\textup{I} &= 6(\ell-1)^2(n+\ell)^2(4-n)nMs^{-1} - (n-2)n^2(n+1)^2(4-n)M^2s^{-2}\\
&\geq 12(n+2)^2(4-n)nM^2s^{-2} - (n-2)n^2(n+1)^2(4-n)M^2s^{-2}\\
&\geq n(4-n)\left(12(n+2)^2-n(n-2)(n+1)^2\right)M^2s^{-2}.
\end{align*}

In this way, we obtain
\begin{align*}
&r^2V^3_{n\ell} \geq (\ell - 1)(n+\ell)\frac{2n^2-3n+4}{3n^2-2n+4}-\frac{3}{8}n(n-2)\\
&+\frac{1}{4D_{n\ell}^2}\left[\underbrace{(\ell-1)^2(n+\ell)^2n(n+2) + n(4-n)\left(12(n+2)^2-n(n-2)(n+1)^2\right)M^2s^{-2}}_{\textup{IV}}\right].
\end{align*}

Applying the second inequality, we find
\[\textup{IV}/D_{n\ell}^2 \geq \frac{n(n+2)(4-n)(48+50n+15n^2-n^4)}{2(n^5-5n^3+6n^2+76n+96)}.\]
With this new lower bound, and the usual choices of $F_{n}$ and $E_{n}$, the result extends to $n =2$ and $n = 3$.

%%%%%%%%%%%%%%%%%%%%%%%%%%%%%%%%%%%%%%
\subsubsection{The Morawetz Estimate}
Borrowing from the angular term using \eqref{PoincarePsi}, we have 
\begin{align*}
&\textup{div} \left( J^{X,\omega^X}[Q^{(+)}_{\ell m_s(n,\ell)}]+ \frac{f'}{1-\mu}\beta |Q^{(+)}_{\ell m_s(n,\ell)}|^2\partial_{r_*}\right)\\
&\geq \frac{f'}{1-\mu}|\slashed{\nabla}_{r_*}Q^{(+)}_{\ell m_s(n,\ell)}+\beta Q^{(+)}_{\ell m_s(n,\ell)}|^2+W|Q^{(+)}_{\ell m_s(n,\ell)}|^2,
\end{align*}
with
\begin{align*}
W(f,\beta)&=\left(-\frac{1}{4}\Box\omega^X-\frac{1}{2}V_{\ell}^{(+)'} f+\frac{1}{2}\mu_r V^{(+)}_{\ell} f  \right)\\
 &+\left( \frac{f''\beta}{1-\mu}+f'\left( \frac{\beta'}{1-\mu}+\frac{n\beta}{r}-\frac{\beta^2}{1-\mu} \right) \right)\\
 &+\frac{(\ell(\ell+n-1)-2n)f}{r^3}\left( 1-\frac{(n+1)M}{r^{n-1}} \right).
\end{align*}
For $n=3$, we make the same choice as for the Regge-Wheeler equation \eqref{RW2}:
\begin{align*}
f&=\left( 1-\frac{(n+1)M}{r^{n-1}} \right)\left( 1-\frac{M}{r^{n-1}}+\frac{4M^2}{5r^{2n-2}} \right)\\
 &=\left( 1-\frac{4M}{r^{2}} \right)\left( 1-\frac{M}{r^{2}}+\frac{4M^2}{5r^{4}} \right),\\
\beta&=\frac{1}{2}.
\end{align*}
For $n=4$ and $\ell\geq 3$, we choose
\begin{align*}
f&=\left( 1-\frac{(n+1)M}{r^{n-1}} \right)\left( 1-\frac{2M}{r^{n-1}}+\frac{6M^2}{5r^{2n-2}} \right)\\
 &=\left( 1-\frac{5M}{r^{3}} \right)\left( 1-\frac{2M}{r^{3}}+\frac{6M^2}{5r^{6}} \right),\\
\beta&=1+\frac{M}{r^3}.
\end{align*}
For $n=4$ and $\ell=2$, we take
\begin{align*}
f&=\left( 1-\frac{(n+1)M}{r^{n-1}} \right)\left( 1-\frac{M}{5r^{n-1}}+\frac{2M^2}{5r^{2n-2}} \right)\\
 &=\left( 1-\frac{5M}{r^{3}} \right)\left( 1-\frac{M}{5r^{3}}+\frac{2M^2}{5r^{6}} \right),\\
\beta&=\frac{1}{2}\left( 2+\frac{M}{5r^3} \right)^3\Bigg/\left( 2+\frac{M}{5r_{P}^3} \right)^3\ \textup{for}\ r\leq r_{P},\\
     &=\frac{1}{2}\left( \frac{1}{2}+\frac{2}{5}\left(\frac{2M}{r^3}\right)^{1/6} \right)\Bigg/\left( \frac{1}{2}+\frac{2}{5}\left(\frac{2M}{r_{P}^3}\right)^{1/6} \right)\ \textup{for}\ r\geq r_{P}.
\end{align*}
Note that $\beta=\frac{1}{2}$  on the photon sphere $r = r_{P}$, so that the divergence theorem still applies. With these choices, $f'>0$ and $W(f,\beta)>0$, and Lemma \ref{Morawetz} can be proved in the same way as in Regge-Wheeler case.

\subsection{Uniform Boundedness and Decay of the Master Quantity}

The estimates for the $Q^{(+)}_{\ell m_s(n,\ell),\alpha\beta}$ in the previous subsection are uniform in the angular mode numbers $\ell$ and $m$, owing to convergence of the potentials $V^{(+)}_{\ell}$ to the limiting potential $V^{(+)}$ \eqref{limitPotential}.  As the relevant energies involve $L^2(S^n)$-terms integrated over the orbit spheres, there is no difficulty in summing the estimates on the angular modes $Q^{(+)}_{\ell m_s(n,\ell), \alpha\beta}$ to obtain estimates on a total object $Q^{(+)}_{\alpha\beta}$, defined as the $L^2(S^n)$-sum
\begin{equation}\label{QPlusDef}
Q^{(+)}_{\alpha\beta} := \sum_{\ell\geq 2}\sum_{m_s(n,\ell)=1}^{d_s(n,\ell)} Q^{(+)}_{\ell m_s(n,\ell), \alpha\beta}.
\end{equation}
Following this reasoning, we have the following estimates for the gauge-invariant master quantity $Q^{(+)}_{\alpha\beta}$:

\begin{theorem}
Let $\delta g$ be a smooth, symmetric two-tensor on a Schwarzschild-Tangherlini spacetime, satisfying the linearized Einstein equation \eqref{linEinstein}.  There exists a gauge-invariant master quantity $Q^{(+)}_{\alpha\beta}$ in the scalar portion $h_1$ of $\delta g$ with harmonics $Q^{(+)}_{\ell m_s(n,\ell), \alpha\beta}$ satisfying the Regge-Wheeler type equations \eqref{RW4}.  Summing estimates for the $Q^{(+)}_{\ell m_s(n,\ell)}$ terms, $Q^{(+)}$ satisfies the uniform boundedness estimate
\begin{equation}
\check{E}^{N}_{Q^{(+)}}(\Sigma_\tau) \lesssim \check{E}^{N}_{Q^{(+)}}(\Sigma_0),
\end{equation}
in all spacetime dimensions.  In six and fewer spacetime dimensions, $Q^{(+)}$ satisfies the uniform decay estimate
\begin{equation}
\check{E}^{N}_{Q^{(+)}}(\Sigma_{\tau}) \lesssim \frac{I_{Q^{(+)}}(\Sigma_0)}{\tau^2},
\end{equation}
where 
\begin{equation} 
I_{Q^{(+)}}(\Sigma_0):=E^2_{Q^{(+)}, \mathcal{L}_{K}Q^{(+)}}(\Sigma_0)+E^N_{Q^{(+)}, \mathcal{L}_{K}Q^{(+)}, \mathcal{L}^2_{K}Q^{(+)}}(\Sigma_0)
\end{equation}
and $\tau \geq 0.$
\end{theorem}

We emphasize that the relevant constants in the comparisons depend only upon the orbit sphere dimension $n$ and the mass $M > 0$.  

%%%%%%%%%%%%%%%%%%%%%%%%%%%%%%%%%%%%%
\section{Proof of Main Theorem}
\begin{theorem}
Let $\delta g$ be a smooth, symmetric two-tensor on a Schwarzschild-Tangherlini spacetime, satisfying the linearized Einstein equation \eqref{linEinstein}.  Performing a spacetime Hodge decomposition of $\delta g$,  each of the portions of $\delta g$ contains gauge-invariant master quantities satisfying decoupled Regge-Wheeler type wave equations.  In particular, the two-tensor portion $h_3 = \hat{h}_{\alpha\beta}$ satisfies the equation \eqref{RW1}, the co-vector portion $h_2$ has quantities $Q^{(-)}_{\alpha\beta}$ \eqref{QMinusDef} and $S_{\alpha\beta}$ \eqref{SDef} satisfying the equation \eqref{RW3}, and the scalar portion $h_1$ has quantities $Q^{(+)}_{\ell m_s(n,\ell), \alpha\beta}$ \eqref{QPlusDef} satisfying the equations \eqref{RW4}.  

As solutions of Regge-Wheeler type equations, the master quantities satisfy the uniform boundedness estimates
\begin{align}
\begin{split}
\check{E}^{N}_{\hat{h}}(\Sigma_\tau) &\lesssim \check{E}^{N}_{\hat{h}}(\Sigma_0),\\
\check{E}^{N}_{Q^{(-)}}(\Sigma_\tau) &\lesssim \check{E}^{N}_{Q^{(-)}}(\Sigma_0),\\
\check{E}^{N}_{S}(\Sigma_\tau) &\lesssim \check{E}^{N}_{S}(\Sigma_0),\\
\check{E}^{N}_{Q^{(+)}_{\ell m_s(n,\ell)}}(\Sigma_\tau) &\lesssim \check{E}^{N}_{Q^{(+)}_{\ell m_s(n,\ell)}}(\Sigma_0),
\end{split}
\end{align}
in all spacetime dimensions.  In six and fewer spacetime dimensions, the master quantities satisfy the uniform decay estimate
\begin{align}
\begin{split}
\check{E}^{N}_{\hat{h}}(\Sigma_{\tau}) &\lesssim \frac{I_{\hat{h}}(\Sigma_0)}{\tau^2},\\
\check{E}^{N}_{Q^{(-)}}(\Sigma_{\tau}) &\lesssim \frac{I_{Q^{(-)}}(\Sigma_0)}{\tau^2},\\
\check{E}^{N}_{S}(\Sigma_{\tau}) &\lesssim \frac{I_{S}(\Sigma_0)}{\tau^2},\\
\check{E}^{N}_{Q^{(+)}_{\ell m_s(n,\ell)}}(\Sigma_{\tau}) &\lesssim \frac{I_{Q^{(+)}_{\ell m_s(n,\ell)}}(\Sigma_0)}{\tau^2},
\end{split}
\end{align}
where 
\begin{equation} 
I_{\Psi}(\Sigma_0):=E^2_{\Psi, \mathcal{L}_{K}\Psi}(\Sigma_0)+E^N_{\Psi, \mathcal{L}_{K}\Psi, \mathcal{L}^2_{K}\Psi}(\Sigma_0)
\end{equation}
and $\tau \geq 0$ for the decay foliation $\Sigma_{\tau}$ of Subsection \ref{decayFoliation}.

Owing to uniformity of the estimates for the $Q^{(+)}_{\ell m_s(n,\ell),\alpha\beta}$ in the angular mode numbers $\ell$ and $m_s(n,\ell)$, we can concisely encode these estimates by considering their $L^2(S^n)$-sum $Q^{(+)}_{\alpha\beta}$ \eqref{QPlusDef}, which satisfies
\begin{align}
\begin{split}
\check{E}^{N}_{Q^{(+)}}(\Sigma_\tau) &\lesssim \check{E}^{N}_{Q^{(+)}}(\Sigma_0),\\
\check{E}^{N}_{Q^{(+)}}(\Sigma_{\tau}) &\lesssim \frac{I_{Q^{(+)}}(\Sigma_0)}{\tau^2}.
\end{split}
\end{align}

\end{theorem}

We remark that further pointwise uniform boundedness and uniform decay estimates can be derived from those above by means of commutation with the angular Killing fields and application of Sobolev estimates on the orbit spheres.
\bibliographystyle{plain}
\bibliography{highDSchwarz}

\end{document}